\documentclass[12pt]{article}%

\usepackage[authoryear]{natbib}

\usepackage{amssymb}
\usepackage{amsfonts}
\usepackage{amsmath}
\usepackage[nohead]{geometry}
\usepackage[singlespacing]{setspace}
\usepackage[bottom]{footmisc}
\usepackage{indentfirst}
\usepackage{endnotes}
\usepackage{graphicx}%
\usepackage{epstopdf}
\usepackage{rotating}
\usepackage{verbatim}
\usepackage{setspace}
\usepackage{multirow}
\usepackage{latexsym}

\usepackage{amsthm}
\usepackage{makeidx}
\usepackage{fancyhdr}
\usepackage{type1cm}

\newcommand{\bb}{\mathrm{\bf b}}
\newcommand{\bff}{\mathrm{\bf f}}
\newcommand{\bx}{\mathrm{\bf x}}

\newcommand{\bA}{\mathrm{\bf A}}

\newcommand{\bu}{\mathrm{\bf u}}
\newcommand{\bv}{\mathrm{\bf v}}
\newcommand{\be}{\mathrm{\bf e}}

\newcommand{\bM}{\mathrm{\bf M}}

\newcommand{\bc}{\mathrm{\bf c}}

\newcommand{\bC}{\mathrm{\bf C}}
\newcommand{\bD}{\mathrm{\bf D}}

\newcommand{\bG}{\mathrm{\bf G}}

\newcommand{\bI}{\mathrm{\bf I}}
\newcommand{\bV}{\mathrm{\bf V}}

\newcommand{\btheta}{\mbox{\boldmath $\theta$}}

\newcommand{\htheta}{\widehat\btheta}

\newcommand{\bbeta}{\mbox{\boldmath $\beta$}}

\newcommand{\stheta}{\mbox{\scriptsize \btheta}}

\newcommand{\bzero}{\mathrm{\bf 0}}

\newcommand{\bmu}{\mbox{\boldmath $\mu$}}

\newcommand{\bSigma}{\mbox{\boldmath $\Sigma$}}

\newcommand{\sbtheta}{\mbox{\scriptsize \boldmath $\theta$}}

\newcommand{\tx}{\widetilde \bx}

\newcommand{\hb}{\widehat \bb}

\newcommand{\hSig}{\widehat\Sig}

\newcommand{\hsig}{\widehat\sigma}

\newcommand{\hbeta}{\widehat\bbeta}

\newcommand{\cov}{\mathrm{cov}}
\newcommand{\Sig}{\mathbf{\Sigma}}

\newcommand{\diag}{\mathrm{diag}}

\newcommand{\bw}{\mbox{\bf w}}

\newcommand{\var}{\mathrm{var}}
\newcommand{\MKT}{\mathrm{MKT}}
\newcommand{\SMB}{\mathrm{SMB}}
\newcommand{\HML}{\mathrm{HML}}
\newcommand{\MOM}{\mathrm{MOM}}
\newcommand{\beq}{\begin{eqnarray*}}
\newcommand{\eeq}{\end{eqnarray*}}

\numberwithin{equation}{section}
\theoremstyle{plain}
\newtheorem{thm}{Theorem}[section]

\newtheorem{lem}{Lemma}[section]
\newtheorem{cor}{Corollary}[section]
\newtheorem{prop}{Proposition}[section]
\newtheorem{assum}{Assumption}[section]
\theoremstyle{definition}

\newtheorem{remark}{Remark}[section]

\usepackage{textcomp}
\usepackage{xr}
\makeatletter
\def\@biblabel#1{\hspace*{-\labelsep}}
\makeatother
\begin{document}

\title{
\bf Power Enhancement  in High Dimensional Cross-Sectional Tests }
\author{Jianqing Fan
\thanks{The authors are grateful to the comments from seminar and conference participants at   UChicago, Princeton, Georgetown,   George Washington, 2014 Econometric Society North America Summer Meeting, UCL workshop on High-dimensional Econometrics Models, The 2014 Annual meeting of Royal Economics Society, The 2014 Asian Meeting of the Econometric Society,  2014 International Conference on Financial Engineering and Risk Management, and 2014 Midwest Econometric Group meeting.  Address: Department of Operations Research and Financial Engineering, Sherrerd Hall, Princeton University, Princeton, NJ 08544, USA. Department of Mathematics, University of Maryland, College Park, MD 20742, USA.  E-mail: \textit{jqfan@princeton.edu},
\textit{yuanliao@umd.edu}, \textit{jiaweiy@princeton.edu}.  The research was partially supported by National Science Foundation grants DMS-1206464 and DMS-1406266, and National Institute of Health grants R01GM100474-01 and R01-GM072611.}$\; ^\dag$, Yuan Liao$^\ddag$ and Jiawei Yao$^*$
\medskip\\{\normalsize $^*$Department of Operations Research and Financial Engineering,  Princeton University}
\medskip\\{\normalsize $^\dag$ Bendheim Center for Finance, Princeton University}
\medskip\\{\normalsize $^\ddag$ Department of Mathematics, University of Maryland}}

\date{}
\maketitle

\sloppy%

\onehalfspacing

%\textbf{Abstract}

\begin{abstract}
  We propose a novel technique to boost the power of testing a high-dimensional vector $H:\btheta=0$ against sparse alternatives where the null hypothesis is violated only by a couple of components. Existing tests  based on  quadratic forms such as the Wald statistic often suffer from low powers  due to the accumulation of errors in estimating high-dimensional parameters. More powerful tests for sparse alternatives such as thresholding and extreme-value tests, on the other hand,  require either  stringent conditions or bootstrap to derive the null distribution and often suffer from size distortions due to the slow convergence. Based on a screening technique, we introduce a ``power enhancement component", which is zero under the null hypothesis with high probability, but diverges quickly under sparse alternatives.   The proposed test statistic combines the power enhancement component with an asymptotically pivotal statistic, and strengthens the power  under sparse alternatives.  The null distribution does not require stringent regularity conditions, and is completely determined by that of the pivotal statistic.  As a byproduct, the power enhancement component also consistently identifies the elements that violate the null hypothesis. As specific applications, the proposed methods are applied to testing   the factor pricing models and validating the cross-sectional independence in panel data models.
  \end{abstract}

\strut

 %Mar ency is further evidenced by our newly designed portfolio based on the d test, which outperforms the SP500 index.


\strut


\onehalfspacing

\section{Introduction}

High-dimensional cross-sectional models have received growing attentions in both theoretical and applied econometrics. These models typically involve a   structural parameter, whose dimension can be either comparable or much larger than the sample size. %This brings new challenges to conventional econometric methods. %, mainly due to the accumulation of a large amount of estimation errors.
This paper addresses testing  a high-dimensional structural parameter:
$$
H_0: \btheta=\bzero,
$$
where $N=\dim(\btheta)$ is allowed to grow faster than the sample size $T.$ We are particularly interested in boosting the power in \textit{sparse} alternatives under which $\btheta$ is approximately a sparse vector. %, whose $L_i$-norm is either bounded or growing slowly with the dimension, for $i=1,2$.
This type of alternative is of particular interest, as the null hypothesis typically represents some economic theory and violations are expected   to be only by some exceptional individuals.

A showcase example is the factor pricing model  in financial economics.
 Let $y_{it}$ be the excess return of the $i$-th asset at time $t$, and $\bff_t=(f_{1t},...,f_{Kt})'$ be the excess returns of $K$ tradable market risk factors.
Then, the excess return has   the following decomposition:
$$
    y_{it}=\theta_i+\bb_i'\bff_t+u_{it}, \quad i=1,...,N, \quad t=1,...,T,
$$
where $\bb_i=(b_{i1},...,b_{iK})'$ is a vector of factor loadings and $u_{it}$ represents the idiosyncratic error.  The key implication from the multi-factor  pricing theory is that the intercept $\theta_i$ should be zero, known as the ``mean-variance efficiency" pricing, for any asset $i$.   An important question is then if such a pricing theory can be validated by empirical data, namely we wish to test the null hypothesis
$
    H_0: \btheta=0,
$
where  $\btheta=(\theta_1,...,\theta_N)'$ is the vector of intercepts for all $N$ financial assets.  As the factor pricing model is derived from theories of  financial economics  \citep{merton, ross}, one would expect that inefficient pricing by the market should only occur to a small fractions of exceptional assets.
Indeed, our empirical study of the constituents in the S\&P 500 index indicates   that there are only a couple of significant nonzero-alpha stocks, corresponding to
a small portion of  mis-priced stocks instead of systematic mis-pricing of the whole market. Therefore, it is  important to construct tests that have high power when $\btheta$ is sparse.

Most of the conventional tests for $H_0: \btheta=0$ are based on a quadratic form:
$$
W=\widehat\btheta'\bV\widehat\btheta.
$$
Here $\htheta$ is an element-wise consistent estimator  of $\btheta$, and  $\bV$ is a high-dimensional positive definite weight matrix, often taken to be the inverse of the asymptotic covariance matrix of $\htheta$ (e.g., the Wald test). %Other widely used statistics are asymptotically equivalent to the quadratic form (e.g., Lagrange multiplier test, likelihood ratio test, Wald test).
After a proper standardization, the standardized $W$ is asymptotically pivotal  under the null hypothesis. In high-dimensional testing problems, however, various difficulties arise when using  a quadratic statistic. First, when $N>T$, estimating $\bV$ is challenging, as the   sample analogue of the covariance matrix is singular. More fundamentally,  tests based on $W$ have low powers under sparse alternatives.
  %, due to the accumulation of a large amount of estimation errors. %To illustrate this problem,  let us temporarily assume that $\htheta$ is component-wise $\sqrt{T}$-consistent, and is normally distributed with mean $\btheta$ and  inverse covariance $\bV=T\Sig^{-1}$ for some positive definite covariance $\Sig.$  In this case, the Wald test statistic $W=T \widehat\btheta' \Sig^{-1}\widehat\btheta$. Under $H_0$, it is  $\chi_N^2$ distributed, with the critical value $\chi^2_{N, q}$ , which is of order $N$, at significance level $q$.  The test has no power at all when$ T \btheta' \Sig^{-1}\btheta = o(N)$ or $\|\btheta\|^2 = o(N/T)$.  This is not unusual for the sparse alternatives in high-dimension-low-sample-size situation we encounter.
The reason is that the quadratic statistic accumulates high-dimensional estimation errors under $H_0$, which results in large critical values that can dominate the signals in the sparse alternatives. A formal proof  of this will be given  in Section 3.3. %Indeed the  $N$ above reflects the noise accumulation in estimating $N$ parameters of $\btheta$.

To overcome the aforementioned drawbacks, this paper introduces a novel technique for high-dimensional cross-sectional testing problems, called the ``power enhancement".   Let $J_1$ be a test statistic that has a correct asymptotic size (e.g., Wald statistic), which may suffer from  low powers under  sparse alternatives.  Let us augment the test by adding a \textit{power enhancement component} $J_0 \geq 0$, which satisfies the following three properties:

\textbf{Power Enhancement Properties}:

\begin{itemize}

\item[(a)] Non-negativity: $J_0\geq0$ almost surely.

 \item [(b)] No-size-distortion: Under $H_0,$ $P(J_0=0|H_0)\rightarrow1$.

\item [(c)]  Power-enhancement: $J_0$  diverges in probability under some specific regions of alternatives $H_a$.

\end{itemize}

Our constructed power enhancement  test  takes  the form
$$
    J= J_0+J_1.
$$
The non-negativity property of $J_0$ ensures that $J$ is at least as powerful as $J_1$. Property (b) guarantees that the asymptotic null distribution of $J$ is   determined by that of $J_1$, and the size distortion due to adding $J_0$ is negligible, and property (c) guarantees significant power improvement under the designated alternatives.  The \textit{power enhancement principle} is thus summarized as follows:  Given a  standard test statistic  with a correct asymptotic size, its power is substantially enhanced with little size distortion; this is achieved  by adding a component $J_0$ that is asymptotically zero under the null, but diverges  and dominates  $J_1$ under some specific regions of alternatives.

An example of such a  $J_0$ is  a \textit{screening statistic}:
$$
J_0= \sqrt{N}\sum_{j\in\widehat S}\widehat\theta_j^2\widehat v_j^{-1}= \sqrt{N}\sum_{j=1}^N\widehat\theta_j^2\widehat v_j^{-1}1\{|\widehat\theta_j|>\widehat v_j^{1/2}\delta_{N,T}\},
$$
where $\widehat S=\{j\leq N: |\widehat\theta_j|>\widehat v_j^{1/2}\delta_{N,T} \}$, and $\widehat v_j$ denotes a data-dependent normalizing factor,   taken as the estimated asymptotic variance of $\widehat\theta_j$.   The threshold $\delta_{N,T}$, depending on $(N,T)$, is a high-criticism threshold, chosen to be slightly larger than the noise level $\max_{j\leq N}|\widehat\theta_j-\theta_j|/\widehat v_j^{1/2}$ so that under $H_0$, $J_0=0$ with probability approaching one.  In addition, we take $J_1$ as a pivotal statistic, e.g., standardized Wald statistic or other quadratic forms such as the sum of the squared  marginal  $t$-statistics \citep{bai1996effect, chen2010two,PY}.   As a byproduct, the screening set $\widehat S$ also consistently identifies indices where the null hypothesis is violated.

One of the major differences of our test from most of the thresholding tests  \citep{fan, hansen05}  is that,  it enhances    the power substantially by adding a screening statistic, which does not introduce extra difficulty in deriving the asymptotic null distribution.  Since $J_0=0$ under $H_0$, it relies on  the   pivotal statistic $J_1$ to determine its null distribution. In contrast,   the existing thresholding tests and extreme value tests often require stringent  conditions to derive their asymptotic null distributions, making them restrictive in  econometric applications, due to slow rates of convergence.   Moreover, the asymptotic null distributions are inaccurate at finite sample. As pointed out by \cite{hansen03}, these statistics are non-pivotal even asymptotically, and require bootstrap methods to simulate the null distributions.

%We study two specific econometric models to illustrate the applications of power enhancement principle: testing ``mean-varianc ciency " in the factor pricing model and testing ``cross-sectional independence" in fixed effect panel data model. The
% \textit{factor pricing model} is one of the most fundamental results in finance.

  %We combine the power enhancement statistic $J_0$ with the standardized Wald test based on $T\htheta'\bW\htheta$, where $\bW$ is proportional to the inverse covariance matrix $\Sig_u^{-1}$ of $(u_{1t},...,u_{Nt})$.  Under the sparse assumption, this high-dimensional covariance is consistently estimated using the thresholding technique.

As for specific applications, this paper studies the tests of the aforementioned {factor pricing model},  and of cross-sectional  independence in mixed effect panel data models:
$$
y_{it}=\alpha+\bx_{it}'\bbeta+\mu_i+u_{it}, \quad i\leq n, t\leq T.
$$
Let
$\rho_{ij}$  denote the correlation between $u_{it}$ and $u_{jt}$, assumed to be time invariant. The   ``cross-sectional independence" test is concerned about the following null hypothesis:
$$
H_0: \rho_{ij}=0,\text{ for all } i\neq j,
$$
that is, under the null hypothesis, the $n\times n$ covariance matrix $\Sig_u$ of $\{u_{it}\}_{i\leq n}$ is diagonal.   In empirical applications,  weak cross-sectional correlations are often present, which results in a sparse covariance $\Sig_u$   with just a few nonzero off-diagonal elements. This results in a sparse vector  $\btheta=(\rho_{12},\rho_{13},...,\rho_{n-1,n})$.  The dimensionality $N=n(n-1)/2$  can be much larger than the number of observations.  Therefore, the power enhancement in sparse alternatives is very important to the testing problem.

There has been a large literature on  high-dimensional  cross-sectional tests. For instance, the literature on testing   the factor pricing model  is found in \cite{GRS},  \cite{MR91},    \cite{BDK}  and \cite{PY}, all  in   quadratic forms.    Moreover, for the mixed   effect panel data model,   most of the existing statistics in the literature are based on the sum of squared residual correlations, which also accumulates many off-diagonal estimation errors in the covariance matrix of $(u_{1t},...,u_{nt})$. The literature includes \cite{BP80},  \cite{PUY}, \cite{BFK}, etc. In addition, our problem is also related  to the test with a restricted parameter space, previously considered by \cite{andrews98}, who improves the power by directing towards the ``relevant" alternatives (also see \cite{hansen03} for a related idea). Recently, \cite{CCK}  proposed a high-dimensional inequality test, and employed an extreme value statistic, whose critical value is determined through applying  the moderate deviation theory on an upper bound of the rejection probability.  In contrast,   the asymptotic distribution of our proposed power enhancement statistic is determined through the pivotal statistic $J_1$, and the power is improved via screening off most of the noises under sparse alternatives.

Two of the referees kindly reminded us a related recent paper   by \cite{GOS}, which   studied estimating and testing about the risk premia in a CAPM model.  While  we also  study a large panel of stock returns as a specific example and  double asymptotics (as $N,T\rightarrow\infty$),  the problems and approaches being considered are very different.  This paper addresses a general problem of enhancing powers under high-dimensional sparse alternatives. %, thanks to the inspirations of these two referees.

The remainder of the paper is organized as follows. Section 2 sets up the preliminaries and highlights  the major differences from existing tests. Section 3  presents the main result of power enhancement test. As applications to specific cases, Section 4 and Section 5 respectively study the factor pricing model and test of cross-sectional independence.   Simulation results are presented in Section 6, along with an empirical application to the stocks in the S$\&$P 500 index in Section 7. Section 8 concludes. All the proofs are given in the appendix.

Throughout the paper,  for a symmetric matrix $\bA$, let $\lambda_{\min}(\bA)$ and $\lambda_{\max}(\bA)$ represent its minimum and maximum eigenvalues. Let    $\|\bA\|_2$ and $\|\bA\|_1$  denote its   operator norm and $l_1$-norm respectively,  defined by $\|\bA\|_2=\lambda_{\max}^{1/2}(\bA' \bA)$ and $\max_{i}\sum_j|\bA_{ij}|$. For a vector $\btheta$, define   $\|\btheta\|=(\sum_{j}\theta_j^2)^{1/2}$ and $\|\btheta\|_{\max}=\max_j|\theta_j|$.  For two deterministic sequences $a_T$ and $b_T$, we write $a_T\ll b_T$ (or equivalently $b_T\gg a_T$) if $a_T=o(b_T)$. Also, $a_T\asymp b_T$ if  there are constants $C_1, C_2>0$ so that $C_1b_T\leq a_T\leq C_2b_T$ for all large $T$.  Finally,  we denote $|S|_0$ as the number of elements in a set $S$.

\section{Power Enhancement in high dimensions}

This section introduces power enhancement techniques and provides   heuristics to justify the techniques.  Their differences with related ideas in the literature are also highlighted.

\subsection{Power enhancement}

Consider a testing problem:
$$
H_0: \btheta=\bzero, \qquad  H_a: \btheta \in \Theta_a,
$$
where $\Theta_a\subset\mathbb{R}^N\backslash\{\bzero\}$ is an alternative set.  A typical example is $\Theta_a = \{\btheta: \btheta \not = 0 \}$. Suppose we observe a stationary process $\bD=\{\bD_t\}_{t=1}^T$ of size $T$. Let $J_1(\bD)$ be a certain  test statistic, and for notational simplicity, we  write $J_1=J_1(\bD)$. Often $J_1$ is constructed such that under $H_0$, it  has a non-degenerate  limiting distribution $F$: As $T, N\rightarrow\infty$,
\begin{equation}\label{e2.1}
J_1|H_0\rightarrow^d F.
\end{equation}
For the significance level $q\in(0,1)$, let $F_q$ be the critical value for $J_1$. Then  the critical region is taken as  $\{\bD: J_1>F_q\}$ and satisfies
\begin{equation}\label{e2.2}
    \limsup_{T, N\rightarrow\infty}P(J_1>F_q|H_0) = q.
\end{equation}
This ensures that $J_1$ has a correct asymptotic size.
In addition, it is often the case that $J_1$  has high power against $H_0$ on a subset $\Theta(J_1)\subset \Theta_a$, namely,
\begin{equation}\label{e2.3}
\liminf_{T, N\rightarrow\infty} \inf_{\stheta \in \Theta(J_1)} P(J_1 > F_q|\btheta)\rightarrow   1.
\end{equation}
Typically, $\Theta(J_1)$ consists of those $\btheta's$, whose $l_2$-norm is relatively large, as $J_1$ is normally an omnibus test (e.g. Wald test).

In a data-rich environment,  econometric models often involve high-dimensional parameters in which $\dim(\btheta)=N$ can grow fast with the sample size $T$.   We are particularly interested in \textit{sparse alternatives} $\Theta_s \subset \Theta_a$ under  which $H_0$ is violated only on a couple of exceptional components of $\btheta$. Specifically, when $\btheta\in\Theta_s$, the number of non-vanishing components is much less than $N$.  As a result, its $l_2$-norm is relatively small.  %It is often the case that  (\ref{e2.2}) is satisfied when $J_1$ is a pivotal statistic  (e.g., Wald test)  after a standardization.  However,
Therefore,
under  sparse alternative $\Theta_s$,
%it is  hard to satisfy (\ref{e2.2}) and (\ref{e2.3}) simultaneously,
the omnibus test $J_1$ typically has lower power,
due to the accumulation of high-dimensional estimation errors.
% However,  it  may have a lower power when  $\btheta\in\Theta_a$ due to the accumulation of large estimation errors.
Detailed explanations are given in Section 3.3 below.

We introduce a \textit{power enhancement principle} for high-dimensional sparse testing, by bringing in a data-dependent component $J_0$ that satisfies the  \textbf{Power Enhancement Properties}   as defined in Section 1.
 The introduced component $J_0$ does not serve as a test statistic on its own, but is added to a   classical statistic $J_1$ that is often pivotal (e.g., Wald-statistic), so  the proposed test statistic is defined by
$$
        J=J_0+J_1.
$$
Our introduced ``power enhancement principle" is explained as follows.
\begin{enumerate}
\item  The critical region of $J$  is defined by
$$
    \{\bD: J>F_q\}.
$$
As $J_0\geq0$,   $P(J> F_q|\btheta)\geq P(J_1> F_q|\btheta)$ for all $\btheta \in \Theta_a$. Hence the power of $J$ is at least as large as that of $J_1$.

\item When $\btheta\in\Theta_s$ is a sparse high-dimensional vector under the alternative, the ``classical" test $J_1$ may have low power as $\|\btheta\|$ is typically relatively small. On the other hand, for $\btheta \in \Theta_s$, $J_0$ stochastically dominates $J_1$. As a result, $P(J> F_q|\btheta)> P(J_1> F_q|\btheta)$ strictly holds, so the power of $J_1$ over the set $\Theta_s$ is enhanced after adding $J_0$. Often $J_0$ diverges fast under sparse alternatives $\Theta_s$, which ensures $P(J>F_q|\btheta)\rightarrow1$ for $\btheta \in \Theta_s$.  In contrast, the classical test only has $P(J_1>F_q|\btheta)<c<1$ for some $c\in(0,1)$ and $\btheta \in \Theta_s$, and when $\|\btheta\|$ is sufficiently small, $P(J_1 > F_q|\btheta)$ is approximately $q$.

\item Under mild conditions, $P(J_0=0|H_0)\rightarrow0$. Hence when (\ref{e2.1}) is satisfied, we have
$$
\limsup_{T, N\rightarrow\infty}P(J>F_q|H_0)=q.
$$
Therefore, adding $J_0$ to $J_1$ does not affect the size of the standard test statistic asymptotically. Both $J$ and $J_1$ have the same limiting distribution under $H_0.$

\end{enumerate}

%Consequently, the power enhancement principle is summarized as:  given a   standard test statistic  whose asymptotic size is correct, its power is substantially enhanced with little size distortion. This is fulfilled by adding a component $J_0$ that is asymptotically zero under the null, but diverges fast and dominates the standard test under the sparse alternative.

It is important to note that the power is enhanced without sacrificing the size asymptotically. In fact the power enhancement principle can be asymptotically fulfilled under a weaker condition  $J_0|H_0\rightarrow^p0.$  However, we construct $J_0$ so that $P(J=0|H_0)\rightarrow1$  to ensure a good finite sample  size.

%so that  the size of the test under $H_0$ is not affected much even in finite sample.

% We shall show that, using a screening technique,  it is relatively easy to construct $J_0$ that satisfies the no-size-distortion condition. We call $J_0$ to be the \textit{power enhancement component}, or \textit{screening statistic}.

\subsection{Construction of power enhancement component}

%The above discussions show that, the standard quadratic test controls the size well, however, has little power against high-dimensional sparse alternatives; some nonstandard   tests based on the thresholding or the extreme values, however, require restrictive conditions to control the size.

We construct a specific power enhancement component $J_0$ that satisfies  (a)-(c) of the power enhancement properties simultaneously, and identify the sparse alternatives in $\Theta_s$.  Such a component can be constructed via \textit{screening} as follows.
Suppose we have a consistent estimator $\htheta$ such that $ {\max_{j\leq N}}|\widehat\theta_j-\theta_j|=o_P(1)$. For some  slowly growing sequence  $\delta_{N,T}\rightarrow\infty$ (as $T, N\rightarrow\infty$), define a screening set:
\begin{equation}\label{e2.4}
\widehat S=\{j: |\widehat\theta_j|>\widehat v_j^{1/2}\delta_{N,T}, j=1,...,N\},
\end{equation}
where $\widehat v_j>0$ is a data-dependent   normalizing constant, often taken as the estimated asymptotic variance of $\widehat\theta_j$.
The sequence  $\delta_{N,T}$, called ``high criticism",   is chosen to be  slightly larger than the maximum-noise-level, satisfying: (recall that $\Theta_a$ denotes the alternative set)
\begin{equation}\label{e2.5} 
\inf_{ \stheta \in \Theta_a \cup \{\bzero\}} P(\max_{j\leq N}|\widehat\theta_j-\theta_j|/\widehat v_j^{1/2}<\delta_{N,T}/2| \btheta)\rightarrow 1
\end{equation}
for $\btheta$ under both null and alternate hypotheses.
The screening statistic $J_0$ is then defined as
$$
J_0= \sqrt{N}\sum_{j\in\widehat S}\widehat\theta_j^2\widehat v_j^{-1}= \sqrt{N}\sum_{j=1}^N\widehat\theta_j^2\widehat v_j^{-1}1\{|\widehat\theta_j|>\widehat v_j^{1/2}\delta_{N,T}\}.
$$
%It uses high criticism $\delta_{N, T}$ to define the screening set $\widehat{S}$ so that b
By (\ref{e2.4}) and  (\ref{e2.5}), under $H_0:\btheta=0$,
$$
P(J_0=0|H_0)\geq P(\widehat S=\emptyset|H_0) = P(\max_{j\leq N}|\widehat\theta_j| / \widehat v_j^{1/2} \leq\delta_{N,T} |H_0)\rightarrow1.
$$
Therefore $J_0$ satisfies the non-negativeness and no-size-distortion properties.

 Let $\{v_j\}_{j\leq N}$ be the population counterpart of $\{\widehat v_j\}_{j\leq N}$. For instance, one can take $v_j $ as the asymptotic variance of $\widehat\theta_j$, and $\widehat v_j$ as its estimator.  To satisfy the power-enhancement property, note that the screening set mimics  \begin{equation}\label{e2.6}
  S(\btheta) =\left\{j : |\theta_j|>2v_j^{1/2}\delta_{N,T}, j=1,...,N \right\},
\end{equation}
and in particular $S(\bzero) = \emptyset$.  We shall show in Theorem \ref{t3.1} below that  $P(\widehat S= S(\btheta)| \btheta)\rightarrow1$, for all $\btheta \in \Theta_a \cup \{\bzero\}$.
Thus, the subvector  $\htheta_{\widehat S}=(\widehat\theta_j: j\in \widehat S)$
behaves like $\btheta_S=(\theta_j: j\in S(\btheta))$,  which can be interpreted as estimated significant signals.
If $S(\btheta)\neq\emptyset$, then by the definition of $\widehat S$ and $\delta_{N,T}\rightarrow\infty$, we have
\begin{equation*}\label{eq2.11}
P(J_0>\sqrt{N} |S(\btheta)\neq \emptyset)\geq P(\sqrt{N} \sum_{j\in \widehat S}\delta_{N,T}^2>\sqrt{N}|S(\btheta)\neq\emptyset )\rightarrow1.
\end{equation*}
Thus, the power of $J_1$ is \textit{enhanced} on the subset
$$
  \Theta_s \equiv\{\btheta\in\mathbb{R}^N:  S(\btheta)\neq \emptyset\}
=\{\btheta\in\mathbb{R}^N: \max_{j\leq N}\frac{|\theta_j|}{v_j^{1/2}}>2\delta_{N,T}\}.$$
As a byproduct, the screening set consistently identifies the elements of $\btheta$ that violate the null hypothesis.

The  introduced  $J_0$    can be combined with any other  test statistic with an accurate asymptotic size. Suppose $J_1$ is a  ``classical" test statistic.    Our power enhancement test is simply
$$J = J_0 + J_1.$$
For instance, suppose we can consistently estimate the  asymptotic inverse  covariance matrix  of $\widehat\btheta$, denoted by $\widehat\var(\htheta)^{-1}$, then   $J_1$ can be chosen as the standardized Wald-statistic:
$$
J_1= \frac{\htheta'\widehat\var(\htheta)^{-1}\htheta-N}{\sqrt{2N}}.
$$
As a result, the asymptotic distribution of $J$ is  $\mathcal{N}(0,1)$ under the null hypothesis.

 In sparse alternatives where $\|\btheta\|$ may not grow fast with $N$ but $\btheta\in \Theta_s$,
the combined test $J_0+J_1$ can be very powerful. In contrast, we will formally show in Theorem \ref{t3.4}  below that  the conventional Wald test $J_1$  can have very low power on its own.  On the other hand, when the alternative is ``dense" in the sense that $\|\btheta\|$ grows fast with  $N$, the conventional test $J_1$ itself is consistent. In this case, $J$ is still as powerful as $J_1$. Therefore,  if we denote $\Theta(J_1)\subset\mathbb{R}^N/\{\bzero\}$ as the set of alternative $\btheta$'s against which the  classical  $J_1$ test has power converging to one,  then the combined $J=J_0+J_1$ test has power converging to one against   $\btheta$ on
 $$
\Theta_s\cup\Theta(J_1).
 $$
 We shall show in Section 3 that    the power is enhanced  uniformly over $\btheta\in\Theta_s\cup\Theta(J_1)$.  %, and  properties like (\ref{e2.2}) and (\ref{e2.3}) are  simultaneously satisfied by $J$.

  \subsection{Comparisons with  thresholding and extreme-value tests}\label{s2.2}

 One of the fundamental differences between our  power enhancement component $J_0$  and existing tests with good power under sparse alternatives is that,   existing test statistics have a non-degenerate distribution under the null, and often require either bootstrap or strong conditions to derive the null distribution.  Such convergences are typically slow and the serious size distortion appears at finite sample.  In contrast, our screening statistic $J_0$ uses ``high criticism" sequence $\delta_{N,T}$  to make  $P(J_0=0|H_0)\rightarrow1$, hence does not serve as a test statistic on its own.    Therefore, the asymptotic null distribution is determined by that of $J_1$, which may not be difficult to derive
 %even if strong dependence   among $\{\widehat\theta_j\}_{j\leq N}$  is present,
 especially when $J_1$ is asymptotically pivotal.  As we shall see in sections below, the required regularity condition is relatively mild,  which makes the power enhancement test applicable  to many econometric problems.

In the high-dimensional testing literature, there are mainly two types  of statistics with good power under sparse alternatives:
  extreme value test and   thresholding  test respectively.
The test based on extreme values studies the maximum deviation from the null hypothesis across the components of $\htheta=(\widehat\theta_1,...,\widehat\theta_N)$, and forms the statistic based on $\max_{j\leq N}|\frac{\widehat\theta_j}{w_j}|^{\delta} $   for some $\delta>0$ and a   weight $w_j$ (e.g., \cite{CLX}, \cite{CCK}).  Such a test statistic typically converges slowly to its asymptotic counterpart.  An alternative test is based on thresholding:   for some $\delta>0$ and pre-determined threshold level $t_T$,
 \begin{equation}\label{e2.7}
 R=\sqrt{T}\sum_{j=1}^N|\frac{\widehat\theta_j}{w_j}|^{\delta}1\{|\widehat\theta_j|>t_{T}w_j\}
 \end{equation}
The accumulation of estimation errors is prevented due to the threshold $1\{|\widehat\theta_j|>t_{T}w_j\}$ (see, e.g.,  \cite{fan} and \cite{ZCX}) for sufficiently large $t_T$.  In a low-dimensional setting, \cite{hansen05} suggested using a threshold to enhance the power in a similar way.

Although (\ref{e2.7}) looks similar to $J_0$, the ideas behind are very different. %It is generally difficult to determine the null distribution of  either the extreme-value test or the thresholding test.
Both extreme value test and thresholding test  require regularity conditions that may be restrictive in econometric applications. For instance, %the null distribution of  the first type of test   depends  on the extreme value distributions, and relies heavily on the   ``near independence" among $(\widehat\theta_1,...,\widehat\theta_N)$. For the thresholding test,
it can be difficult to employ the central limit theorem directly on  (\ref{e2.7}), as it requires the covariance between $\widehat\theta_j$ and $\widehat\theta_{j+k}$ decay fast enough as $k\rightarrow\infty$  \citep{ZCX}. %Similar mixing conditions are also required by the so-called \textit{higher criticism} test (\cite{DH09}).
In cross-sectional testing problems, this essentially requires an explicit ordering among the cross-sectional units which is, however, often unavailable in panel data applications. %Above all,  as an ``power enhancement component", $J_0$ can also be added to $R$ directly to form $J=J_0+R$. As a result,  the power of $R$ can still be enhanced, due to $J_0\geq0$,  without size distortion asymptotically, due to $P(J_0=0|H_0)\rightarrow1$.
In addition, as \eqref{e2.7} involves effectively limited terms of summations due to thresholding, the asymptotic theory does not provide adequate approximations, resulting size-distortion in applications.
For example, when $t_T$ is taken slightly less than $\max_{j \leq N} |\widehat\theta_j|/w_j$, $R$ becomes the extreme statistic.  When $t_T$ is small (e.g. 0), $R$ becomes a traditional test, which is not powerful in detecting sparse alternatives, though it can have good size properties.

%whose critical value is determined through applying  the moderate deviation theory on an upper bound of the rejection probability.

              %   Moreover, as pointed out by \cite{hansen03}, non-quadratic statistics are non-pivotal even asymptotically, and one typically requires employ a bootstrap method to simulate from the null distribution.

        % We would like to emphasize that,  When employing the central limit theorem directly is difficult,  researchers sometimes derive the null distribution for an upper bound of $R$. While strong regularity conditions are still required, controlling the size of an upper bound may be potentially  over-conservative.

        % In particular,   under $H_0$, one has to assume strong  to either apply the central limit theorem on $R$ or derive the null distribution for an upper bound of $R$.

%  Note that the threshold value $t_T$ in the  thresholding test (\r   relatively small so that $P(R\neq 0|H_0)>c>0$ for some constant $c$.

 %and deriving the null distribution of $R$   requires the sequence $\{\widehat\theta_j\}_{j\leq N}$ potentially have an order structure, with fast decaying autocovariances.

%The major difference of our proposal is that, as we show below, our ``thresholding" statistic $J_0$ equals zero under $H_0$ with   probability approaching one. Therefore, it does not introduce extra difficulty in deriving the asymptotic null distributions, and thus we call it ``screening statistic" instead of ``thresholding". Plus, the power is substantially enhanced by combining it with a standard pivotal statistic.

 \section{Asymptotic properties}

\subsection{Main results}

This section  presents the regularity conditions and formally establishes the claimed power enhancement properties. Below we use $P(\cdot|\btheta)$ to denote the probability measure defined from  the sampling distribution with parameter $\btheta$.  Let $\Theta\subset\mathbb{R}^N$ be the parameter space of $\btheta$. When we write $ \inf_{\stheta\in\Theta}{P}(\cdot|\btheta)$, the infimum is taken in the space that covers the union of both null and alternative space.

We begin with a high-level assumption.  In specific applications, they can be verified with primitive conditions.

\begin{assum}\label{a3.1}
As $T, N\rightarrow\infty$, the sequence $\delta_{N,T}\rightarrow\infty$,  and the estimators $\{\widehat\theta_j, \widehat v_j\}_{j\leq N}$ are such that   \\
(i) $\inf_{\stheta\in\Theta}P(\max_{j\leq N}|\widehat\theta_j-\theta_j|/\widehat v_j^{1/2}<\delta_{N,T}/2|\btheta)\rightarrow1;$\\
(ii)   $\inf_{\stheta\in\Theta}P(4/9< \widehat v_j/v_j<16/9,\forall j=1,...,N |\btheta)\rightarrow1.$
\end{assum}

The normalizing constant  $v_j$ is often taken as the asymptotic variance of $\widehat\theta_j$, with $\widehat v_j$ being its consistent estimator.  The constants $4/9$ and $16/9$ in condition (ii) are not optimally chosen, as this condition only requires   $\{\widehat v_j\}_{j\leq N}$  be  \textit{not-too-bad} estimators of their population counterparts.

In many high-dimensional problems with strictly stationary data that satisfy strong mixing conditions,  following from the large-deviation theory,   typically, $\max_{j\leq N}|\widehat\theta_j-\theta_j|/\widehat v_j^{1/2}=O_P(\sqrt{\log N})$. Therefore,  we shall fix
\begin{equation} \label{e3.1}
    \delta_{N,T}=\log(\log T)\sqrt{\log N},
\end{equation}
which is a high criticism that slightly dominates the standardized noise level. We shall provide primitive  conditions for this choice of $\delta_{N,T}$ in the subsequent sections, so that Assumption \ref{a3.1} holds.

Recall that  $\widehat S$ and $S(\btheta)$ are defined by  (\ref{e2.4}) and (\ref{e2.6}) respectively for a given $\btheta\in\Theta$ and its consistent estimator $\htheta$. In particular, $  S(\btheta) =\left\{j : |\theta_j|>2v_j^{1/2}\delta_{N,T}, j=1,...,N \right\}$, so  under $H_0:\btheta=0$, $S(\btheta)=\emptyset$. Note that $\Theta$ denotes the parameter space containing both the null and alternative hypotheses. The following theorem    characterizes the asymptotic behavior of  $
J_0=\sqrt{N}\sum_{j\in\widehat S}\widehat\theta_j^2\widehat v_j^{-1}$ under both the null and alternative hypotheses.

 Define the ``grey area set" as
$$
    \mathcal{G}(\btheta)=\{j: |\theta_j|/v_j^{1/2}\asymp \delta_{N,T}, j=1,...,N \}.
$$

\begin{thm}\label{t3.1}
Let Assumption \ref{a3.1} hold. As $T, N\rightarrow\infty$, we have,
 under $H_0:\btheta=0$,  $P(\widehat S=\emptyset|H_0)\rightarrow1$. Hence
$$
    P(J_0=0|H_0)\rightarrow1 \quad \mbox{and} \quad \inf_{ \{\stheta\in\Theta: S(\stheta)\neq\emptyset\}}P(J_0>\sqrt{N}|\btheta)\rightarrow1.
$$
In addition,
$$
    \inf_{\stheta\in\Theta}{P}(S(\btheta)\subset \widehat S|\btheta)\rightarrow1
    \quad \mbox{and} \quad  \inf_{\stheta\in \Theta}P(\widehat S\setminus S(\btheta) \subset \mathcal{G}(\btheta)|\btheta)\to 1.
$$
\end{thm}
Besides the asymptotic behavior of $J_0$, %under both $H_0$ and sparse alternatives $$\btheta\in  \{\btheta\in\Theta: S(\btheta)\neq\emptyset\},$$
Theorem \ref{t3.1} also provides  a   ``sure screening" property of $\widehat S$.
%selection consistency $P(\widehat S=S(\btheta)|\btheta)\rightarrow1$ for the true parameter $\btheta$ uniformly in the parameter space.
Sometimes we wish to find out the identities  of  the elements in $S(\btheta)  $, which  represent the   components of $\btheta$ that deviate from zero.   Therefore, we are particularly  interested in a type of alternative hypothesis that satisfies the following
\textit{empty grey area} condition.

\begin{assum}[Empty grey area]\label{a3.2}
    For any $\btheta\in\Theta$,     $\mathcal{G}(\btheta)=\emptyset$.
\end{assum}

Theorem~\ref{t3.1} shows that the ``large" $\theta_j$'s can be selected with no missing discoveries and Corollary~\ref{c31} below further asserts that the selection is consistent with no false discoveries either, under both the null and alternative hypotheses.
  %The consistency is not restricted to just the null or the alternative hypotheses.

\begin{cor}\label{c31} Under Assumptions  \ref{a3.1}, \ref{a3.2},  as $T,N\rightarrow\infty$,
$$
  \inf_{\stheta\in\Theta} P(\widehat S=S(\btheta)|\btheta)\rightarrow 1.
$$
   \end{cor}
\begin{proof}
Corollary  \ref{c31} follows immediately from Theorem \ref{t3.1} and Assumption \ref{a3.2}:
$$
\inf_{\stheta\in \Theta}P(\widehat S\setminus S(\btheta)=\emptyset|\btheta)\geq
\inf_{\stheta\in \Theta}P(\widehat S\setminus S(\btheta) \subset \mathcal{G}(\btheta)|\btheta)\to 1.
$$\end{proof}

\begin{remark}

Corollary \ref{c31} and  its required assumptions (Assumptions \ref{a3.1} and \ref{a3.2}) are stated uniformly over $\btheta\in\Theta$.   The empty grey area condition (Assumption \ref{a3.2}) rules out   $\btheta$'s  that have components on the boundary of the screening set.
 Intuitively, when a component $\theta_j$ is on the boundary of the screening, it is hard to decide whether or not to eliminate it from the screening step. Note that the boundary of the screening depends on $(N, T)$, which is similar in spirit to the local alternatives in classical   testing problems, and  is also a common practice for asymptotic analysis of high-dimensional tests (e.g., \cite{Cai10, CCK}).
% This condition is  weak because the chance of falling exactly at the boundary is   low.  %  Note that when $\widehat\theta_j$ is $\sqrt{T}$-consistent, $v_j^{1/2}  = T^{-1/2}v_j$ for some   constant $v_j>0$  that is bounded away from zero.     In this case $S(\btheta)=\{j: \sqrt{T}|\theta_j|>\delta_{N,T}v_j\}$, and $\mathcal{G}(\btheta)=\{j: \sqrt{T}\theta_j\asymp \delta_{N,T}\}$. By the selection consistency,  $\widehat S$ can detect all the $\theta_j$ as long as $|\theta_j|>v_j(\log T)\sqrt{\log N/T}$. We see that the required signal strength does not need to be strong.

\end{remark}

We are now ready to formally show the power enhancement argument. The enhancement is achieved uniformly on the following set:
\begin{equation} \label{e3.2}
\Theta_s = \{\btheta\in\Theta: \max_{j\leq N}\frac{|\theta_j|}{v_j^{1/2}}>2\delta_{N,T}\}.
%=\{\btheta\in\Theta:  S(\btheta)\neq \emptyset\}.
\end{equation}
In particular, if $\widehat\theta_j$ is  $\sqrt{T}$-consistent, and $v_j^{1/2}$ is the asymptotic standard deviation of $\widehat\theta_j$, then $\sigma_j=\sqrt{Tv_j}$ is bounded away from both zero and infinity.  Using  (\ref{e3.1}), we have
$$
    \Theta_s=\{\btheta\in\Theta: \max_{j\leq N}|\theta_j|/\sigma_j>2\log(\log T)\sqrt{\frac{\log N}{T}}\}.
$$
This is a relatively weak condition on the strength of the maximal signal in order to be detected by $J_0$.

 A test is said to have high power uniformly on a set $\Theta^\star \subset \mathbb{R}^N\setminus\{ {\bf 0}\}$ if
 $$
    \inf_{\stheta\in \Theta^\star}P(\text{reject $H_0$ by the test}|\btheta)\rightarrow1.
 $$
 For a given distribution function $F$, let $F_q$ denote its $q$th quantile.

\begin{thm}\label{t3.2}
Let  Assumptions  \ref{a3.1}-\ref{a3.2} hold.  Suppose  there is a test $J_1$ such that
\begin{itemize}\itemsep 0 in
\item [(i)] it has an asymptotic non-degenerate null distribution $F$, and the critical region takes the form  $\{\bD: J_1>F_q\}$ for the significance level $q\in(0, 1)$,
\item [(ii)] it has   high power uniformly on some set $\Theta(J_1)\subset\Theta$,
\item [(iii)] there is $c>0$ so that  $\inf_{\stheta\in\Theta_s}P(c\sqrt{N}+J_1>F_q|\btheta)\rightarrow1, $ as $T,N\rightarrow\infty$.
\end{itemize}
Then the power enhancement   test $J=J_0+J_1$ has the asymptotic null distribution $F$, and has  high power uniformly on the set $\Theta_s \cup \Theta(J_1)$: as $T,N\rightarrow\infty$
%Here $2 \delta_{N,T}\max_{j\leq N}\sigma_{j}$ is of order $(\log T)\sqrt{\frac{\log N}{T}}$. Precisely,
$$
\inf_{\stheta\in\Theta_s\cup\Theta(J_1)}P(J>F_q|\btheta)\rightarrow1.
$$
\end{thm}

The three required conditions for $J_1$ are easy to understand: Conditions (i) and (ii) respectively require the size and power conditions for $J_1$. Condition (iii) requires $J_1$ be dominated by $J_0$   under $\Theta_s$.  This condition is not restrictive  since  $J_1$ is typically standardized (e.g., \cite{donald2003empirical}).

Theorem \ref{t3.2} also shows that $J_1$ and $J$  have the critical regions  $\{\bD: J_1 >F_q\}$ and $\{\bD: J>F_q\}$ respectively, but the power is enhanced from $\Theta(J_1)$  to $\Theta_s\cup\Theta(J_1)$.  In high-dimensional testing problems with a  fast-growing dimension, $\Theta_s\cup\Theta(J_1)$ can be much larger than $\Theta(J_1)$. As a result, the power of $J_1$ can be substantially enhanced by adding $J_0.$

 \subsection{Power enhancement for quadratic tests}

 As an example of $J_1$, we consider   the widely used  quadratic test statistic, which is asymptotically pivotal:
$$
J_{Q}=\frac{T\htheta'\bV\htheta-N(1+\mu_{N,T})}{\xi_{N,T}\sqrt{N}},
$$
where $\mu_{N,T}$ and $\xi_{N,T}$ are deterministic sequences that may depend on $(N,T)$ and $\mu_{N,T} \rightarrow0$, $\xi_{N,T}\rightarrow\xi\in(0,\infty)$. The weight matrix $\bV$ is positive definite, whose eigenvalues are bounded away from both zero and infinity. Here $T\bV$ is often  taken to be the inverse of the asymptotic covariance matrix of $\htheta$.  Other popular choices are $\bV = \diag (\sigma_1^{-2}, \cdots, \sigma_N^{-2})$ with $\sigma_j=\sqrt{Tv_j}$ \citep{bai1996effect,chen2010two,PY} and $\bV = \bI_{N}$, the $N \times N$ identity matrix.
%In this case,  the feasible statistic   replaces $\bV$ by its consistent estimator (with respect to the operator norm). However, when the dimension of $\htheta$ grows fast with the sample size, it is often difficult to estimate $\bV$ unless further structures on the off-diagonal entries are imposed.
We set $J_1=J_Q$, whose power enhancement version is $J=J_0+J_Q.$ For the moment, we shall assume $\bV$ to be known,  and just focus on the  power enhancement properties.   % Note that estimating large covariance matrices itself is a hard problem, and sparse estimation is a promising solution, when it is reasonable to assume either  $\bV$ or its inverse is  sparse   \citep{Bickel08a, ElKaroui08}.  We shall consider this  case when
We will deal with unknown $\bV$ for testing factor pricing problem in the next section.

\begin{assum}\label{a3.3} (i) There is a non-degenerate distribution $F$ so that under $H_0,$ $J_Q\rightarrow^d F$\\
(ii) The critical value $F_q=O(1)$ and the critical region of $J_Q$  is $\{\bD: J_Q>F_q\}$,\\
(iii) $\bV$ is positive definite, and there exist two positive constants $C_1$ and $C_2$ such that $ C_1 \leq \lambda_{\min}(\bV) \leq \lambda_{\max}(\bV) \leq C_2$.\\
(iv)  $ C_3\leq Tv_j  \leq C_4, j=1,..., N$ for  positive constants $C_3$ and $C_4$.
\end{assum}

Analyzing the power properties of $J_Q$ and applying Theorem \ref{t3.2}, we  obtain the following theorem. Recall that  $\delta_{N,T}$ and $\Theta_s$ are defined by (\ref{e3.1}) and (\ref{e3.2}).

\begin{thm}\label{t3.3}
Under Assumptions \ref{a3.1}-\ref{a3.3},   the power enhancement test $J=J_0+J_Q$ satisfies: as $T,N\rightarrow\infty$,
\begin{itemize}
\item [(i)] under the null hypothesis $H_0: \btheta=\bzero$, $J \to^d F$,

\item [(ii)] there is $C>0$ so that $J$ has high power uniformly on the set
$$
  \Theta_s \cup \{\btheta\in\Theta: \|\btheta\|^2>C\delta_{N,T}^2N/T\}\equiv\Theta_s\cup\Theta(J_Q);
$$
that is,
$
\inf_{\stheta\in\Theta_s\cup\Theta(J_Q)}P(J>F_q|\btheta)\rightarrow1
$  for any $q\in(0,1)$.
\end{itemize}

\end{thm}

\subsection{Low power of quadratic statistics under sparse alternatives}

When $J_Q$ is used on its own,  it can suffer from a low power  under sparse alternatives  if  $N$ grows much faster than the sample size, even though  it  has been commonly used in the econometric literature. Mainly, $T\htheta'\bV\htheta$
 aggregates high-dimensional  estimation errors under $H_0$,  which become large with a non-negligible probability and potentially override the sparse signals under the alternative. %Hence the desired property (\ref{e2.3}) is hard to satisfy.
The following result gives this intuition a more precise description.

To simplify our discussion, we shall focus on the Wald-test  with $T\bV$ being the  inverse
of the asymptotic covariance matrix of $\widehat\btheta$, assumed to exist.  Specifically, we assume  the standardized $T\widehat\btheta'\bV\widehat\btheta$ to be asymptotically normal under $H_0$:
\begin{equation}\label{e3.3}
\frac{T\htheta'\bV\htheta-N}{\sqrt{2N}}|H_0\rightarrow^d\mathcal{N}(0,1).
\end{equation}
This is one of the most commonly seen cases in various testing problems.  The diagonal entries of $\frac{1}{T}\bV^{-1}$ are given by $\{v_j\}_{j\leq N}$.
%$\diag(\frac{1}{T}\bV^{-1})=\{v_j\}_{j\leq N}$

 %Recall that the $\|.\|_1$ norm of a symmetric matrix is the maximum of the absolute row (or column) sums. This assumption requires $\bV$ and $\bV^{-1}$ both have bounded row sums. Though the dimension of $\bV$ can grow fast, a bounded row sum is not a a strong assumption, and is typically satisfied for sparse or block-diagonal covariance matrices.

\begin{thm}\label{t3.4}
Suppose that (\ref{e3.3})  holds with  $\|\bV\|_1<C$ and  $\|\bV^{-1}\|_1<C$  for some $C>0$. Under Assumptions \ref{a3.1}- \ref{a3.3},   $T=o(\sqrt{N})$ and $\log N=o(T^{1-\gamma})$ for some $0<\gamma<1$, the quadratic test $J_Q$ has low power at the sparse alternative  $\Theta_b$ given by
$$
    \Theta_b =\{\btheta\in\Theta: \sum_{j=1}^N 1\{\theta_j\neq0\}=o(\sqrt{N}/T), \|\btheta\|_{\max} = O(1) \}.
$$
In other words,   $\forall \btheta\in\Theta_b$, for any significance level $q$,
%\in (0, 0.5)$,
 $$
    \lim_{T, N \to \infty} P(J_Q>z_q|\btheta) = q,
 $$
 where $z_q$ is the $q$th quantile of standard normal distribution.
\end{thm}

 In the  above theorem,   the alternative is a sparse vector.  However, using the quadratic test itself,   the asymptotic  power  of the test is as low as $q$.  This is because    the   signals in the sparse alternative are dominated by the aggregated high-dimensional estimation errors: $T\sum_{i: \theta_i = 0}\widehat\theta_i^2$. In contrast, the  nonzero components of $\btheta$
(fixed constants) are actually detectable by using $J_0$.    The  power enhancement test $J_0+J_Q$ takes this into account, and  has a substantially improved power.

%Below we study in detail a model where  the inverse covariance matrix is sparse.  We also consider a  panel data model with a fixed effect in Section 5, where $\bV$ is actually   known. We realize that estimating $\bV$ in other cases might be difficult. But one can always apply the power enhancement principle by adding $J_0$ to a testing statistic with non-degenerate null distribution, whenever it is available.

\section{Application: Testing Factor Pricing Models}

\subsection{The model}
The multi-factor pricing model,   derived by \cite{ross}  and \cite{merton},  is one of the most fundamental results in finance.  It postulates how financial returns are related to market risks, and has many important practical applications.   Let $y_{it}$ be the excess return of the $i$-th asset at time $t$ and $\bff_t=(f_{1t},...,f_{Kt})'$ be the observable excess returns of $K$ market risk factors.
Then, the excess return has   the following decomposition:
$$
    y_{it}=\theta_i+\bb_i'\bff_t+u_{it}, \quad i=1,...,N, \quad t=1,...,T,
$$
where $\bb_i=(b_{i1},...,b_{iK})'$ is a vector of factor loadings and $u_{it}$ represents the idiosyncratic error.
To make the notation consistent, we pertain to use $\btheta$ to represent the commonly used ``alpha" in the finance literature.

%The imply alpha=0 for traded factors and, if no-arbitrage restrictions hold true, we get mean-variance ncy (Chamberlain-Rothshild (1983)).

 The key implication from the multi-factor  pricing theory for tradable factors is that under no-arbitrage restrictions,  the intercept $\theta_i$ should be zero for any asset $i$ \citep{ross,merton,CR}.  An important question is then testing the null hypothesis
\begin{equation}  \label{e4.1}
    H_0: \btheta=0,
\end{equation}
namely, whether the factor pricing model
is consistent with empirical data,
where  $\btheta=(\theta_1,...,\theta_N)'$ is the vector of intercepts for all $N$ financial assets. One typically picks  five-year monthly data, %  and does not increase the testing period any longer,
because the factor pricing model is technically a one-period model whose factor loadings can be time-varying; see \cite{GOS} on how to model the time-varying effects using firm characteristics and market variables.  As the theory of the factor pricing model applies to all tradable assets,  rather than a handful selected portfolios, the number of assets $N$ should be much larger than $T$.  This ameliorates the selection biases in the construction of testing portfolios. On the other hand, if the theory does not hold,  it is expected that there are only a few significant nonzero components of $\btheta$, corresponding to a small portion of  mis-priced stocks instead of systematic mis-pricing of the whole market. Our empirical studies on the S\&P500 index lend further support to such kinds of sparse alternatives, under which  there are  only a few nonzero components of $\btheta$ compared to $N$.

Most   existing tests to the problem (\ref{e4.1}) are based on the quadratic statistic $W=T\htheta'\bV\htheta$,   where $\htheta$ is the OLS estimator for $\btheta$, and  $\bV$ is some positive definite matrix.  Prominent examples are given by \cite{GRS},   \cite{MR91} and   \cite{BDK}. When $N$ is possibly much larger than $T$,  \cite{PY}  showed that, under  regularity conditions (Assumption \ref{a4.1} below),
$$
    J_1=\frac{a_{f, T}T\htheta'\Sig_u^{-1}\htheta-N}{\sqrt{2N}}\rightarrow^d\mathcal{N}(0,1).
$$
where $a_{f, T}>0$ is a constant that depends only on  factors'  empirical moments, and $\Sig_u$ is the $N\times N $ covariance matrix of $\bu_t=(u_{1t},...,u_{Nt})'$, assumed to be time-invariant.

%Recently \cite{GOS} considered estimating and testing about the risk premia in a CAPM model, and studied both $\alpha$'s and $\beta$'s of an asset pricing model.  While  we share the similarity of studying a large panel of stock returns  and  double asymptotics (as $p,T\rightarrow\infty$),  the problems and approaches being considered are   different. In fact,  \cite{GOS}   tested  pricing restrictions that do not involve a high-dimensional parameter.  Hence they do not have the issue of enhancing the power under sparse alternatives.  %In addition, they do not involve a Wald statistic that depends on a high-dimensional covariance matrix. Therefore,   the sparsity  of  covariance matrices are defined and estimated in different contexts.

Recently,  \cite{GOS} propose a novel approach to modeling and estimating time-varying risk premiums using two-pass least-squares method under asset pricing restrictions.  Their problems and approaches differ substantially from ours, though both papers study similar problems in finance.
As a part of their model validation, they develop test statistics against the asset pricing restrictions and  weak risk factors.  Their test statistics are based on a weighted sum of squared residuals of the cross-sectional regression, which, like all classical test statistics, have power only when there are many violations of the asset pricing restrictions. % For this reason, % they prove only the consistency of the test, rather than consider the power under the contiguous alternatives as we shall do. In particular,
 They do not consider the issue of enhancing the power under sparse alternatives, nor do they involve a Wald statistic that depends on a high-dimensional covariance matrix. In fact, their testing power can be enhanced by using our techniques.

\subsection{Power enhancement component}

%The high dimensionality  brings  two main challenges to the traditional Wald test and its standardized version.The first challenge arises from estimating $\Sig_u^{-1}$ in the presence of   cross-sectional correlations. It is well known that  the sample residual covariance matrix  is  singular when $N>T$.   Even if $N<T$, replacing $\Sig_u^{-1}$  with the inverse sample covariance can still accumulate  substantial  estimation errors when $N^2$ is close to $T$. Specifically, as shown by \cite{FFL08}, under the Frobenius norm, the estimation error of the inverse sample covariance matrix  can  distort  the null distribution of the test statistic in this case.Another  but perhaps more fundamental  challenge arises from the  loss of power due to estimating the high-dimensional but sparse parameter vector $\btheta$, even when $\Sig_u^{-1}$ is known.  This again can be ascribed to the noise accumulation, as explained   and shown  in Section 3.3.

We propose a  new statistic that depends on (i) the power enhancement component $J_0$,  and (ii) a feasible Wald component based on a consistent covariance estimator for $\Sig_u^{-1}$, which controls the size under the null even when $N/T\rightarrow\infty$.

%Because both   $y_{it}$ and $\bff_t$ are  observable,   we employ  OLS    to estimate $\btheta$.

Denote by $\bar\bff=\frac{1}{T}\sum_{t=1}^T\bff_t$ and $\bw=(\frac{1}{T}\sum_{t=1}^T\bff_t\bff_t')^{-1}\bar{\bff}$.  Also define
$$
a_{f, T}=1-\bar\bff'\bw,\quad \text{and } a_f=1-E\bff_t'(E\bff_t\bff_t')^{-1}E\bff_t.
$$
 The OLS estimator of $\btheta$ can be expressed as
\begin{equation}\label{e4.2}
\htheta=(\widehat\theta_1,...,\widehat\theta_N)',\quad
\widehat\theta_j= \frac{1}{Ta_{f,T}}\sum_{t=1}^Ty_{jt}(1-\bff_t'\bw).
\end{equation}
When $\cov(\bff_t)$ is positive definite, under mild regularity conditions (Assumption \ref{a4.1} below),  $a_{f, T}$ consistently estimates $a_f$, and $a_f>0$.  In addition,  without serial correlations, the conditional variance of $\widehat\theta_j$ (given $\{\bff_t\}$)  converges in probability to
$$v_j=\var(u_{jt})/(Ta_f),$$
which can be estimated by $\widehat v_j$ based on the residuals of OLS estimator:
$$
\widehat v_j=\frac{1}{T}\sum_{t=1}^T\widehat u_{jt}^2/(Ta_{f, T}),\quad \text{ where } \widehat u_{jt}=y_{jt}-\widehat\theta_j-\widehat \bb_j'\bff_t.
$$

We show in Proposition \ref{p4.1} below that $\max_{j\leq N}|\widehat\theta_j-\theta_j|/\widehat v_j^{1/2}=O_P(\sqrt{\log N})$.  Therefore,   $\delta_{N,T}=\log(\log T)\sqrt{\log N}$  slightly dominates  the maximum estimation noise. The screening set  and the power enhancement component are defined as
$$
\widehat S=\{j: |\widehat\theta_j|>\widehat v_j^{1/2}\delta_{N,T}, j=1,...,N\},
$$
and
$$
J_0=\sqrt{N}\sum_{j\in\widehat S}\widehat\theta_j^2\widehat v_j^{-1}.
$$

\subsection{Feasible Wald test in high dimensions}
 Assuming no serial correlations among $\{\bu_t\}_{t=1}^T$ and conditional homoskedasticity (Assumption \ref{a4.1} below),  given the observed factors,  the conditional covariance of $\htheta$ is $\Sig_u/(Ta_{f, T})$.  If the covariance matrix $\Sig_u$ of $\bu_t$ were known, the  standardized  Wald test statistic is
\begin{equation} \label{e4.3}
\frac{Ta_{f, T}\htheta'\Sig_u^{-1}\htheta-N}{\sqrt{2N}}.
\end{equation}
Under $H_0:\btheta=0$, it converges in distribution to $\mathcal{N}(0,1)$.  Note that the idiosyncratic errors $(u_{1t},...,u_{Nt})$ are often cross-sectionally correlated, which leads to a non-diagonal  inverse covariance matrix $\Sig_u^{-1}$.
When $N/T\rightarrow\infty$, it is  practically difficult to estimate $\Sig_u^{-1}$, as there are $O(N^2)$ free off-diagonal parameters.

 To consistently estimate $\Sig_u^{-1}$ when $N/T\rightarrow\infty$, without parametrizing the off-diagonal elements, we assume   $\Sig_u=\cov(\bu_t)$  be a  sparse matrix. This assumption is natural for large covariance estimations for factor models, and was previously considered by  \cite{FLM11}. Since the common factors dictate preliminarily the co-movement across the whole panel, a particular asset's idiosyncratic shock is usually correlated significantly only with a few  of other assets. For example, some shocks only exert influences on a particular industry, but are not pervasive for the whole economy \citep{CK93}.

 Following the approach of \cite{Bickel08a}, we can consistently estimate $\Sig_u^{-1}$ via thresholding:  let  $s_{ij}=\frac{1}{T}\sum_{t=1}^T\widehat u_{it}\widehat u_{jt}$. Define the covariance estimator as
$$
(\widehat \Sig_{u})_{ij}=\begin{cases}
s_{ij}, & \text{ if } i=j,\\
h_{ij} (s_{ij}), & \text{ if } i\neq j,
\end{cases}
$$
where $h_{ij} (\cdot)$ is a generalized thresholding function \citep{AF,Rothman09}, with threshold value $\tau_{ij}=C(s_{ii}s_{jj}\frac{\log N}{T})^{1/2}$ for some constant $C>0$, designed to keep only the sample correlation whose magnitude exceeds $C(\frac{\log N}{T})^{1/2}$. The hard-thresholding function, for example, is $h_{ij}(x) = x 1\{|x| > \tau_{ij}\}$, and many other thresholding functions such as soft-thresholding and  SCAD \citep{Fan01}  are specific examples. In general, $h_{ij} (\cdot)$ should satisfy:
\begin{itemize}
\item [(i)] $h_{ij}(z)=0$  if $|z|<\tau_{ij};$
\item [(ii)] $|h_{ij}(z)-z|\leq \tau_{ij};$
\item [(iii)] there are constants $a>0$ and $b>1$ such that $|h_{ij}(z)-z|\leq a\tau_{ij}^2$ if $|z|>b\tau_{ij}$.
\end{itemize}
The thresholded covariance matrix estimator sets most of the off-diagonal estimation noises in $(\frac{1}{T}\sum_{t=1}^T\widehat u_{it}\widehat u_{jt})$ to zero. As studied in \cite{POET},   the constant $C$ in the threshold can be chosen in a data-driven manner so that $\hSig_u$ is strictly positive definite in finite sample even when $N>T$. %Moreover, it can be shown that $\widehat\Sig_u^{-1}$ consistently estimates $\Sig_u^{-1}$ under the operator norm.

With $\hSig_u^{-1}$, we are ready to define  the \textit{feasible standardized Wald statistic}:
\begin{equation}\label{e4.4}
J_{wald}=\frac{Ta_{f, T}\htheta'\hSig_u^{-1}\htheta-N}{\sqrt{2N}},
\end{equation}
whose power can be enhanced under sparse alternatives by:
\begin{equation}\label{e4.5}
J=J_0+J_{wald}.
\end{equation}

\subsection{Does the thresholded covariance estimator affect the size?}

  A natural but technical question to address is that when $\Sig_u$ indeed admits a sparse structure,  is the thresholded estimator $\hSig_u^{-1}$ accurate enough so that the feasible $J_{wald}$ is still asymptotically normal? The answer is affirmative if $N(\log N)^4=o(T^2)$,  and  still we can allow $N/T\rightarrow\infty$.  However, such a simple question is far more technically involved than anticipated, as we now explain.

When $\Sig_u$ is a sparse matrix, under regularity conditions (Assumption \ref{a4.2} below), \cite{FLM11}  showed that
\begin{equation}\label{e4.6}
\|\Sig_u^{-1}-\hSig_u^{-1}\|_2=O_P(\sqrt{\frac{\log N}{T}}).
\end{equation}
By the lower bound derived by \cite{Cai10},  the convergence rate is minimax optimal for the  sparse covariance estimation. On the other hand, when replacing $\Sig_u^{-1}$ in (\ref{e4.3}) by $\hSig_u^{-1}$,  one needs to show that the effect of such a replacement is asymptotically negligible, namely, under $H_0$,
\begin{equation}\label{eq4.7}
    T\htheta'(\Sig_u^{-1}-\hSig_u^{-1})\htheta/\sqrt{N}=o_P(1).
\end{equation}
However, when $\btheta=0$, with careful analysis, $\|\htheta\|^2 = O_P(N/T)$. Using this and (\ref{e4.6}), by the Cauchy-Schwartz inequality, we have
$$
    |T\htheta'(\Sig_u^{-1}-\hSig_u^{-1})\htheta|/\sqrt{N}
    = O_P(\sqrt{\frac{N\log N}{T}}).
$$
We see that  it  requires $N\log N=o(T)$ to converge, which is basically a low-dimensional scenario.

The above simple derivation uses, however,  a Cauchy-Schwartz bound, which is too crude for a large $N$.
%accumulates too many  estimation errors in $\|\htheta-\btheta\|^2$ (with $\btheta = 0$ under $H_0$) under a large $N$.
In fact, $\htheta'(\Sig_u^{-1}-\hSig_u^{-1})\htheta $ is a weighted estimation error of $\Sig_u^{-1}-\hSig_u^{-1}$, where the weights $\htheta$ ``average down" the accumulated estimation errors in estimating elements of $\Sig_u^{-1}$, and  result in an improved rate of convergence. The formalization of this argument requires further regularity conditions and novel technical arguments.  These are formally presented in the following subsection.

\subsection{Regularity conditions}

We are now ready to present the regularity conditions. These conditions are imposed for three technical purposes: (i) Achieving the uniform convergence for $\widehat\btheta-\btheta$ as required in Assumption \ref{a3.1}, (ii) defining the sparsity of $\Sig_u$ so that $\hSig_u^{-1}$ is consistent, and (iii) showing (\ref{eq4.7}), so that the errors from estimating $\Sig_u^{-1}$ do not affect the size of the   test.

Let $\mathcal{F}_{-\infty}^0$ and $\mathcal{F}_{T}^{\infty}$ denote the $\sigma$-algebras generated by $\{\bff_t: -\infty\leq t\leq 0\}$ and  $\{\bff_t:  T\leq t\leq \infty\}$ respectively. In addition, define the $\alpha$-mixing coefficient
\begin{equation*}
\alpha(T)=\sup_{A\in\mathcal{F}_{-\infty}^0, B\in\mathcal{F}_{T}^{\infty}}|P(A)P(B)-P(AB)|.
\end{equation*}

\begin{assum}\label{a4.1}
 (i) $\{\bu_{t}\}_{t\leq T}$ is  i.i.d. $\mathcal{N}(0,\Sig_u)$, where both $\|\Sig_u\|_1$ and $\|\Sig_u^{-1}\|_1$ are bounded; \\
 (ii) $\{\bff_t\}_{t\leq T}$ is strictly stationary, independent of $\{\bu_t\}_{t\leq T}$, and there are $r_1, b_1>0$ so that
$$ \max_{i\leq K}P(|f_{it}|>s)\leq\exp(-(s/b_1)^{r_1}).$$
 (iii) %{\em Strong mixing}:
 There exists   $r_2>0$  such that $r_1^{-1}+r_2^{-1}>0.5$ and $C>0$,  for all $T\in\mathbb{Z}^+$,
$$\alpha(T)\leq \exp(-CT^{r_2}).$$
(iv)  $\cov(\bff_t)$ is positive definite, and  $\max_{i\leq N}\|\bb_i\|<c_1$ for some $c_1>0.$

\end{assum}

Some remarks are in order for the conditions in Assumption \ref{a4.1}.
\begin{remark}
Condition (i),  perhaps somewhat restrictive, substantially facilitates our technical analysis. Here  $\bu_t$ is required to be serially uncorrelated across $t$.  Under this condition, the conditional covariance of $\htheta$, given the factors, has a simple expression $\Sig_u/(Ta_{f, T})$.
%and the  problem is still technically  difficult.
On the other hand, if serial correlations are present  in $\bu_t$, there would be additional autocovariance terms in the covariance matrix, which need to be further estimated via regularizations. Moreover, given that $\Sig_u$ is a sparse matrix, the Gaussianity ensures that most of the idiosyncratic errors are cross-sectionally independent so that $\cov(u_{it}^2, u_{jt}^l)=0$, $l=1,2$, for most of the pairs in $\{(i,j): i\neq j\}$. %In other words, the Gaussianity can indeed by relaxed and replaced with this weaker assumption.

Note that we do allow the factors $\{\bff_t\}_{t\leq T}$ to be weakly correlated across $t$, but satisfy the strong mixing condition Assumption~\ref{a4.1} (iii). %Also note that the conditional homoskedasticity is assumed,  following from the independence assumption in condition (ii).
   %Also, it is  always true that  $E\bff_t'(E\bff_t\bff_t')^{-1}E\bff_t\leq 1$. We rule out the equality to guarantee that the asymptotic variance of $\sqrt{T}\widehat\theta_j$ does not degenerate for each $j.$

 \end{remark}
\begin{remark}
The conditional homoskedasticity $E(\bu_t\bu_t'|\bff_t) = E(\bu_t\bu_t')$ is assumed, granted by condition (ii).  We admit that handling conditional heteroskedasticity, while important in empirical applications, is very technically challenging in our context.  Allowing the high-dimensional covariance matrix $E(\bu_t\bu_t'|\bff_t)$ to be time-varying is possible with suitable \textit{continuum of sparse} conditions on the time domain. In that case, one can  require the sparsity condition to hold  uniformly across $t$ and continuously apply thresholding.  However,  unlike in the traditional case, technically, estimating   the family of large  inverse covariances $\{E(\bu_t\bu_t'|\bff_t)^{-1}: t=1,2,...\}$ uniformly over $t$ is highly challenging. As we shall see in the proof of Proposition \ref{p4.2},  even in the homoskedastic  case, proving the effect of estimating $\Sig_u^{-1}$ to be first-order negligible when $N/T\rightarrow\infty$ requires delicate technical analysis.  %Therefore we focus on the homoskedasticity case.

%\textbf{also   emphasize "fundamental difference" from projected-PC   }

 \end{remark}
%\begin{remark}

%The assumption that $\|\Sig_u\|_1$ and $\|\Sig_u^{-1}\|_1$ are bounded is needed to estimate the weight matrix $\Sig_u^{-1}$ in the Wald-statistic accurately. This condition also  ensures that all the eigenvalues of $\Sig_u$ are bounded away from both zero and infinity, which is  typically assumed in the factor analysis literature (e.g., \cite{bai03}).  The economic intuition is that, unlike the common factors, the information of idiosyncratic component $u_{it}$ does not grow with the increase of the dimension.%, resulting in a ``not-too-dense" covariance matrix $\Sig_u.$

% \end{remark}

  %Sparsity is one of the commonly used assumptions on high-dimensional covariance matrix estimation, which has been extensively studied recently. We refer  to El Karoui (2 , \cite{Bickel08a}, Lam and Fan ( 09), Cai and Liu ( 1), and the references therein.

   To  characterize the sparsity of $\Sig_u$ in our context, define
$$m_N=\max_{i\leq N}\sum_{j=1}^N1\{(\Sig_u)_{ij}\neq0\},\quad D_N=\sum_{i\neq j}1\{(\Sig_u)_{ij}\neq0\}.$$
Here
$m_N$ represents  the maximum number of nonzeros in each row,  and $D_N$  represents the total number of nonzero off-diagonal entries.  Formally, we assume:

\begin{assum} \label{a4.2}  Suppose $ N^{1/2}(\log N)^{\gamma}=o(T)$ for some $\gamma>2$, and  \\
(i) $\min_{(\Sig_u)_{ij}\neq0}|(\Sig_u)_{ij}|\gg \sqrt{(\log N)/T}$;\\
(ii)  at least one of the following cases holds:\\
(a) $D_N=O(N^{1/2})$, and  $m_N^2=O(\frac{T}{N^{1/2}(\log N)^{\gamma}})$\\
(b) $D_N=O(N)$, and $m_N^2=O(1)$.
  \end{assum}

As regulated in  Assumption \ref{a4.2}, we consider two kinds of sparse matrices, and develop our results  for both cases. In the first case (Assumption \ref{a4.2} (ii)(a)), $\Sig_u$ is required to have no more than $O(N^{1/2})$ off-diagonal nonzero entries, but allows a diverging $m_N$. One typical example of this case is that there are  only a small portion (e.g., finitely many) of firms whose individual   shocks ($u_{it}$) are correlated with many other firms'.   %This is the case, for instance, when there is one industry in which firms' individual shocks are correlated with a number of other industries, while most of other firms' individual  shocks are mutually uncorrelated.
 In the   second case (Assumption \ref{a4.2}(ii)(b)), $m_N$ should be bounded, but $\Sig_u$ can have $O(N)$ off-diagonal nonzero entries. This allows block-diagonal  matrices with finite size of blocks or banded matrices with finite number of bands.  This case typically arises when firms' individual shocks are   correlated only within industries but not across industries. %, giving rise to a block diagonal covariance matrix. Note that  $m_N=O(\sqrt{T/\log N})$ is not a restrictive condition and is commonly assumed when estimating sparse covariance matrices \citep{Bickel08a,Rothman09}. In particular,  the optimal rate of convergence is given by  $\|\widehat\Sig_u^{-1}-\Sig_u^{-1}\|_2=O_P(m_N\sqrt{(\log N)/T})$. Hence this condition is needed for $\widehat\Sig_u^{-1}$ to be consistent.

Moreover,  we require $N^{1/2}(\log N)^{\gamma}=o(T)$, which is the price to pay for estimating a large error covariance matrix.  But still we allow $N/T\rightarrow\infty.$   It  is also required that  the minimal signal for the nonzero components be larger than the noise level (Assumption \ref{a4.2} (i)), so that nonzero components are not thresholded off when estimating $\Sig_u$.

%Under the above generalized sparsity assumption, the thresholded covariance estimator $\hSig_u$ is positive definite, and consistently estimates $\Sig_u$ under the operator norm.

  \subsection{Asymptotic  properties}

The following result verifies the uniform convergence required in Assumption \ref{a3.1} over the entire parameter space that contains both the null and alternative hypotheses. Recall that the OLS estimator and its asymptotic standard error are defined in (\ref{e4.2}).

\begin{prop} \label{p4.1}Suppose the distribution of $(\bff_t, \bu_t)$ is independent of $\btheta$.
Under Assumption \ref{a4.1}, for  $\delta_{N,T}=\log(\log T)\sqrt{\log N}$, as $T,N\rightarrow\infty$,
\begin{eqnarray*}
&&\inf_{\stheta\in\Theta}P(\max_{j\leq N}|\widehat\theta_j-\theta_j|/\widehat v_j^{1/2}<\delta_{N,T}/2|\btheta)\rightarrow1.\cr
&&\inf_{\stheta\in\Theta}P(4/9< \widehat v_j/v_j<16/9,\forall j=1,...,N |\btheta)\rightarrow1.
\end{eqnarray*}
\end{prop}

\begin{prop} \label{p4.2}  Under Assumptions  \ref{a3.2},  \ref{a4.1}, \ref{a4.2}, and  $H_0$,
$$
    J_{wald}=\frac{Ta_{f, T}\htheta'\hSig_u^{-1}\htheta-N}{\sqrt{2N}} \rightarrow^d\mathcal{N}(0,1).
$$
\end{prop}

As shown, the effect of replacing $\Sig_u^{-1}$ by its thresholded estimator is asymptotically negligible and the size of the standard Wald statistic can be well controlled.

%The  condition $m_N^4(\log N)^4N=o(T^2)$ requires $\Sig_u$ be sufficiently sparse. Note that for block-diagonal matrices, often we have $m_N=O(1)$. Hence the required condition can reduce to $(\log N)^4  N=o(T^2)$, which is the price to pay for estimating a large residual matrix. % Intuitively, the sample size should be relatively large and the covariance matrix should be sufficiently sparse (corresponding to a slowly-growing $m_N$) to ensure that $\Sig_u^{-1}$  can be estimated accurately.

 %Consequently,   the feasible standardized Wald statistic $J_{wald}=(Ta_{f, T}\htheta'\hSig_u^{-1}\htheta-N)/\sqrt{2N}$ is asymptotically standard normal under $H_0.$ W

 We are now ready to apply  Theorem \ref{t3.3} to obtain the asymptotic properties of $J=J_0+J_{wald}$ as follows.
 For $\delta_{N,T}=\log(\log T)\sqrt{\log N}$, let
\begin{eqnarray*}
&&\Theta_s=\{\btheta\in\Theta:  \max_{j\leq N}\frac{T^{1/2}|\theta_j|}{\var^{1/2}(u_{jt})}>2a_f^{-1/2}\delta_{N,T}\},\cr
&& \Theta(J_{wald})=\{\btheta\in\Theta: \|\btheta\|^2>C\delta_{N,T}^2N/T\}.
\end{eqnarray*}

\begin{thm} \label{t4.1} Suppose the assumptions of Propositions \ref{p4.1} and \ref{p4.2} hold.
\begin{itemize}
\item [(i)] Under the null hypothesis $H_0: \btheta=\bzero$, as $T,N\rightarrow\infty$,
$$
   P(J_0=0|H_0)\rightarrow0, \quad J_{wald} \to^d \mathcal{N}(0,1),
$$
and hence
$$ J=J_0+J_{wald} \to^d \mathcal{N}(0,1).$$
\item [(ii)] There is $C>0$ so that   for any $q\in(0,1)$, as $T,N\rightarrow\infty$,
$$    \inf_{\stheta\in\Theta_s}P(J_{0}>\sqrt{N}|\btheta)\rightarrow1,
\quad \inf_{\stheta\in\Theta(J_{wald})}P(J_{wald}>z_q|\btheta)\rightarrow1,
$$
and hence
$$\inf_{\stheta\in\Theta_s\cup\Theta(J_{wald})}P(J>z_q|\btheta)\rightarrow1,
$$  where $z_q$ denotes the $q$th quantile of the standard normal distribution.

\end{itemize}
\end{thm}

We see that the power is substantially enhanced after $J_0$ is added, as the region where the test has power is enlarged from $\Theta(J_{wald})$ to $\Theta_s\cup\Theta(J_{wald})$.

\section{Application:  Testing Cross-Sectional Independence}

\subsection{The model}

Consider a mixed effect panel data model
$$
y_{it}=\alpha+\bx_{it}'\bbeta+\mu_i+u_{it}, \quad i\leq n, t\leq T,
$$
where the idiosyncratic error $u_{it}$ is assumed to be Gaussian.   The  regressor $\bx_{it}$ could be correlated with the individual random effect $\mu_i$,  but is uncorrelated with $u_{it}$. Let
$\rho_{ij}$  denote the correlation between $u_{it}$ and $u_{jt}$, assumed to be time invariant. The goal is to test the following hypothesis:
$$
H_0: \rho_{ij}=0,\text{ for all } i\neq j,
$$
that is,  whether the cross-sectional dependence is present. It is commonly known that the cross-sectional dependence leads to efficiency loss for OLS, and sometimes it may even cause  inconsistent estimations \citep{andrews05}. Thus testing $H_0$ is an important problem in applied panel data models.
 If we let $N=n(n-1)/2$, and let $\btheta=(\rho_{12},...,\rho_{1n}, \rho_{23},...,\rho_{2n},...,\rho_{n-1, n})'$ be an $N\times 1$ vector stacking  all the  mutual correlations, then the problem is equivalent to testing about a high-dimensional vector $
H_0: \btheta=0.
$
Note that often the cross-sectional dependences are weakly present. Hence  the alternative hypothesis of interest is often    a sparse vector $\btheta$, corresponding to a sparse  covariance matrix $\Sig_u$ of $u_{it}$.

Most of the existing tests  are based on the quadratic statistic
$W=\sum_{i<j}T\widehat\rho_{ij}^2=T\htheta'\htheta,
$
where $\widehat\rho_{ij}$ is the sample correlation between $u_{it}$ and $u_{jt}$,    estimated by the within-OLS \citep{baltagi}, and $\htheta=(\widehat\rho_{12},...,\widehat\rho_{n-1,n})$.  \cite{PUY} and \cite{BFK} studied the rescaled $W$, and showed that after a proper standardization, the rescaled $W$ is asymptotically normal when both $n, T\rightarrow\infty.$ However,  the quadratic test    suffers from a low power if $\Sig_u$ is a sparse matrix under the alternative. In particular, as is shown in Theorem \ref{t3.4}, when $n/T\rightarrow\infty$, the quadratic test cannot detect the sparse alternatives with $\sum_{i<j}1\{\rho_{ij}\neq0\}=o(n/T)$, which is very restrictive.
%not unusual with high-dimensional panel data.
Such a sparse structure is present, for instance, when $\Sig_u$ is a block-diagonal sparse matrix with finitely many blocks and finite block sizes.

\subsection{Power enhancement test}
% The power enhancement method can resolve this problem.
  Following the conventional notation of panel data models, let  $\widetilde y_{it}=y_{it}-\frac{1}{T}\sum_{t=1}^Ty_{it}$, $\widetilde\bx_{it}=\bx_{it}-\frac{1}{T}\sum_{t=1}^T\bx_{it}$, and $\widetilde{u}_{it}=u_{it}-\frac{1}{T}\sum_{t=1}^Tu_{it}$.
%$\bar{y}_i=\frac{1}{T}\sum_{t=1}^Ty_{it}$, $\bar{\bx}_i=\frac{1}{T}\sum_{t=1}^T\bx_{it}$, and $\bar u_i=\frac{1}{T}\sum_{t=1}^Tu_{it}$. Also, let Then
Then
$
\widetilde y_{it}= \widetilde \bx_{it}'\bbeta+\widetilde u_{it}.$
 The within-OLS estimator $\widehat\bbeta$  is obtained by regressing $\widetilde y_{it}$ on $\widetilde \bx_{it}$, which leads to the estimated residual $\widehat u_{it}=\widetilde y_{it}-\widetilde \bx_{it}'\widehat\bbeta$. Then $\rho_{ij}$ is estimated by
$$
\widehat\rho_{ij}=\frac{\widehat\sigma_{ij}}{\widehat\sigma_{ii}^{1/2}\widehat\sigma_{jj}^{1/2}}, \quad \widehat\sigma_{ij}=\frac{1}{T}\sum_{t=1}^T\widehat u_{it}\widehat u_{jt}.
$$
For the within-OLS,  the asymptotic variance of $\widehat\rho_{ij}$ is given by $v_{ij}=(1-\rho_{ij}^2)^2/T$, and is estimated  by $\widehat v_{ij}=(1-\widehat\rho_{ij}^2)^2/T.$  Therefore the screening statistic for the power enhancement test is defined as
\begin{equation}\label{e5.1}
J_0=\sqrt{N}\sum_{(i,j)\in \widehat S}\widehat\rho_{ij}^2\widehat v_{ij}^{-1},\quad  \quad  \widehat S=\{(i,j): |\widehat\rho_{ij}|/\widehat v_{ij}^{1/2}>\delta_{N,T}, i<j\leq n\}.
\end{equation}
where  $\delta_{N,T}=\log(\log T)\sqrt{\log N}$ as before.   The set $\widehat S$ screens off most of the estimation errors.

To control the size, we employ   \cite{BFK}'s  bias-corrected quadratic statistic:
\begin{equation}\label{e5.2}
J_{1}=\sqrt{\frac{1}{n(n-1)}}\sum_{i<j}(T\widehat\rho^2_{ij}-1)-\frac{n}{2(T-1)}.
\end{equation}
Under regularity conditions (Assumptions \ref{a5.1}, \ref{a5.2} below), $J_{1}\rightarrow^d\mathcal{N}(0,1)$ under $H_0$.  Then the power enhancement test can be constructed as
$
J=J_0+J_{1}.
$
The  power is substantially  enhanced to cover the region
\begin{equation}\label{e5.3}
\Theta_s=\{\btheta: \max_{i<j}\frac{\sqrt{T}|\rho_{ij}|}{1-\rho_{ij}^2} >2\log(\log T)\sqrt{\log N}\},
\end{equation}
in addition to the region detectable by $J_{1}$ itself. As a byproduct, it  also identifies pairs $(i,j)$ for $\rho_{ij}\neq0$ through $\widehat S.$ Empirically,  this set helps us understand better  the underlying pattern of cross-sectional correlations.

\subsection{Asymptotic properties}

In order for the power to be uniformly enhanced,  the parameter space of  $\btheta=(\rho_{12},...,\rho_{1n}, \rho_{23},...,\rho_{2n},...,\rho_{n-1, n})'$  is required to be: $\btheta$ is element-wise bounded away from $\pm 1$: there is $\rho_{\max}\in(0,1)$,
$$
\Theta=\{\btheta\in\mathbb{R}^N:  \|\btheta\|_{\max}\leq \rho_{\max}\}.
$$
We denote $E(u_{it}^r|\btheta)$ as the $r$th moment of $u_{it}$ when the correlation vector of the underlying data generating process is $\btheta.$ The following regularity conditions are imposed.

\begin{assum}\label{a5.1}
There are $C_1, C_2>0$, so that \\
(i) $ \sup_{\stheta\in\Theta}\sum_{i\neq j\leq n}|E\tx_{it}'\tx_{jt}E(u_{it}u_{jt}|\btheta)|<C_1n$,\\
(ii) $\sup_{\stheta\in\Theta}\max_{j\leq n}E(u_{jt}^4|\btheta)<C_1$, $\inf_{\stheta\in\Theta}\min_{j\leq n} E(u_{jt}^2|\btheta)>C_2,$

\end{assum}

Condition (i) is needed for the within-OLS to be $\sqrt{nT}$-consistent (see, e.g., \cite{baltagi}). It is usually satisfied by weak cross-sectional correlations (sparse alternatives) among the error terms, or weak dependence among the regressors. We require the second moment of $u_{jt}$  be bounded away from zero uniformly in $j\leq n$ and $\btheta\in\Theta$, so that the cross-sectional   correlations can be estimated stably.

The following conditions are assumed in \cite{BFK}, which are needed for the asymptotic normality of $J_1$ under $H_0$.

\begin{assum}\label{a5.2}
(i) $\{\bu_t\}_{t\leq T}$ are i.i.d. $N(0,\Sig_u)$, $E(\bu_t|\{\bff_t\}_{t\leq T},\btheta)=0$ almost surely.\\
(ii) With probability approaching one,  all the eigenvalues of  $\frac{1}{T}\sum_{t=1}^T\tx_{jt}\tx_{jt}'$  are bounded away from both zero and infinity uniformly in $j\leq n.$
\end{assum}

\begin{prop} \label{p5.1}
Under Assumptions \ref{a5.1} and \ref{a5.2}, for  $\delta_{N,T}=\log(\log T)\sqrt{\log N}$,  and $N=n(n-1)/2$,  as $T,N\rightarrow\infty$,
\begin{eqnarray*}&&
\inf_{\stheta\in\Theta}P(\max_{ij}|\widehat \rho_{ij}-\rho_{ij}|/\widehat v_{ij}^{1/2}< \delta_{N,T}/2|\btheta)\rightarrow1\cr
&&\inf_{\stheta\in\Theta}P(  4/9<\widehat v_{ij}/v_{ij}<16/9,\forall i\neq j    |\btheta)\rightarrow1.
\end{eqnarray*}
\end{prop}

Define
$$\Theta(J_1)=\{\btheta\in\Theta: \sum_{i<j}\rho_{ij}^2\geq  Cn^2\log n /T\}.$$
For $J_1$ defined in (\ref{e5.2}), let
\begin{equation}\label{e5.4}
J=J_0+J_1.
\end{equation}

The main result is presented as follows.

\begin{thm} \label{t5.1} Suppose Assumptions  \ref{a3.2},  \ref{a5.1},  \ref{a5.2}  hold.  As $T,N\rightarrow\infty$,
\begin{itemize}
\item [(i)] under the null hypothesis $H_0: \btheta=\bzero$,
$$
   P(J_0=0|H_0)\rightarrow0, \quad J_{1} \to^d \mathcal{N}(0,1),
$$
and hence
$$ J=J_0+J_{1} \to^d \mathcal{N}(0,1);$$
\item [(ii)] there is $C>0$ in the definition of $\Theta(J_1)$ so that   for any $q\in(0,1)$,
$$
  \inf_{\stheta\in\Theta_s}P(J_{0}>\sqrt{N}|\btheta)\rightarrow1,\quad
  \inf_{\stheta\in\Theta(J_{1})}P(J_1>z_q|\btheta)\rightarrow1,
  $$
and hence
$$\inf_{\stheta\in\Theta_s\cup\Theta(J_{1})}P(J>z_q|\btheta)\rightarrow1.
$$

\end{itemize}
\end{thm}
Therefore the power is enhanced from $\Theta(J_1)$ to $\Theta_s\cup\Theta(J_{1})$ uniformly over sparse alternatives. In particular, the required signal strength of $\Theta_s$ in (\ref{e5.3}) is mild: the maximum cross-sectional correlation is only required to exceed a magnitude of $\log(\log T)\sqrt{(\log N)/T}$.

\section{Monte Carlo Experiments}

In this section, Monte Carlo simulations are employed to examine the finite sample performance of the power enhancement tests.  We  respectively  study the  factor pricing model and the  cross-sectional independence test. % s and investigate the empirical size and power of our new test defined in (\ref{e4.5}). We then focus on fixed effects panel data model to test for cross-sectional dependence via test statistic (\ref{e5.4}).

\subsection{Testing factor pricing models}

%\subsubsection{Experimental Design}
 %Using  the US equity market data, we calibrate a submodel to generate the loadings $\bb_i=(b_{i1}, b_{i2},  b_{i3})'$, the idiosyncratic noises $\bu_t$ and the factors $\bff_t=(f_{1t},f_{2t}, f_{3t})'$. %Besides calibration, this section is composed of another two parts, simulation and results.
 To mimic the real data application, we consider the \cite{FF} three-factor model:
$$
y_{it} = \theta_i+\bb_i'\bff_t+u_{it}.
$$
We simulate $\{\bb_i\}_{i=1}^N$,  $\{\bff_t\}_{t=1}^T$  and $\{\bu_t\}_{t=1}^T$ independently from  $\mathcal{N}_3(\bmu_B,\bSigma_B)$,  $\mathcal{N}_3(\bmu_f, \Sig_f)$, and $\mathcal{N}_N(0,\Sig_u)$ respectively. The parameters are set to be the same as those in the simulations of  \cite{POET}, which are calibrated using  daily returns of S\&P 500's top 100 constituents, for the period  from July $1^{st}$, 2008 to June $29^{th}$ 2012. These parameters are listed in the following table.
 \begin{table}[ht]
 \caption{Means and covariances used to generate $\bb_i$ and $\bff_t$}
 \centering
 \begin{tabular}{c|ccc|c|ccc}

 \hline
 $\bmu_B$       &&  $\bSigma_B$   &  &   $\bmu_f$      &   & $\Sig_f$\\
\hline
0.9833 &0.0921 & -0.0178 & 0.0436  &0.0260  & 3.2351&0.1783&0.7783\\
-0.1233 &-0.0178 & 0.0862 &-0.0211   &0.0211 &0.1783&0.5069&0.0102\\
0.0839&0.0436 & -0.0211 & 0.7624 & -0.0043&0.7783&0.0102&0.6586\\
\hline
\end{tabular}

\label{tab1}
\end{table}

Set $\Sig_u=\diag\{\bA_1,...,\bA_{N/4}\}$ to  be a block-diagonal covariance matrix. Each diagonal block $\bA_j$ is a $4\times 4$ positive definite matrix,   whose correlation matrix has equi-off-diagonal entry $\rho_j$, generated from Uniform$[0,0.5]$. The diagonal entries of $\bA_j$ are obtained via $(\Sig_u)_{ii}=1+\|\bv_i\|^2$, where $\bv_i$ is generated independently from $\mathcal{N}_3(0,0.01\bI_3)$.

%Two types of error covariance matrix $\Sig_u$ are considered: 1) \textbf{Diagonal} $\Sig_u^{(1)}$ is a diagonal matrix with diagonal entries $(\Sig_u)_{ii}=1+\|\bv_i\|^2$, where $\bv_i$ are generated independently from $\mathcal{N}_3(0,0.01\bI_3)$. In this case no cross-sectional correlations are present. 2) \textbf{Block-diagonal}   $\Sig_u^{(2)}=\diag\{\bA_1,...,\bA_{N/5}\}$ is a block-diagonal covariance, where each diagonal block $\bA_j$ is a $5\times 5$ positive definite matrix,   whose correlation matrix has equi-off-diagonal entry $\rho_j$, generated from Uniform$[0,0.5]$. The diagonal entries of $\bA_j$ are generated as those of $\Sig_u^{(1)}$.

We evaluate the power of the test under two specific alternatives (we set $N>T$):
\begin{eqnarray*}
\text{sparse alternative }  H_a^1: && \theta_i=\begin{cases}
0.3, & i\leq  \frac{N}{T}\\
0, & i>\frac{N}{T}
\end{cases}
\cr
\text{weak theta } H_a^2: &&\theta_i=\begin{cases}
\sqrt{\frac{\log N}{T}}, & i\leq  N^{0.4}\\
0, & i>N^{0.4}
\end{cases}.
\end{eqnarray*}
Under $H_a^1$, there are only a few nonzero $\theta$'s with a relative large magnitude. Under $H_{a}^2$,  there are many non-vanishing $\theta$'s, but their magnitudes are all relatively small. In our simulation setup, $\sqrt{\log N/T}$ varies from $0.05$ to $0.10$.
We therefore expect that under $H_a^1$, $P(\widehat S=\emptyset)$ is close to zero because most of the first $N/T$ estimated $\theta$'s should survive from the screening step.  These survived $\hat \theta$'s contribute importantly to the rejection of the null hypothesis.  In contrast, $P(\widehat S=\emptyset)$ should be much larger under $H_a^2$ because the non-vanishing $\theta$'s are too weak  to be detected.

%with a positive probability, $\widehat S=\emptyset$.  Alternative $H_a^3$ has been considered by \cite{PY}.

%\subsubsection{Results}

For each test, we calculate the relative frequency of rejection under $H_0, H_a^1$ and $H_a^2$ based on 2000 replications, with significance level $q=0.05$. We also calculate the relative frequency of $\widehat S$ being empty, which approximates $P(\widehat S=\emptyset)$. We use the soft-thresholding to estimate the error covariance matrix. %and  fix the threshold at $\sqrt{\log N/T}$, as suggested in \cite{POET}.

%When $\Sig_u=\Sig_u^{(1)}$,  we assume the diagonal structure to be known, and compare the performances of the working-independent quadratic test $J_{wi}$ (considered by Pesaran and Yamagat 12) with the power enhanced test $J_0+J_{wi}$. When $\Sig_u=\Sig_u^{(2)}$, we do not assume we know the block-diagonal structure. In this case, four tests are carried out and compared: (1) the working-independent quadratic test $J_{wi}$,  (2) the proposed feasible  Wald test $J_{sw}$ based on estimated error covariance matrix, (3) the power enhanced test $J_0+J_{wi}$, and (4) the power enhanced test $J_0 + J_{sw}$.  For the thresholded covariance matrix, we use the soft-thresholding function and  fix the threshold at $\sqrt{\log N/T}$, as suggested in Fan et al. ( . For each test, we calculate the frequency of rejection under $H_0, H_a^1$ and $H_a^2$ based on 500 replications, with significance level $q=0.05$. We also calculate the frequency of $\widehat S$ being empty, which approximates $P(\widehat S=\emptyset)$. Results are summarized in Tables \ref{tab2}-\ref{tab4}.

 \begin{table}%[ht]
 \caption{Size and power (\%) of tests for simulated Fama-French three-factor model}

\begin{tabular}{cc|ccc|ccc|ccc}
\hline
 &  &  & $H_0$ &  &  & $H_a^1$ &  &  & $H_a^2$ \\
$T$ & $N$ & $J_{wald}$ & PE & $P(\widehat S=\emptyset)$ & $J_{wald}$ & PE & $P(\widehat S=\emptyset)$ & $J_{wald}$ & PE & $P(\widehat S=\emptyset)$\\
\hline
&&&&&&&&&\\

300 & 500 & 5.2 & 5.4 & 99.8 & 48.0 & 97.6 & 2.6 & 69.0 & 76.4 & 64.6 \\
 & 800 & 4.9 & 5.1 & 99.8 & 60.0 & 99.0 & 1.2 & 69.2 & 76.2 & 62.2 \\
 & 1000 & 4.6 & 4.7 & 99.8 & 54.6 & 98.4 & 2.6 & 75.8 & 82.6 & 63.2 \\
 & 1200 & 5.0 & 5.4 & 99.6 & 64.2 & 99.2 & 0.8 & 74.2 & 81.0 & 63.6 \\
&&&&&&&&&\\
500 & 500 & 5.2 & 5.3 & 99.8 & 33.8 & 99.2 & 0.8 & 73.4 & 77.2 & 77.8 \\
 & 800 & 4.8 &5.0 & 99.8 & 67.4 & 100.0 & 0.0 & 72.4 & 76.4 & 75.0 \\
 & 1000 & 5.0 & 5.2 & 99.8 & 65.0 & 100.0 & 0.2 & 76.8 & 80.4 &74.0 \\
 & 1200 & 5.2 & 5.2 & 100.0 & 58.0 & 100.0 & 0.2 & 74.2 & 78.4 & 77.0 \\

%300 & 500 & 0.076 & 0.090 & 0.986 & 0.536 & 0.950 & 0.082 & 0.730 & 0.790 & 0.642 \\
 %& 800 & 0.096 & 0.116 & 0.984 & 0.714 & 0.982 & 0.026 & 0.782 & 0.822 & 0.634 \\
 %& 1000 & 0.088 & 0.100 & 0.982 & 0.648 & 0.976 & 0.032 & 0.810 & 0.852 & 0.630 \\
 %& 1200 & 0.098 & 0.110 & 0.980 & 0.770 & 0.992 & 0.012 & 0.808 & 0.860 & 0.602 \\

%500 & 500 & 0.070 & 0.070 & 0.994 & 0.436 & 0.980 & 0.024 & 0.788 & 0.808 & 0.796 \\
 %& 800 & 0.076 & 0.078 & 0.996 & 0.718 & 1 & 0 & 0.788 & 0.810 & 0.796 \\
 %& 1000 & 0.068 & 0.076 & 0.994 & 0.670 & 0.996 & 0.004 & 0.796 & 0.828 & 0.754 \\
 %& 1200 & 0.056 & 0.058 & 0.998 & 0.624 & 1 & 0 & 0.786 & 0.812 & 0.788 \\
  %& 3000 & 0.090 & 0.090 & 0.998 & 0.96 & 1 & 0 & 86 & 812 & 08 \\
\hline
\end{tabular}

\small
Note: This table reports the frequencies of rejection and $\widehat S=\emptyset$ based on 2000 replications. Here $J_{wald}$ is the standardized Wald test, and PE the power enhanced test. These tests are conducted at 5\% significance level.
\label{tab2}
\end{table}

Table \ref{tab2} presents the empirical size and power of the feasible standardized Wald test $J_{wald}$ as well as those of the power enhanced test $J=J_0+J_{wald}$. First of all, the size of  $J_{wald}$ is close to the significance level.  Under $H_0$, $P(\widehat S=\emptyset)$ is close to one, implying that the  power enhancement component $J_0$ screens off most of the estimation errors. The power enhanced test (PE) has approximately the same size as the original test $J_{Wald}$. Under $H_a^1$, the PE test significantly improves the power of the standardized  Wald-test. In this case, $P(\widehat S=\emptyset)$  is nearly zero because the screening set manages to capture the big thetas.   Under $H_a^2$,  as the non-vanishing thetas are very week, it follows that $\widehat S$ has  a large probability of being empty. But, whenever $\hat{S}$ is non-empty, it contributes to the power of the test. The PE test still slightly improves the power of the quadratic test.

%When $\Sig_u$ is diagonal,  we see that  under the null  $P(\widehat S=\emptyset)$ is close to one, which demonstrates that the screening statistic $J_0$ indeed manages to screen out most of the estimation errors under the null hypothesis, which results in approximately the same size as the original test. On the other hand,  under $H_a^1$, the PE test significantly improves the power of the standardized quadratic test. In this case, $P(\widehat S=\emptyset)$  is nearly zero because the estimated  nonzero thetas still survive after screening. Under $H_a^2$, however, the nonzero thetas are very week, which leads to  a large probability that $\widehat S$ is an empty set.  But the PE test still slightly improves the power of the quadratic test. For the non-diagonal covariance,    similar patterns are observed. We additionally find  that the power of Wald test is slightly larger than the working independent test.% interesting behaviors of $J_{sw}$ when a thresholded covariance matrix is used. Under the alternatives, $J_{sw}$ has  larger powers than $J_{wi}$ does because it takes into account the cross-sectional correlations, and the power is further  significantly improved by the   tests. In addition,  ignoring the cross-sectional correlation structure,  $J_{wi}$ yields more stable test statistics than  the thresholded  test $J_{sw}$, so the  sizes of  $J_{wi}$ can be more accurately determined under the null.

\subsection{Testing cross-sectional independence}
We use the following data generating process in our experiments,
\begin{eqnarray}
  y_{it} &=& \alpha+\beta x_{it}+\mu_i+u_{it}, \quad i\leq n, t\leq T, \\
  x_{it} &=& \xi x_{i,t-1}+\mu_i+\varepsilon_{it}. \label{e6.1}
\end{eqnarray}
Note that we model $\{x_i\}$'s as AR(1) processes, so that $x_{it}$ is possibly correlated with  $\mu_i$, but not with $u_{it}$, as was the case in \cite{Im}.   For each $i$, initialize $x_{it}=0.5$ at $t=1$.
%\subsubsection{Experimental Design}
We specify the parameters as follows:  $\mu_i$ is drawn from $\mathcal{N}(0,0.25)$ for $i=1,...,n$.  The parameters $\alpha$ and $\beta$ are  set $-1$ and 2 respectively. In regression (\ref{e6.1}), $\xi=0.7$ and $\varepsilon_{it}\sim\mathcal{N}(0,1)$.

We generate $\{\bu_t\}_{t=1}^T$ from $\mathcal{N}_n(0,\Sig_u)$. Under the null hypothesis, $\Sig_u$ is set to be a diagonal matrix $\Sig_{u,0}=\diag\{\sigma_1^2,...,\sigma_n^2\}$. Following \cite{BFK}, consider the heteroscedastic  errors
\begin{equation}\label{e6.3}
\sigma_i^2=\sigma^2(1+\kappa\bar x_i)^2
\end{equation}
with $\kappa=0.5$, where $\bar x_i$ is the average of $x_{it}$ across $t$. Here  $\sigma^2$ is scaled to fix the average of $\sigma_i^2$'s at one.%We consider two cases $\kappa=0,0.5$, the former being the homoscedastic case for the error and the latter the heteroscedastic case. In both cases, $\sigma^2$ is scaled to fix the average of $\sigma_i^2$'s be one.

For alternative specifications, we use a spatial model for the errors $u_{it}$.  \cite{BFK} considered a tri-diagonal error covariance matrix in this case. We extend it by allowing for higher order spatial autocorrelations, but require that not all the errors be spatially correlated with their immediate neighbors. Specifically, we start with $\Sig_{u,1}=\diag\{\Sig_1,...,\Sig_{n/4}\}$ as a block-diagonal matrix with $4\times4$ blocks located along the main diagonal. Each $\Sig_i$ is assumed to be $\bI_4$ initially. We then randomly choose $\lfloor n^{0.3} \rfloor$ blocks among them and make them non-diagonal by setting $\Sig_i(m,n)=\rho^{|m-n|}(m,n\leq 4)$, with $\rho=0.2$. To allow for error cross-sectional heteroscedasticity, we set $\Sig_u=\Sig_{u,0}^{1/2}\Sig_{u,1}\Sig_{u,0}^{1/2}$, where $\Sig_{u,0}=\diag\{\sigma_1^2,...,\sigma_n^2\}$ as specified in (\ref{e6.3}).

The Monte Carlo experiments are conducted for different pairs of $(n,T)$ with significance level $q=0.05$ based on 2000 replications. The empirical size, power and the frequency of $\widehat S=\emptyset$ as in (\ref{e5.1}) are recorded.

\begin{table}[h]
\caption{Size and power (\%) of tests for cross-sectional independence}

\begin{tabular}{c|c|llll}
\hline

$H_0$& $T $ \hspace{0.2cm}& \multicolumn{1}{c}{$n=200$} & \multicolumn{1}{c}{$n=400$} & \multicolumn{1}{c}{$n=600$} & \multicolumn{1}{c}{$n=800$} \\

& &  $J_1$/PE /$P(\widehat S=\emptyset)$ &  $J_1$/PE /$P(\widehat S=\emptyset)$ &  $J_1$/PE /$P(\widehat S=\emptyset)$ &  $J_1$/PE /$P(\widehat S=\emptyset)$
\\  \hline       
  & 100 & 4.7/5.5 /99.1  & 4.9/5.3  /99.6  & 5.5/5.7  /99.7  & 4.9/5.2 /99.7 \\
  & 200 & 5.3/5.3 /100.0  & 5.5/5.9 /99.6  & 4.7/5.1  /99.4  & 4.9/5.1 /99.8  \\
  & 300 & 5.2/5.2 /100.0  & 5.2/5.2 /100.0  & 4.6/4.6 /100.0 & 4.9/4.9 /100.0 \\
  & 500 & 4.7/4.7 /100.0  & 5.5/5.5 /100.0  & 5.0/5.0 /100.0 & 5.1/5.1 /100.0 \\
  %& 800 & 5.2/5.2 (100.0)  & 5.1/5.1 (100.0)  & 5.3/5.3 (100.0) & 5.2/5.2 (100.0) \\
\hline
$H_a$& $T $ \hspace{0.2cm}& \multicolumn{1}{c}{$n=200$} & \multicolumn{1}{c}{$n=400$} & \multicolumn{1}{c}{$n=600$} & \multicolumn{1}{c}{$n=800$} \\
\hline
&&&&&\\
  & 100 & 26.4/95.5  /5.0  & 19.8/98.0  /2.3  & 13.5/98.2 /2.0  & 12.2/99.2  /0.9 \\
  & 200 & 54.6/98.8  /1.6  & 40.3/99.6  /0.5  & 24.8/99.6 /0.4  & 21/99.7    /0.3  \\
  & 300 & 78.9/99.25 /1.1  & 65.3/100.0 /0.1  & 41.7/99.9 /0.2  & 37.2/100.0 /0.1 \\
  & 500 & 93.5/99.85 /0.2  & 89.0/100.0 /0.0  & 69.1/100.0 /0.0 & 61.8/100.0 /0.0 \\
 % & 800 & 99.0/100.0 (0.0)  & 98.3/100.0 (0.0)  & 91.2/100.0 (0.0) & 87.7/100.0 (0.0) \\
\hline
\end{tabular}
\small
Note: This table reports the frequencies of rejection by $J_1$ in (\ref{e5.2}) and PE in (\ref{e5.4}) under the null and alternative hypotheses, based on 2000 replications. The frequency of $\widehat S$ being empty is also recorded. These tests are conducted at 5\% significance level.
\label{tab3}
\end{table}
Table \ref{tab3} gives the size and power of the bias-corrected quadratic test $J_1$ in (\ref{e5.2}) and those of the power enhanced test $J$ in (\ref{e5.4}). The sizes of both tests are close to 5\%. In particular, the power enhancement test has little distortion of the original size.

 The bottom panel shows the power of the two tests under the   alternative specification. The PE test demonstrates almost full power under all combinations of $(n,T)$. In contrast, the quadratic test $J_1$ as in (\ref{e5.2}) only gains power when $T$ gets large. As $n$ increases, the proportion of nonzero off-diagonal elements in   $\Sig_u$ gradually decreases. It becomes harder for $J_1$ to effectively detect those deviations from the null hypothesis. This explains the low power exhibited by the quadratic test when facing a high sparsity level. %The results are consistent with the asymptotic theory given in Theorem \ref{t5.1}.

\section{Empirical Study}

As an empirical application, we consider a test of  \cite{carhart1997persistence}'s 
four-factor model on the S$\&$P 500 index. Our empirical findings show that  there are only a few significant nonzero ``alpha" components, corresponding to
 a small portion of  mis-priced stocks instead of systematic mis-pricing of the whole market.

%that the sparse alternatives are often present, that is, there are often a small portion of stocks that have significant nonzero components of alphas. % tThis provides an empirical evidence of sparse alternatives.
% at the end of each month between 1984.12-2012.12. We examine the month-end test results based on five-year window of monthly data to see whether the model holds true during each given period.
%We apply the proposed thresholded Wald test and the PE  to the securities in the S$\&$P 500 index, by  employing the Fama-French three-factor (FF-3) model to conduct our test. One of our empirical findings is that mark iciency is primarily caused by a small portion of stocks with positive thetas, instead of a large portion of slightly mispriced assets. This provides empirical evidence of sparse alternatives. In addition,  the market  iciency is further evidenced by our newly created portfolio based on the PE test, which outperforms the S\&P500 index.

%Composed of large cap U.S. stocks, the S$\&$P 500 index has diverse constituency and is therefore a good market representation.
We collect monthly excess returns on all the S$\&$P 500 constituents from the CRSP database for the period January 1980 to December 2012. %, during which a total of 1170 stocks have entered the index for our study.
We test whether $\btheta=0$ (all alpha's  are zero)  in the factor-pricing model on a rolling window basis: for each month, we evaluate our test statistics $J_{wald}$ and $J$ (as in (\ref{e4.4}) and (\ref{e4.5}) respectively) using the preceding 60 months' returns $(T=60)$. %The procedure resembles that of Lewellen and Nagel (2006), though who are only concerned with the (conditional) estimates of alphas and betas.
The panel at each testing month consists of stocks without missing observations in the past five years, which yields a balanced panel with the cross-sectional dimension larger than the time-series dimension $(N>T)$. In this manner we not only capture the up-to-date information in the market, but also mitigate the impact of time-varying factor loadings and sampling biases. In particular, for testing months $\tau=1984.12,...,2012.12$, we run the regressions
\begin{equation}\label{eq7.1}
r_{it}^{\tau}-r_{ft}^{\tau}=\theta_i^{\tau}+\beta_{i,\MKT}^{\tau}(\MKT_t^{\tau}-r_{ft}
^{\tau})+\beta_{i,\SMB}^{\tau} \SMB_t^{\tau}+\beta_{i,\HML}^{\tau} \HML_t^{\tau}
+\beta_{i,\MOM}^{\tau} \MOM_t^{\tau}
+u_{it}^{\tau},
\end{equation}
for $i=1,...,N_{\tau}$ and $t=\tau-59,...,\tau$, where $r_{it}$ represents the return for stock $i$ at month $t$, $r_{ft}$ the risk free rate, and $\MKT, \SMB$, $\HML$ and $\MOM$ constitute market, size, value and momentum factors. The time series of factors are downloaded from Kenneth French's website. %Our null hypothesis $\theta_i^{\tau}=0$ for all $i$ implies that the market is mean-variance e ent for the testing window ending at $\tau$.
To make the notation consistent, we   use
$ \theta_i^{\tau}$ to represent the ``alpha" of stock $i$.
\begin{comment}
\begin{table}[h]%[ht]
\caption{Variable descriptive statistics for the four-factor model}
\centering

\begin{tabular}{lccccc}
\hline
 Variables  & Mean & Std dev. & Median & Min & Max   \\
\hline

 $N_{\tau}$  & 617.70  &  26.31  & 621 & 574 & 665 \\
$|\widehat S|_0$  & 5.49  &  5.48  & 4 & 0 & 37 \\
$\overline{|\widehat \theta|}_i^{\tau} (\%) $  & 0.9973 &  0.1630  & 0.9322 & 0.7899 & 1.3897 \\
$\overline{|\widehat \theta|}_{i\in \widehat S}^{\tau} (\%)$  & 4.3003  & 0.9274 & 4.1056 & 1.7303 & 8.1299\\
$p$-value of $J_{wi}$  & 0.2844  &  0.2998  & 0.1811 & 0 & 0.9946 \\
$p$-value of $J_{sw}$  & 0.1861  &  0.2947  & 0.0150 & 0 & 0.9926 \\
$p$-value of PE  & 0.1256 &  0.2602  & 0.0003 & 0 & 0.9836 \\

\hline
\end{tabular}

 \label{tab5}
\end{table}
\end{comment}

\begin{table}[h]%[ht]
\caption{Summary of descriptive statistics and testing results}
\centering

\begin{tabular}{lccccc}
\hline
 Variables  & Mean & Std dev. & Median & Min & Max   \\
\hline

 $N_{\tau}$  & 617.70  &  26.31  & 621 & 574 & 665 \\
$|\widehat S|_0$  & 5.20  &  3.50  & 5 & 0 & 20 \\
$\overline{|\widehat \theta|}_i^{\tau} (\%) $  & 0.9767 &  0.1519  & 0.9308 & 0.7835 & 1.3816 \\
$\overline{|\widehat \theta|}_{i\in \widehat S}^{\tau} (\%)$  & 4.5569  & 1.4305 & 4.1549 & 1.7839 & 10.8393\\
$p$-value of $J_{wald}$  & 0.2351  &  0.2907  & 0.0853 & 0 & 0.9992 \\
$p$-value of $J$ (PE)  & 0.1148 &  0.2164  & 0.0050 & 0 & 0.9982 \\

\hline
\end{tabular}

 \label{tab5}
\end{table}

Table \ref{tab5} summarizes descriptive statistics for different components and estimates in the model. On average, 618  stocks (which is more than 500 because we are recording stocks that have \textit{ever} become the constituents of the index)  %in some months new stocks enter to replace some others)
enter the panel of the regression during each five-year estimation window. Of those, merely 5.2 stocks are selected by the screening set $\widehat S$, which directly implies the presence of sparse alternatives.  The threshold $\delta_{N,T}=\sqrt{(\log N)} \log(\log T)$ varies as the panel size $N$ changes at the end of each month, and is about 3.5 on average, a high-criticism thresholding. The selected stocks have much larger alphas ($\theta$) than other stocks do. In addition, $64.05\%$ of all the estimated alphas are positive, whereas  $87.33\%$ of the selected alphas in $\widehat S$ are positive.   This indicates that the power enhancement component in our test is primarily contributed by stocks with extra returns. %, instead of a large portion of stocks with small thetas, demonstrating the sparse alternatives. %an average (or aggregated) sense (that is, more likely due to a  large $\min_{\theta_j>0}\theta_j$ other than a large $\|\btheta\|$).
We also notice that the $p$-values of the Wald test $J_{wald}$ are generally smaller than those of  the power enhanced test $J$. %Due to an enhanced power,   exhibits even lower $p$-values than the previous two.

\begin{figure}[!htbp]
\caption{Dynamics of p-values and percents of selected stocks}
\begin{center}
\includegraphics[scale=0.75]{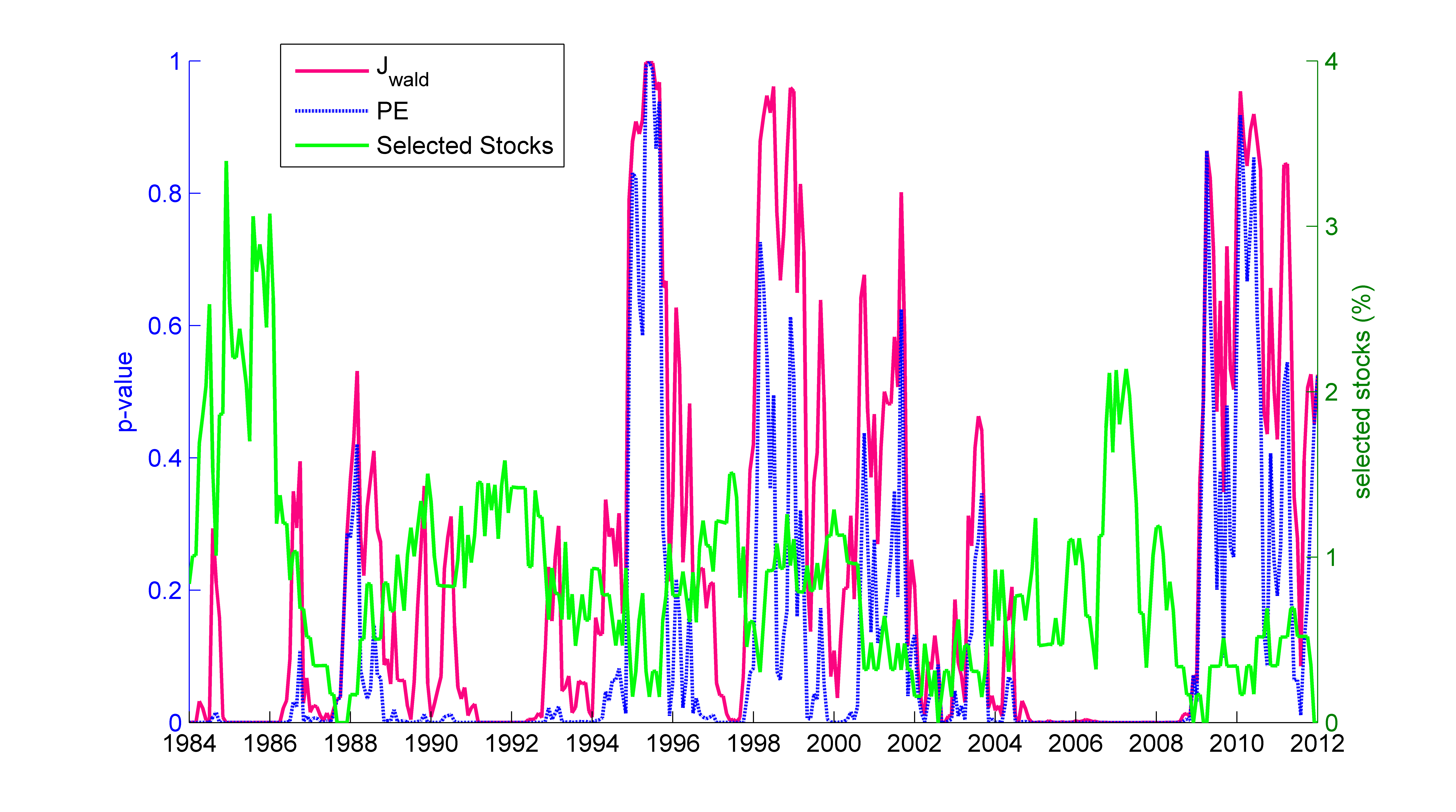}
\end{center}
\label{fig1}
\end{figure}

\begin{figure}[!htbp]
\caption{Histograms of $p$-values for $J_{wald}$ and PE.   }
\begin{center}
\includegraphics[scale=0.8]{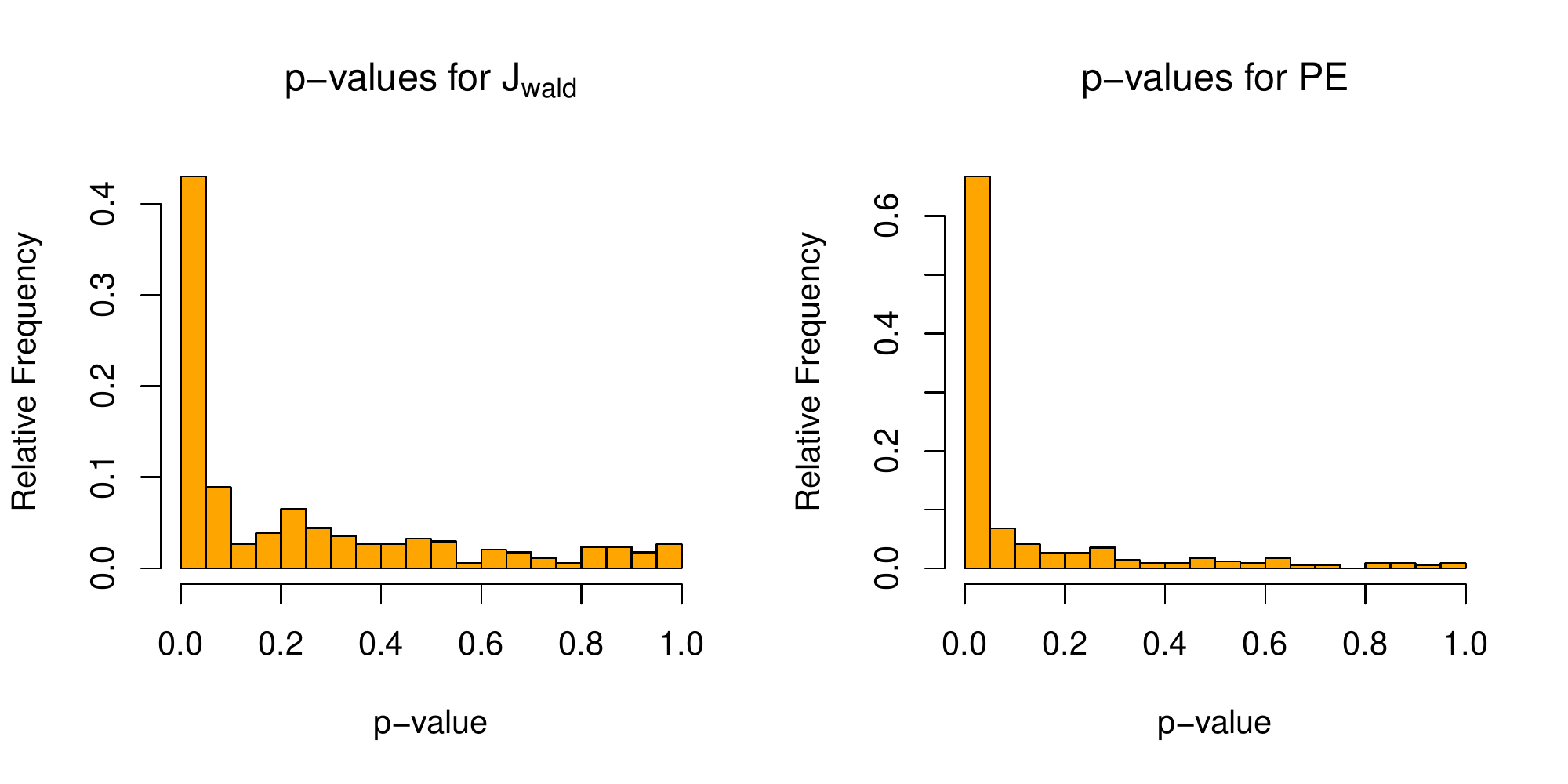}

\end{center}
\label{fig2}
\end{figure}

Similar to \cite{PY}, we plot the running $p$-values of $J_{wald}$ and the PE test from December 1984 to December 2012. We also add the dynamics of the percentage of selected stocks ($|\widehat S|_0/N$) to the plot, as shown in Figure \ref{fig1}. There is a strong negative correlation between the stock selection percentage and the $p$-values of these tests. In other words, the months at which the null hypothesis is rejected typically correspond to a few stocks with alphas exceeding the threshold. Such evidence of sparse alternatives has originally motivated our study.
We also observe that the $p$-values of the PE test lie beneath those of $J_{wald}$ test as a result of enhanced power, and hence it captures several important market disruptions ignored by the latter (e.g. collapse of Japanese bubble in 1990).
Indeed, the null hypothesis of $\btheta=0$ is rejected by the PE  test at $5\%$ level for almost all months during financial crisis, including major financial crisis such as Black Wednesday in 1992, Asian financial crisis in 1997, the financial crisis in 2008, which is also partially detected by $J_{wald}$ tests.  The histograms of the $p$-values of the two test statistics are displayed in Figure \ref{fig2}. By inspection, we see that of 43.03\% and 66.07\% of the study months, $J_{wald}$ and the PE test  reject the null hypothesis respectively. %The fact that the factor pricing model fails to hold is not surprising, but in what way?
Again, the test results indicate the existence of sparse alternatives when faced with high crosse-sectional dimension. %, and are again consistent with the theoretical properties of power enhancement test.

\section{Concluding remarks}

  We consider testing a high-dimensional vector $H:\btheta=0$ against sparse alternatives where the null hypothesis is violated only by a few components. Existing tests  based on  quadratic forms such as the Wald statistic often suffer from low powers  due to the accumulation of errors in estimating high-dimensional parameters.

   We introduce a ``power enhancement component" based on a screening technique, which is zero under the null, but diverges quickly under sparse alternatives.   The proposed test statistic combines the power enhancement component with a classical statistic that is often  asymptotically pivotal, and strengthens the power  under sparse alternatives. On the other hand, the null distribution does not require stringent regularity conditions, and is completely determined by that of the pivotal statistic.  As a byproduct, the screening  statistic also consistently identifies the elements that violate the null hypothesis. As specific applications, the proposed methods are applied to testing   the mean-variance efficiency in factor pricing models and testing the cross-sectional independence in panel data models.

Our empirical study on the S\&P500 index shows that  there are only a few significant nonzero components, corresponding to   a small portion of  mis-priced stocks instead of systematic mis-pricing of the whole market.
This provides empirical evidence of sparse alternatives.

%\newpage

\vspace{0.5 in}

\begin{center}
{\Large \bf APPENDIX}
\end{center}

\appendix

\small

Throughout the proofs, let $C$ denote a generic constant, which may differ at difference places.

\section{Proofs for Section 3}

\subsection{Proof of Theorem \ref{t3.1}}

\proof

 Define  events
$$
A_1=\left\{\max_{j\leq N}|\widehat\theta_j-\theta_j|/\widehat v_j^{1/2}<\delta_{N,T}            \right\}, \quad
A_2=\left\{   \frac{4}{9}< \widehat v_j/v_j<\frac{16}{9},\forall j=1,...,N \right\}.
$$    For any $j\in S(\btheta)$, by the definition of $S(\btheta)$, $|\theta_j|>2 \delta_{N,T}v_j^{1/2}$.    Under $A_1\cap A_2$,
\begin{eqnarray*}
  \frac{|\widehat{\theta}_j|}{{\widehat v_j^{1/2}}}
  &\geq&  \frac{|{\theta}_j|-|\widehat\theta_j-\theta_j|}{{\widehat v_j^{1/2}}}\geq \frac{3|\theta_j|}{4v_j^{1/2}}-\frac{\delta_{N,T}}{2}>\delta_{N,T}.
\end{eqnarray*}
 This  implies that $j\in\widehat S$, hence  $S(\btheta)\subset \widehat S$.  If $j\in \widehat S$, by similar arguments,  we have $\frac{|\theta_j|}{{v_j^{1/2}}} > \delta_{N,T}/3$ on $A_1\cap A_2$. Hence $\widehat S\setminus S(\btheta)\subset\{j: \delta_{N,T}/3<\frac{|\theta_j|}{{v_j^{1/2}}} <2 \delta_{N,T}\}\subset \mathcal{G}(\btheta)$.
In fact, we have proved that $S(\btheta)\subset \widehat S$ and $\widehat S\setminus S(\btheta)\subset \mathcal{G}(\btheta)$ on the event $A_1\cap A_2$  uniformly for $\btheta\in\Theta$.   This yields
$$ \inf_{\stheta\in\Theta}{P}(S(\btheta)\subset \widehat S|\btheta)\rightarrow1, \quad \mbox{and} \quad \inf_{\stheta\in \Theta}P(\widehat S\setminus S(\btheta) \subset \mathcal{G}(\btheta))\to 1.$$
Moreover,  it is readily seen that, under $H_0: \btheta=0,$ by Assumption \ref{a3.1},
$$
    P(J_0=0|H_0)\geq P(\widehat S=\emptyset|H_0)=P(\max_{j\leq N} \{|\widehat\theta_j|/\widehat v_j^{1/2}\}
    <\delta_{N,T}|H_0)\rightarrow1.
$$
In addition, $\inf_{\stheta\in\Theta}P(J_0>\sqrt{N} |S(\btheta)\neq \emptyset)$ is bounded from below by
$$
 \inf_{\stheta\in\Theta}P(\sqrt{N} \sum_{j\in \widehat S}\delta_{N,T}^2>\sqrt{N}|S(\btheta)\neq \emptyset )\geq \inf_{\stheta\in\Theta}P(\sqrt{N}\delta_{N,T}^2>\sqrt{N}|S(\btheta)\neq \emptyset)-o(1)\rightarrow1.
$$
Note that the last convergence holds uniformly in $\btheta\in\Theta$ because $\delta_{N,T}\rightarrow\infty$.  This completes the proof.

\subsection{Proof of Theorem \ref{t3.2}}

\proof It follows immediately from $P(J_0=0|H_0)\rightarrow1$ that $J\rightarrow^d F$, and hence the critical region $\{\bD: J>F_q\}$ has size $q$. Moreover, by the power condition of $J_1$ and $J_0\geq0,$
$$
\inf_{\stheta\in\Theta(J_1)}P(J>F_q|\btheta)\geq \inf_{\stheta\in\Theta(J_1)}P(J_1>F_q|\btheta)\rightarrow1.
$$
This together with the fact
$$
    \inf_{\stheta\in\Theta_s\cup\Theta(J_1)}P(J>F_q|\btheta)\geq \min\{\inf_{\stheta\in\Theta_s}P(J>F_q|\btheta), \inf_{\stheta\in\Theta(J_1)}P(J>F_q|\btheta)\},
$$
establish the theorem, if we show $\inf_{\stheta\in\Theta_s}P(J>F_q|\btheta)\rightarrow1.$

By the definition of $\widehat S$ and $J_0$, we have $\{J_0<\sqrt{N}\delta_{N,T}^2 \} = \{ \widehat S=\emptyset\}$.  Since
$\inf_{\theta\in\Theta}P(S(\btheta)\subset \widehat S|\btheta)\rightarrow1$  and $\Theta_s=\{\btheta\in\Theta: S(\btheta)\neq \emptyset\}$, we have
\begin{eqnarray*}
\sup_{\Theta_s}P(J_0<\sqrt{N}\delta_{N,T}^2 |\btheta)
 %& \leq & \sup_{\Theta_s}P(\sqrt{N} \sum_{j\in \widehat S}\delta_{N,T}^2<\sqrt{N}\delta_{N,T}^2| \btheta)\\
&= & \sup_{\Theta_s}P(\widehat S=\emptyset|\btheta)\cr
&\leq & \sup_{\{\stheta\in\Theta: S(\stheta)\neq \emptyset\}} P(\widehat S=\emptyset, S(\btheta)\subset \widehat S|\btheta)+o(1),
\end{eqnarray*}
which converges to zero, since the first term is zero.
This implies $\inf_{\Theta_s}P(J_0\geq \sqrt{N}\delta_{N,T}^2|\btheta)\rightarrow1.$ Then by condition (ii),  as $\delta_{N,T}\rightarrow\infty$,
$$
\inf_{\stheta\in\Theta_s}P(J>F_q|\btheta)\geq \inf_{\stheta\in\Theta_s}P(\sqrt{N}\delta_{N,T}^2+J_1>F_q|\btheta)\geq\inf_{\theta\in\Theta_s}P(c\sqrt{N}+J_1>F_q|\btheta)\rightarrow1.
$$
This completes the proof.

\subsection{Proof of Theorem \ref{t3.3}}

%=$$\var(\btheta)^{-1}=T\bV, \frac{1}{T}\bV^{-1}=\var(\btheta)$$

\proof    It suffices to verify conditions (i)-(iii) in  Theorem \ref{t3.2} for $J_1=J_Q$. Condition (i) follows from Assumption \ref{a3.3}. Condition (iii) is fulfilled   for $c>2/\xi$, since
$$
\inf_{\stheta\in\Theta_s}P(c\sqrt{N}+J_Q>F_q|\btheta)\geq \inf_{\stheta\in\Theta_s}P(c\sqrt{N}-\frac{N(1+\mu_{N,T})}
{\xi_{N,T}\sqrt{N}}>F_q|\btheta) \to 1,
$$
by using $F_q=O(1)$, $\xi_{N,T}\rightarrow\xi$, and $ \mu_{N,T}\rightarrow0$.
We now verify condition (ii) for the $\Theta(J_Q)$ defined in the theorem.  Let $\bD=\diag(v_1,...,v_N)$.  Then $\|\bD\|_2<C_3/T$ by Assumption~\ref{a3.3}(iv). On the event
$A=\{\|(\htheta-\btheta)'\bD^{-1/2}\|^2<\delta_{N,T}^2N/4\}$,
we have
\begin{eqnarray*}
|(\htheta-\btheta)'\bV\btheta|&\leq&\|(\htheta-\btheta)'\bD^{-1/2}\|
\|\bD^{1/2}\bV\btheta\| \\
&\leq &  \delta_{N,T}\sqrt{N}\|\bD\|_2^{1/2}\|\bV\|_2^{1/2}(\btheta'\bV\btheta)^{1/2}/2 \\
&\leq &  \delta_{N,T}\sqrt{N} (C_3/T)^{1/2}\|\bV\|_2^{1/2}(\btheta'\bV\btheta)^{1/2}/2.
\end{eqnarray*}
For $\|\btheta\|^2>C\delta_{N,T}^2N/T$ with $C=4C_3\|\bV\|_2/\lambda_{\min}(\bV)$, we can bound further that
$$
|(\htheta-\btheta)'\bV\btheta| \leq \btheta'\bV\btheta/4.
$$
Hence,
$
\htheta'\bV\htheta\geq \btheta'\bV\btheta-2(\htheta-\btheta)'\bV\btheta\geq \btheta'\bV\btheta/2.
$
Therefore,
\begin{eqnarray*}
\sup_{\stheta\in\Theta(J_Q)}P(J_Q\leq F_q|\btheta) & \leq & \sup_{\Theta(J_Q)}P(\frac{T\btheta'\bV\btheta/2-2N}{\xi\sqrt{N}}\leq F_q|\btheta)+\sup_{\Theta(J_Q)}P(A^c|\btheta)\cr
&\leq&\sup_{\Theta(J_Q)}P({T\lambda_{\min}(\bV)\|\btheta\|^2}<2F_q{\xi\sqrt{N}}+4N|\btheta)+o(1)\cr
&\leq&\sup_{\Theta(J_Q)}P(  \lambda_{\min}(\bV)C\delta_{N,T}^2N  <5N|\btheta)+o(1),
\end{eqnarray*}
which converges to zero since $\delta_{N,T}^2\rightarrow\infty.$ This implies
$
\inf_{\Theta(J_Q)}P(J_Q>F_q|\btheta)\rightarrow1
$
and finishes the proof.

\subsection{Proof of Theorem \ref{t3.4}}

\begin{proof}

Through this proof, $C$ is a generic constant, which can vary from one line to another. Without loss of generality, under the alternative, write
$$
    \btheta'=(\btheta_1',\btheta_2')=(\bzero',\btheta_2'), \qquad \htheta'=(\htheta_1',\htheta_2'),
$$
where $\dim(\btheta_1)=N-r_N$ and $\dim(\btheta_2)=r_N$. Corresponding to  $(\btheta_1',\btheta_2')$, we partition $\bV^{-1}$ and $\bV$ into:
$$
 \bV^{-1}=\begin{pmatrix}
                                         \bM_1&\bbeta'\\
                                         \bbeta&\bM_2                                           \end{pmatrix}
\quad \mbox{and} \quad
\bV=\begin{pmatrix}
                                        \bM_1^{-1}+\bA&\bG'\\
                                        \bG&\bC                                         \end{pmatrix},
$$
where $\bM_1$ and  $\bA$ are $(N-r_N)\times (N-r_N)$; $\bbeta$ and $\bG$ are $ r_N\times (N-r_N)$; $\bM_2$ and $\bC$ are $r_N\times r_N.$

  By the matrix inversion formula,
$$
    \bA=\bM_1^{-1}\bbeta'(\bM_2-
    \bbeta\bM_1^{-1}\bbeta')^{-1}\bbeta\bM_1^{-1}.
$$

Let $\Delta = T\htheta'\bV\htheta-T\htheta_1'\bM_1^{-1}\htheta_1$.
Note that
$$
\Delta
=T\htheta_1'\bA\htheta_1+2T\htheta_2'\bG\htheta_1+T\htheta_2'\bC\htheta_2.
$$
We first look at $T\htheta_1'\bA\htheta_1$. Let $\lambda_{N, T} = T\lambda_{\max}((\bM_2-\bbeta\bM_1^{-1}\bbeta')^{-1})$ and
$\bD_1=\diag(\frac{1}{T}\bM_1)$.  Note that the diagonal entries of $\frac{1}{T}\bV^{-1}$ are given by $\diag(\frac{1}{T}\bV^{-1})=\{v_j\}_{j\leq N}$. Therefore $\bD_1$ is a diagonal matrix with entries $\{v_j\}_{j\leq N-r_N}$, and $\max_j v_j=O(T^{-1})$.

  Since $\bbeta$ is $r_N\times (N-r_N)$, using the expression of $\bA$,  we have
\begin{eqnarray*}
T\htheta_1'\bA\htheta_1 %&=&\htheta_1'\bM_1^{-1}\bbeta'(\bM_2-\bbeta\bM_1^{-1}\bbeta')^{-1}\bbeta\bM_1^{-1}\htheta_1\cr
&\leq& \lambda_{N, T}  \|\bbeta\bM_1^{-1} \htheta_1\|^2\cr
&\leq& \lambda_{N, T}  r_N \|\bM_1^{-1}(\htheta_1-\btheta_1)\|^2_{\max}(\max_{i\leq r_N}\sum_{j\leq N-r}|\beta_{ij}|)^2\cr
&\leq& \lambda_{N, T} r_N \|\bM_1^{-1}\bD_1^{1/2}\|_1^2\|\bD_1^{-1/2}(\htheta_1-\btheta_1)\|^2_{\max}   \|\bV^{-1}\|_1^2,
\end{eqnarray*}
where we used $\btheta_1 =0$ in the second inequality and the fact that $\max_{i\leq r_N}\sum_{j\leq N-r}|\beta_{ij}|\leq \|\bV^{-1}\|_1$.
Note that   $\|\bV\|_1=O(1)=\|\bV^{-1}\|_1$. Hence,
$$
    \|\bM_1^{-1}\bD_1^{1/2}\|_1^2=O(T^{-1}),   \quad \mbox{and} \quad
    \lambda_{N, T} =O(T).
$$
Thus, there is $C>0$, with probability approaching one,
$$
T\htheta_1'\bA\htheta_1\leq Cr_N\|\bD_1^{-1/2}(\htheta_1-\btheta_1)\|^2_{\max}\leq Cr_N\delta_{N,T}^2.
$$

Note that the uniform convergence in Assumption~\ref{a3.1} and boundness of $\|\btheta\|_{\max}$ imply that $P(\|\htheta\|_{\max}\leq C) \to 1$ for a sufficient large constant $C$.
For  $\bG=(g_{ij})$, note that $\max_{i\leq r}\sum_{j=1}^{N-r}|g_{ij}|\leq \|\bV\|_1$. Hence, by using $\btheta_1 =0$ again, with probability approaching one,
\begin{eqnarray*}
|T\htheta_2'\bG\htheta_1|&=&T|\htheta_2'\bG\bD_1^{1/2}\bD_1^{-1/2}
    (\htheta_1-\btheta_1)|  \\
    & \leq & T\|\htheta_2\|_{\max} \|\bD_1^{-1/2}(\htheta_1-\btheta_1)\|_{\max} \sum_{i=1}^{r_N}\sum_{j=1}^{N-r}|g_{ij}|\sqrt{ v_j}\cr
&\leq&C  r_N \delta_{N,T}\sqrt{T}.
\end{eqnarray*}
Moreover, $T\htheta_2'\bC\htheta_2\leq T \|\htheta_2\|^2\|\bC\|_2=O_P(r_NT)$.  Combining all the results above, it yields that for any $\btheta\in\Theta_b$,
$$
    \Delta=O_P(r_N\delta_{N,T}^2+r_NT).
$$

We denote $\var(\widehat\btheta), \var(\widehat\btheta_1), \var(\widehat\btheta_2)$ to be the asymptotic covariance matrix of $\widehat\btheta, \widehat\btheta_1$ and $\widehat\btheta_2.$ Then $\frac{1}{T}\bV^{-1}=\var(\widehat\btheta)$ and $\frac{1}{T}\bM_1=\var(\widehat\btheta_1)$. It then follows from (\ref{e3.3}) that
 $$Z\equiv\frac{T\htheta_1'\bM_1^{-1}\htheta_1-(N-r_N)}{\sqrt{2(N-r_N)}}\rightarrow^d\mathcal{N}(0,1).$$

For any $0<\epsilon<F_q$, define the event $A=\{|\Delta-r_N|<\sqrt{2N}\epsilon\}.$   Hence,    suppressing the dependence of $\btheta$,
 \begin{eqnarray*}
P(J_Q>F_q)&=&P(\frac{T\htheta_1'\bM_1^{-1}\htheta_1+\Delta-N}{\sqrt{2N}}>F_q)\\
&=& P(Z\sqrt{\frac{N-r_N}{N}}+\frac{\Delta-r_N}{\sqrt{2N}}>F_q) \cr
&\leq& P(Z\sqrt{\frac{N-r_N}{N}}+\epsilon>F_q)+P(A^c),
\end{eqnarray*}
which is further bounded by $1-\Phi(F_q-\epsilon)+P(A^c)+o(1)$.
Since $1-\Phi(F_q)=q$, for small enough $\epsilon$, $1-\Phi(F_q-\epsilon)= q + O(\epsilon)$. By letting $\varepsilon \to 0$ slower than $O(Tr_N/\sqrt{N})$, we have  $P(A^c)=o(1)$,  and
$\limsup_{N\to \infty, T\to \infty} P(J_Q>F_q) \leq q$.  On the other hand,
$P(J_Q > F_q ) \geq P(J_1 > F_q)$, which converges to $q$.  This proves the result.

 \end{proof}

\section{Proofs for Section 4}

\begin{lem}\label{lb.0} When $\cov(\bff_t)$ is positive definite, $E\bff_t'(E\bff_t\bff_t')^{-1}E\bff_t<1$.
\end{lem}
\proof  If $E\bff_t=0$, then $E\bff_t'(E\bff_t\bff_t')^{-1}E\bff_t<1$. If $E\bff_t\neq 0$,  because $\cov(\bff_t)$ is positive definite,  let $\bc=(E\bff_t\bff_t')^{-1}E\bff_t$, then $\bc'(E\bff_t\bff_t'-E\bff_tE\bff_t')\bc>0$.  Hence $\bc'E\bff_tE\bff_t'\bc<\bc'E\bff_t\bff_t'\bc$ implies
$
E\bff_t'(E\bff_t\bff_t')^{-1}E\bff_t>(E\bff_t'(E\bff_t\bff_t')^{-1}E\bff_t)^2.
$
This implies $E\bff_t'(E\bff_t\bff_t')^{-1}E\bff_t<1.$

\subsection{Proof of Proposition \ref{p4.1}}

Recall that $v_j=\var(u_{jt})/(T-TE\bff_t'(E\bff_t\bff_t')^{-1}E\bff_t),$ and  $\widehat v_j=\frac{1}{T}\sum_{t=1}^T\widehat u_{jt}^2/(Ta_{f, T})$.
Write $\sigma_{ij}=(\Sig_u)_{ij}$,  $\hsig_{ij}=\frac{1}{T}\sum_{t=1}^T\widehat u_{it}\widehat u_{jt}$,  $\sigma_j^2=Tv_j,$ and $ \widehat\sigma_j^2=T\widehat v_j$.

Simple calculations yield
\begin{equation*}
\widehat\theta_i=\theta_i+a_{f, T}^{-1}\frac{1}{T}\sum_{t=1}^Tu_{it}(1-\bff_t'\bw).
\end{equation*}

%\textbf{Proof of Proposition \ref{p4.1}}

We first prove the second statement.  Note that there is $\sigma_{\min}>0$ (independent of $\btheta$) so that $\min_j\sigma_j>\sigma_{\min}$.
By Lemma \ref{la.1add},  there is $C>0$, $\inf_{\Theta}P(\max_{j\leq N}|\hsig_j-\sigma_j|<C\sqrt{\frac{\log N}{T}}|\btheta)\rightarrow1$.
On the event $\{\max_{j\leq N}|\hsig_j-\sigma_j|<C\sqrt{\frac{\log N}{T}}\}$,
 $$
 \max_{j\leq N}\left|\frac{\widehat v_j^{1/2}}{v_j^{1/2}}-1\right|\leq\max_{j\leq N}\frac{|\hsig_j-\sigma_j|}{\sigma_j}\leq\frac{C\sqrt{\log N}}{\sigma_{\min}\sqrt{T}}.
 $$
This proves the second statement.  We can now use this to prove  the first statement.

 Note that $v_j$   is independent of $\btheta$, so there is $C_1$ (independent of $\btheta$) so that $\max_{j\leq N}v_j^{-1/2}<C_1\sqrt{T}$.
On the event $\{\max_{j\leq N}v_j^{1/2}/\widehat v_j^{1/2}<2\}\cap \{\max_{j\leq N}|\widehat\theta_j-\theta_j|<C\sqrt{\frac{\log N}{T}}\}$,
$$
\max_{j\leq N}\frac{|\widehat\theta_j-\theta_j|}{\widehat v_j^{1/2}}\leq C\sqrt{\frac{\log N}{T}}2\max_{j}v_j^{-1/2}\leq  2CC_1\sqrt{\log N}<\delta_{N,T}.$$
The constants $C, C_1$ appeared are independent of $\btheta$, and Lemma \ref{la.1add} holds uniformly in $\btheta$. Hence the desired result also holds uniformly in $\btheta.$

\subsection{Proof of Proposition \ref{p4.2}}

By Theorem 1 of \cite{PY}(Theorem 1),
$$
    (Ta_{f, T}\htheta'\Sig_u^{-1}\htheta-N)/\sqrt{2N} \rightarrow^d \mathcal{N}(0,1).
$$
Therefore, we only need to show
$$
    \frac{T\htheta'(\Sig_u^{-1}-\hSig_u^{-1})\htheta}{\sqrt{2N}}=o_P(1).
$$

%\begin{proof}
The left hand side is equal to
$$
\frac{T\htheta'\Sig_u^{-1}(\hSig_u-\Sig_u)\Sig_u^{-1}\htheta'}{\sqrt{N}}+\frac{T\htheta'(\hSig_u^{-1}-\Sig_u^{-1})(\hSig_u-\Sig_u)\Sig_u^{-1}\htheta'}{\sqrt{N}}\equiv a+b.
$$
It was shown by \cite{FLM11} that $\|\hSig_u-\Sig_u\|_2=O_P(m_N\sqrt{\frac{\log N}{T}})=\|\hSig_u^{-1}-\Sig_u^{-1}\|_2$. In addition, under $H_0$, $\|\htheta\|^2=O_P(N\log N/T)$. Hence $b=O_P(\frac{m_N^2\sqrt{N}(\log N)^2}{T})=o_P(1)$.

The challenging part is to prove $a=o_P(1)$ when $N>T$. As is described in the main text, simple inequalities like Cauchy-Schwarz accumulate estimation errors, and hence do not work.
Define $\be_t=\Sig_u^{-1}\bu_t=(e_{1t},...,e_{Nt})'$, which is an $N$-dimensional vector with mean zero and covariance $\Sig_u^{-1}$, whose entries are stochastically bounded.  Let $\bar \bw=(E\bff_t\bff_t')^{-1}E\bff_t$. A key step of proving this  proposition is to establish the following two convergences:  \begin{equation}\label{eb.1}
    \frac{1}{T}E|\frac{1}{\sqrt{NT}}\sum_{i=1}^N
    \sum_{t=1}^T(u_{it}^2-Eu_{it}^2)(\frac{1}{\sqrt{T}}\sum_{s=1}^Te_{is}
    (1-\bff_s' \bar \bw))^2|^2=o(1),
\end{equation}
\begin{equation}\label{eb.2}
   \frac{1}{T}E|\frac{1}{\sqrt{NT}}\sum_{i\neq j, (i,j)\in S_U}\sum_{t=1}^T(u_{it}u_{jt}-Eu_{it}u_{jt})
    [\frac{1}{\sqrt{T}}\sum_{s=1}^Te_{is}(1-\bff_s'\bar \bw)]
    [\frac{1}{\sqrt{T}}\sum_{k=1}^Te_{jk}(1-\bff_k' \bar \bw)]|^2=o(1),
\end{equation}
where $$S_U=\{(i,j): (\Sig_u)_{ij}\neq0\}.$$
The sparsity condition assumes that most of the off-diagonal entries of $\Sig_u$ are outside of $S_U$. The above two convergences are weighted cross-sectional and serial double sums, where the weights satisfy
$\frac{1}{\sqrt{T}}\sum_{t=1}^Te_{it}(1-\bff_t'\bar\bw)=O_P(1)$ for each $i$.  The proofs of (\ref{eb.1}) and (\ref{eb.2})  are given in the supplementary material in Appendix D.

 We consider the hard-thresholding covariance estimator.  The proof for the generalized sparsity case  as in \cite{Rothman09} is  very similar.
 Let $s_{ij}= \frac{1}{T}\sum_{t=1}^T\widehat u_{it}\widehat u_{jt}$ and $\sigma_{ij}=(\Sig_u)_{ij}$.
 Under hard-thresholding,
\begin{equation*}\label{eqa.1}
\hsig_{ij}=(\widehat \Sig_{u})_{ij}=\begin{cases}
s_{ii}, & \text{ if } i=j,\\
 s_{ij}, & \text{ if } i\neq j, |s_{ij}|>C(s_{ii}s_{jj}\frac{\log N}{T})^{1/2}\\
 0, & \text{ if } i\neq j, |s_{ij}|\leq C(s_{ii}s_{jj}\frac{\log N}{T})^{1/2}
\end{cases}
\end{equation*}
Write $(\htheta'\Sig_u^{-1})_i$ to denote the $i$th element of $\htheta'\Sig_u^{-1}$, and $S_U^c=\{(i,j): (\Sig_u)_{ij}=0\}$.
For $\sigma_{ij}\equiv(\Sig_u)_{ij}$ and $\widehat\sigma_{ij}=(\widehat\Sig_u)_{ij}$, we have
\begin{eqnarray*}
a&=&\frac{T}{\sqrt{N}}\sum_{i=1}^N(\htheta'\Sig_u^{-1})_i^2(\hsig_{ii}-\sigma_{ii})+\frac{T}{\sqrt{N}}\sum_{i\neq j, (i,j)\in S_U}(\htheta'\Sig_u^{-1})_i(\htheta'\Sig_u^{-1})_j (\hsig_{ij}-\sigma_{ij})\cr
&&+\frac{T}{\sqrt{N}}\sum_{(i,j)\in S_U^c}(\htheta'\Sig_u^{-1})_i(\htheta'\Sig_u^{-1})_j (\hsig_{ij}-\sigma_{ij})\cr
&=&a_1+a_2+a_3\end{eqnarray*}

We first examine $a_3. $  Note that
$$
    a_3=\frac{T}{\sqrt{N}}\sum_{(i,j)\in S_U^c}(\htheta'\Sig_u^{-1})_i(\htheta'\Sig_u^{-1})_j \hsig_{ij}. %\leq\max_{i\leq N}|(\htheta'\Sig_u^{-1})_i|^2 \frac{T}{\sqrt{N}}\sum_{(i,j)\in S_U^c} |\hsig_{ij}| \equiv a_{31}.
$$
Obviously,
$$
    P(a_{3}>T^{-1})\leq P(\max_{(i,j)\in S_U^c}|\hsig_{ij}|\neq0)\leq P(\max_{(i,j)\in S_U^c}|s_{ij}|>C(s_{ii}s_{jj}\frac{\log N}{T})^{1/2}).
$$
Because $s_{ii}$ is uniformly (across $i$) bounded away from zero with probability approaching one, and $\max_{(i,j)\in S_U^c}|s_{ij}|=O_P(\sqrt{\frac{\log N}{T}})$. Hence for any $\epsilon>0$, when $C$ in the threshold  is large enough,  $P(a_{3}>T^{-1})<\epsilon$, this implies   $a_3=o_P(1)$.

The proof is finished once we establish $a_i=o_P(1)$ for $i=1,2$, which are  given in Lemmas \ref{la.3} and \ref{lb.2add} respectively in the supplementary material.

%\end{proof}

\textbf{Proof of Theorem \ref{t4.1}}
%\proof  It follows from \cite{PY}(Theorem 1) that  $$(Ta_{f, T}\htheta'\Sig_u^{-1}\htheta-N)/\sqrt{2N}\rightarrow^d\mathcal{N}(0,1).$$ Hence
Part (i) follows from Proposition \ref{p4.2} and that $P(J_0=0|H_0)\rightarrow1.$ Part (ii) follows immediately from Theorem \ref{t3.3}.

\section{Proofs for Section 5}

\subsection{Proof of Proposition \ref{t5.1}}

%Since $\btheta=((\Sig_u)_{i,j}: i\neq j)$, without causing confusions, we do not distinguish $\btheta\in\Theta$ and $\Sig_u\in\Theta$.  The following lemmas are needed.

\begin{lem}\label{lc.1}
Under Assumption \ref{a5.1},
$
\inf_{\stheta\in\Theta}P(\sqrt{nT}\|\hbeta-\bbeta\|<\sqrt{\log n}|\btheta)\rightarrow1.$
\end{lem}

\proof Note that
$$
    \sqrt{nT}\|\hbeta-\bbeta\|=\|(\frac{1}{nT}\sum_{i=1}^n\sum_{t=1}^T\tx_{it}\tx_{it}')^{-1}(\frac{1}{\sqrt{nT}}\sum_{i=1}^n\sum_{t=1}^T\tx_{it}\widetilde u_{it})\|.
$$
Uniformly for $\btheta\in\Theta$, due to serial independence, and $\frac{1}{nT}\sum_{i=1}^n\sum_{t=1}^TE\tx_{it}'\tx_{it}E\widetilde u_{it}\widetilde u_{it}\leq C_1$,
\begin{eqnarray*}
 E\|\frac{1}{\sqrt{nT}}\sum_{i=1}^n\sum_{t=1}^T\tx_{it}\widetilde u_{it}\|^2&=&\frac{1}{nT}\sum_{i=1}^n\sum_{t=1}^T\sum_{j=1}^n\sum_{s=1}^TE\tx_{it}'\tx_{js}\widetilde u_{it}\widetilde u_{js}\cr
&=&\frac{1}{nT}\sum_{i=1}^n\sum_{t=1}^TE\tx_{it}'\tx_{it}E\widetilde u_{it}\widetilde u_{it}+\frac{1}{nT}\sum_{i\neq j}\sum_{t=1}^TE\tx_{it}'\tx_{jt}E\widetilde u_{it}\widetilde u_{jt}\cr
&\leq& C_1+\frac{1}{n}\sum_{i\neq j}|E\tx_{it}'\tx_{jt}| |E\widetilde u_{it}\widetilde u_{jt}|\leq C.
\end{eqnarray*}
Hence the result follows from the Chebyshev inequality and  that $\lambda_{\min}(\frac{1}{nT}\sum_{i=1}^n\sum_{t=1}^T\tx_{it}\tx_{it}')$ is bounded away from zero with probability approaching one, uniformly in $\btheta$.

\begin{lem} \label{lc.2}Suppose $\max_{j\leq n}\|\frac{1}{T}\sum_t\tx_{jt}\tx_{jt}'\|_2<C'$ with probability approaching one and $\sup_{\theta}E(u_{jt}^4|\btheta)<C'$. There is $C>0$, so that  \\
(i) $
\inf_{\stheta\in\Theta}P(\max_{j\leq n}|\frac{1}{T}\sum_{t=1}^T  u_{jt}|<C\sqrt{\log n/T}|\btheta)\rightarrow1$\\
(ii) $
\inf_{\stheta\in\Theta}P(\max_{i,j\leq n}|\frac{1}{T}\sum_{t=1}^T  u_{it}u_{jt}-Eu_{it}u_{jt}|<C\sqrt{\log n/T}|\btheta)\rightarrow1$\\
(iii) $
\inf_{\stheta\in\Theta}P(\max_{j\leq n}\frac{1}{T}\sum_{t=1}^T  (u_{jt}-\widehat u_{jt})^2<C\log n/T|\btheta)\rightarrow1$\\
(iv)   $
\inf_{\stheta\in\Theta}P(\max_{i,j\leq n}|\frac{1}{T}\sum_{t=1}^T  \widehat u_{it}\widehat u_{jt}-Eu_{it}u_{jt}|<C\sqrt{\log n/T}|\btheta)\rightarrow1$
\end{lem}

\proof (i)   By the Bernstein inequality, for   $C=(8\max_{j\leq n}\sup_{\stheta\in\Theta}E(u_{jt}^2|\btheta))^{1/2}$, we have
\begin{eqnarray*}
&&\sup_{\stheta\in\Theta}P(\max_{j\leq n}|\frac{1}{T}\sum_{t=1}^T  u_{jt}|\geq C\sqrt{\frac{\log n}{T}}|\btheta)\leq \sup_{\stheta\in\Theta}n\max_{j\leq n}P(|\frac{1}{T}\sum_{t=1}^T  u_{jt}|\geq C\sqrt{\frac{\log n}{T}}|\btheta)\cr
&&\leq  \exp(\log n-\frac{C^2\log n}{4\max_{j\leq n}\sup_{\stheta\in\Theta}E(u_{jt}^2|\btheta)})=\frac{1}{n}.
\end{eqnarray*}
Hence (i) is proved as
$
\inf_{\stheta\in\Theta}P(\max_{j\leq n}|\frac{1}{T}\sum_{t=1}^T  u_{jt}|<C\sqrt{\log n/T}|\btheta)\geq 1-\frac{1}{n}.
$

(ii)  For $C=(12\max_{j\leq n}\sup_{\stheta\in\Theta}E(u_{jt}^4|\btheta))^{1/2}$, we have
 \begin{eqnarray*}
&&\sup_{\stheta\in\Theta}P(\max_{i,j\leq n}|\frac{1}{T}\sum_{t=1}^T  u_{it}u_{jt}-Eu_{it}u_{jt}|\geq C\sqrt{\frac{\log n}{T}}|\btheta)\cr
&&\leq \sup_{\stheta\in\Theta}n^2\max_{i,j\leq n}P(|\frac{1}{T}\sum_{t=1}^T  u_{it}u_{jt}-Eu_{it}u_{jt}|\geq C\sqrt{\frac{\log n}{T}}|\btheta)\cr
&&\leq  \exp(2\log n-\frac{C^2\log n}{4\max_{j\leq n}\sup_{\stheta\in\Theta}E(u_{jt}^4|\btheta)})=\frac{1}{n}.
\end{eqnarray*}

(iii) Note that $\widehat u_{jt}-u_{jt}=-\frac{1}{T}\sum_{t=1}^Tu_{jt}-\tx_{jt}'(\hbeta-\bbeta)$, and $\max_{j\leq n}\|\frac{1}{T}\sum_t\tx_{jt}\tx_{jt}'\|_2<C$ with probability approaching one. The result then follows from part (i) and Lemma  \ref{lc.1}.

 (iv) Observe that
 \begin{eqnarray*}
 &&   |\frac{1}{T}\sum_{t=1}^T  \widehat u_{it}\widehat u_{jt}-Eu_{it}u_{jt}|\leq|\frac{1}{T}\sum_{t=1}^T   u_{it}  u_{jt}-Eu_{it}u_{jt}|+|\frac{1}{T}\sum_{t=1}^T   u_{it}  u_{jt}-\widehat u_{it}\widehat u_{jt}|\cr
 &&\leq |\frac{1}{T}\sum_{t=1}^T   u_{it}  u_{jt}-Eu_{it}u_{jt}|+\frac{1}{T}\sum_{t=1}^T  (\widehat u_{jt}-u_{jt})^2+(\frac{2}{T}\sum_tu_{jt}^2)^{1/2}(\frac{2}{T}\sum_t(\widehat u_{jt}-u_{jt})^2)^{1/2}
\end{eqnarray*}
The first two terms and  $(\frac{2}{T}\sum_t(\widehat u_{jt}-u_{jt})^2)^{1/2}$ in the third term are bounded by results in (ii) and (iii). Therefore, it suffices to show that there is a constant $M>0$ so that
$$
\inf_{\stheta\in\Theta}P(\max_{j\leq n}\frac{1}{T}\sum_t u_{jt}^2<M|\btheta)\rightarrow1.
$$
Note that $\max_{j\leq n}\frac{1}{T}\sum_t u_{jt}^2\leq \max_{j\leq n}|\frac{1}{T}\sum_t u_{jt}^2-Eu_{jt}^2|+\max_{j\leq n}Eu_{jt}^2$.
In addition, by (ii), there is $C>0$ so that
$$
\inf_{\stheta\in\Theta}P(\max_{j\leq n}|\frac{1}{T}\sum_{t=1}^T  u_{jt}^2-Eu_{jt}^2|<C\sqrt{\log n/T}|\btheta)\rightarrow1.
$$
Hence we can pick up $M$ so that $M-\sup_{\stheta\in\Theta}\max_{j\leq n}E(u_{jt}^2|\btheta)>C\sqrt{\log n/T}$, and
  \begin{eqnarray*}
 && \sup_{\stheta\in\Theta}P(\max_{j\leq n}\frac{1}{T}\sum_t u_{jt}^2\geq M|\btheta)\leq  \sup_{\stheta\in\Theta}P( \max_{j\leq n}|\frac{1}{T}\sum_t u_{jt}^2-Eu_{jt}^2|    \geq M- \max_{j\leq n}Eu_{jt}^2 |\btheta)\cr
  &&\leq  \sup_{\stheta\in\Theta}P( \max_{j\leq n}|\frac{1}{T}\sum_t u_{jt}^2-Eu_{jt}^2|    \geq C\sqrt{\frac{\log n}{T}} |\btheta)\rightarrow0.
\end{eqnarray*}
This proves the desired result.

\begin{lem} \label{lc.3} Under Assumption \ref{a5.1},  there is $C>0$,
$\inf_{\stheta\in\Theta}P(\max_{ij}|\widehat \rho_{ij}-\rho_{ij}|<C\sqrt{\log n/T}|\btheta)\rightarrow1.$
\end{lem}
\proof By the definition $\widehat\rho_{ij}=(\frac{1}{T}\sum_{t=1}^T\widehat u_{it}^2)^{-1/2}(\frac{1}{T}\sum_{t=1}^T\widehat u_{jt}^2)^{-1/2}\frac{1}{T}\sum_{t=1}^T\widehat u_{it}\widehat u_{jt}$.  By the triangular inequality,
  \begin{eqnarray*}
&&|\widehat\rho_{ij}-\rho_{ij}|\leq \underbrace{\frac{|\frac{1}{T}\sum_t\widehat u_{it}\widehat u_{jt}-u_{it}u_{jt}|}{(\frac{1}{T}\sum_{t=1}^T\widehat u_{it}^2)^{1/2}(\frac{1}{T}\sum_{t=1}^T\widehat u_{jt}^2)^{1/2}}}_{X_1}\cr
&&+\underbrace{|\frac{1}{T}\sum_tu_{it}u_{jt}|   | (\frac{1}{T}\sum_{t=1}^T\widehat u_{it}^2\frac{1}{T}\sum_{t=1}^T\widehat u_{jt}^2)^{-1/2}-  (\frac{1}{T}\sum_{t=1}^Tu_{it}^2\frac{1}{T}\sum_{t=1}^T u_{it}^2)^{-1/2}|}_{X_2}
\end{eqnarray*}
By part (iv) of Lemma \ref{lc.2},  $
\inf_{\stheta\in\Theta}P(\max_{i,j\leq n}|\frac{1}{T}\sum_{t=1}^T  \widehat u_{it}\widehat u_{jt}-Eu_{it}u_{jt}|<C\sqrt{\log n/T}|\btheta)\rightarrow1$.    Hence for sufficiently large  $M>0$ such that $\inf_{\stheta}\min_{j}E(u_{jt}^2|\btheta)-C/M>C\sqrt{\log n/T}$,
  \begin{eqnarray*}
&&\sup_{\stheta\in\Theta}P(\max_{ij}|X_1|>M\sqrt{\frac{\log n}{T}}|\btheta)\leq \sup_{\stheta\in\Theta}P(\min_{j} \frac{1}{T}\sum_t\widehat u_{jt}^2<C/M |\btheta)+o(1)\cr
&&\leq  \sup_{\stheta\in\Theta}P(\max_j| \frac{1}{T}\sum_t\widehat u_{jt}^2-Eu_{jt}^2|>\min_{j}Eu_{jt}^2-C/M |\btheta)+o(1)=o(1).
\end{eqnarray*}
By a similar argument, there is $M'>0$ so that $\sup_{\stheta\in\Theta}P(\max_{ij}|X_2|>M'\sqrt{\frac{\log n}{T}}|\btheta)=o(1)$. The result then follows as,
  \begin{eqnarray*}
&&\sup_{\theta}P(\max_{ij}|\widehat\rho_{ij}-\rho_{ij}|\geq 2(M+M')\sqrt{\log n/T}) \cr
&&\leq\sup_{\theta}P(\max_{ij}|X_1|\geq (M+M')\sqrt{\log n/T})+\sup_{\theta}P(\max_{ij}| X_2|\geq (M+M')\sqrt{\log n/T})=o(1).
\end{eqnarray*}

%For $i\neq j$,  $\sqrt{T}(\widehat\rho_{ij}-\rho_{ij})=\sqrt{T}(r_{ij}-\rho_{ij})+o_P(1)$, hence $\sqrt{T}(\widehat\rho_{ij}-\rho_{ij})\rightarrow^d\mathcal{N}(0, (1-\rho_{ij}^2)^2)$.

\textbf{Proof of Proposition \ref{p5.1}}

\proof  As $1-\rho_{ij}^2>1-c$ uniformly for $(i,j)$ and $\btheta$,   the second convergence follows from Lemma \ref{lc.3}. Also, with probability approaching one,
$$
\frac{|\widehat\rho_{ij}-\rho_{ij}|}{\widehat v_{ij}^{1/2}}\leq \frac{3\sqrt{T}}{2(1-c)}C\sqrt{\frac{\log n}{T}}<\delta_{N,T}/2.$$

\subsection{Proof of Theorem \ref{t5.1}}

\begin{lem} \label{lc.4} There is $C>0$ so that
$J_1$ has power uniformly on $\Theta(J_1)=\{ \sum_{i<j}\rho_{ij}^2\geq  Cn^2\log n /T\}$.
\end{lem}

\proof   By Lemma \ref{lc.3},  there is $C>0$,
$\inf_{\stheta\in\Theta}P(\max_{ij}|\widehat \rho_{ij}-\rho_{ij}|<C\sqrt{\log n/T}|\btheta)\rightarrow1.$ If we define
$$
A=\{\sum_{i<j}(\widehat\rho_{ij}-\rho_{ij})^2< C^2n^2(\log n/T) \},
$$
then $\inf_{\Theta}P(A|\btheta)\rightarrow1.$ On the event $A$,  we have, uniformly in $\btheta=\{\rho_{ij}\},$
$$\sum_{i<j}(\widehat\rho_{ij}-\rho_{ij})\rho_{ij}\leq (\sum_{i<j}(\widehat\rho_{ij}-\rho_{ij})^{2})^{1/2}(\sum_{i<j}\rho_{ij}^2)^{1/2}\leq \frac{Cn\sqrt{\log n}}{\sqrt{T}}(\sum_{i<j}\rho_{ij}^2)^{1/2}.
$$
Therefore, when $  \sum_{i<j}\rho_{ij}^2\geq  16C^2n^2\log n /T$,
$$\sum_{i<j}\widehat\rho_{ij}^2=\sum_{i<j}(\widehat\rho_{ij}-\rho_{ij})^2+\rho_{ij}^2+2(\widehat\rho_{ij}-\rho_{ij})\rho_{ij}\geq \sum_{i<j}\rho_{ij}^2-\frac{2Cn\sqrt{\log n}}{\sqrt{T}}(\sum_{i<j}\rho_{ij}^2)^{1/2}\geq\frac{1}{2}\sum_{i<j}\rho_{ij}^2.
$$
This entails that when $  \sum_{i<j}\rho_{ij}^2\geq  16Cn^2\log n /T$, we have
\begin{eqnarray*}
&&\sup_{\Theta(J_1)}P(J_1<F_q|\btheta)\leq \sup_{\Theta(J_1)}P(
\sum_{i<j}\widehat\rho^2_{ij}<  \frac{n(n-1)}{2T}+(F_q+\frac{n}{2(T-1)})  \frac{\sqrt{n(n-1)}}{T}  |\btheta)\cr
&&\leq \sup_{\Theta(J_1)}P(\frac{1}{2}\sum_{i<j}\rho_{ij}^2<  \frac{n(n-1)}{2T}+(F_q+\frac{n}{2(T-1)})  \frac{\sqrt{n(n-1)}}{T}  |\btheta)+\sup_{\Theta(J_1)}P(A^c|\btheta)\rightarrow0.
\end{eqnarray*}

\textbf{Proof of Theorem \ref{t5.1}}

It suffices to verify conditions (i)-(iii) of Theorem \ref{t3.2}. Condition (i) follows from Theorem 1 of \cite{BFK}.  As for condition (ii), note that $J_1\geq -\frac{\sqrt{n(n-1)}}{2}-\frac{n}{2(T-1)}$ almost surely. Hence as $n,T\rightarrow\infty$,
$$\inf_{\stheta\in\Theta_s}P(c\sqrt{N}+J_1>z_q|\btheta) \geq  \inf_{\stheta\in\Theta_s}P(c\sqrt{N}-\frac{\sqrt{n(n-1)}}{2}-\frac{n}{2(T-1)}>z_q|\btheta)=1. $$
Finally,  condition (iii) follows from Lemma \ref{lc.4}.

\newpage

\section{Supplementary Material}

\subsection{Auxiliary lemmas for the proof of Proposition \ref{p4.2}}

Define $\be_t=\Sig_u^{-1}\bu_t=(e_{1t},...,e_{Nt})'$, which is an $N$-dimensional vector with mean zero and covariance $\Sig_u^{-1}$, whose entries are stochastically bounded.  Let $\bar \bw=(E\bff_t\bff_t')^{-1}E\bff_t$. Also recall that
$$
a_1=\frac{T}{\sqrt{N}}\sum_{i=1}^N(\htheta'\Sig_u^{-1})_i^2(\hsig_{ii}-\sigma_{ii}),
$$
$$
a_2=\frac{T}{\sqrt{N}}\sum_{i\neq j, (i,j)\in S_U}(\htheta'\Sig_u^{-1})_i(\htheta'\Sig_u^{-1})_j (\hsig_{ij}-\sigma_{ij}).
$$

One of the key steps of proving   $a_1=o_P(1), a_2=o_P(1)$ is to establish the following two convergences:  \begin{equation}\label{eq.D1}
    \frac{1}{T}E|\frac{1}{\sqrt{NT}}\sum_{i=1}^N
    \sum_{t=1}^T(u_{it}^2-Eu_{it}^2)(\frac{1}{\sqrt{T}}\sum_{s=1}^Te_{is}
    (1-\bff_s' \bar \bw))^2|^2=o(1),
\end{equation}
\begin{equation}\label{eq.D2}
   \frac{1}{T}E|\frac{1}{\sqrt{NT}}\sum_{i\neq j, (i,j)\in S_U}\sum_{t=1}^T(u_{it}u_{jt}-Eu_{it}u_{jt})
    [\frac{1}{\sqrt{T}}\sum_{s=1}^Te_{is}(1-\bff_s'\bar \bw)]
    [\frac{1}{\sqrt{T}}\sum_{k=1}^Te_{jk}(1-\bff_k' \bar \bw)]|^2=o(1),
\end{equation}
where $S_U=\{(i,j): (\Sig_u)_{ij}\neq0\}.$ The proofs of (\ref{eq.D1}) and (\ref{eq.D2})  are given later below.

\begin{lem} \label{la.3}Under $H_0$,
$a_1=o_P(1)$.
\end{lem}

\begin{proof}
We have $a_1=\frac{T}{\sqrt{N}}\sum_{i=1}^N(\htheta'\Sig_u^{-1})_i^2\frac{1}{T}\sum_{t=1}^T(\widehat u_{it}^2-Eu_{it}^2)$,  which is
$$
\frac{T}{\sqrt{N}}\sum_{i=1}^N(\htheta'\Sig_u^{-1})_i^2\frac{1}{T}\sum_{t=1}^T(\widehat u_{it}^2- u_{it}^2)+\frac{T}{\sqrt{N}}\sum_{i=1}^N(\htheta'\Sig_u^{-1})_i^2\frac{1}{T}\sum_{t=1}^T(u_{it}^2-Eu_{it}^2)=a_{11}+a_{12}.
$$
For $a_{12}$, note that $(\htheta'\Sig_u^{-1})_i=(1-\bar\bff'\bw)^{-1}\frac{1}{T}\sum_{s=1}^T(1-\bff_s'\bw)(\bu_s'\Sig_u^{-1})_i=c\frac{1}{T}\sum_{s=1}^T(1-\bff_s'\bw)e_{is}$, where $c=(1-\bar\bff'\bw)^{-1}=O_P(1).$  Hence
$$
a_{12}=\frac{Tc}{\sqrt{N}}\sum_{i=1}^N(\frac{1}{T}\sum_{s=1}^T(1-\bff_s'\bw)e_{is})^2\frac{1}{T}\sum_{t=1}^T(u_{it}^2-Eu_{it}^2)
$$
By (\ref{eq.D1}), $Ea_{12}^2=o(1)$. On the other hand,
$$
a_{11}=\frac{T}{\sqrt{N}}\sum_{i=1}^N(\htheta'\Sig_u^{-1})_i^2\frac{1}{T}\sum_{t=1}^T(\widehat u_{it}-u_{it})^2+\frac{2T}{\sqrt{N}}\sum_{i=1}^N(\htheta'\Sig_u^{-1})_i^2\frac{1}{T}\sum_{t=1}^Tu_{it}(\widehat u_{it}-u_{it})=a_{111}+a_{112}.
$$
Note that $\max_{i\leq N}\frac{1}{T}\sum_{t=1}^T(\widehat u_{it}-u_{it})^2=O_P({\frac{\log N}{T}})$ by Lemma 3.1 of \cite{FLM11}. Since $\|\htheta\|^2=O_P(\frac{N\log N}{T})$, $\|\Sig_u^{-1}\|_2=O(1)$ and $N(\log N)^3=o(T^2)$,
$$a_{111}\leq O_P({\frac{\log N}{T}})\frac{T}{\sqrt{N}} \|\widehat\btheta'\Sig_u^{-1}\|^2=O_P(\frac{(\log N)^2\sqrt{N}}{T})=o_P(1),
$$
To bound $a_{112}$,  note that
$$
\widehat u_{it}-u_{it}=\widehat\theta_i-\theta_i+(\hb_i-\bb_i)'\bff_t,\quad \max_i|\widehat \theta_i-\theta_i|=O_P(\sqrt{\frac{\log N}{T}})=\max_i\|\hb_i-\bb_i\|.
$$
Also,  $\max_i|\frac{1}{T}\sum_{t=1}^Tu_{it}|=O_P(\sqrt{\frac{\log N}{T}})=\max_i\|\frac{1}{T}\sum_{t=1}^Tu_{it}\bff_t\|$. Hence
\begin{eqnarray*}
a_{112}&=&\frac{2T}{\sqrt{N}}\sum_{i=1}^N(\htheta'\Sig_u^{-1})_i^2\frac{1}{T}\sum_{t=1}^Tu_{it}(\widehat \theta_{i}-\theta_{i})+\frac{2T}{\sqrt{N}}\sum_{i=1}^N(\htheta'\Sig_u^{-1})_i^2(\hb_i-\bb_i)'\frac{1}{T}\sum_{t=1}^T\bff_tu_{it}\cr
&\leq& O_P(\frac{\log N}{\sqrt{N}})\|\htheta'\Sig_u^{-1}\|^2=o_P(1).
\end{eqnarray*}
In summary, $a_1=a_{12}+a_{111}+a_{112}=o_P(1)$.
\end{proof}

 \begin{lem} \label{lb.2add}Under $H_0$,
$a_2=o_P(1)$.
\end{lem}
\begin{proof} We have $a_2=\frac{T}{\sqrt{N}}\sum_{i\neq j, (i,j)\in S_U}(\htheta'\Sig_u^{-1})_i(\htheta'\Sig_u^{-1})_j \frac{1}{T}\sum_{t=1}^T(\widehat u_{it}\widehat u_{jt}-Eu_{it}u_{jt})$, which is
$$\frac{T}{\sqrt{N}}\sum_{i\neq j, (i,j)\in S_U}(\htheta'\Sig_u^{-1})_i(\htheta'\Sig_u^{-1})_j \left(\frac{1}{T}\sum_{t=1}^T(\widehat u_{it}\widehat u_{jt}-u_{it}u_{jt})+\frac{1}{T}\sum_{t=1}^T( u_{it}u_{jt}-Eu_{it}u_{jt})\right)=a_{21}+a_{22}.$$
where
$$
a_{21}=\frac{T}{\sqrt{N}}\sum_{i\neq j, (i,j)\in S_U}(\htheta'\Sig_u^{-1})_i(\htheta'\Sig_u^{-1})_j \frac{1}{T}\sum_{t=1}^T(\widehat u_{it}\widehat u_{jt}-u_{it}u_{jt}).
$$
Under $H_0$, $\Sig_u^{-1}\widehat\btheta=\frac{1}{T}(1-\bar\bff'\bw)^{-1}\sum_{t=1}^T\Sig_u^{-1}\bu_t(1-\bff_t'\bw)$, and $\be_t=\Sig_u^{-1}\bu_t$, we have
\begin{eqnarray*}
a_{22}&=&\frac{T}{\sqrt{N}}\sum_{i\neq j, (i,j)\in S_U}(\htheta'\Sig_u^{-1})_i(\htheta'\Sig_u^{-1})_j \frac{1}{T}\sum_{t=1}^T(u_{it}u_{jt}-Eu_{it}u_{jt})\cr
&=&\frac{Tc}{\sqrt{N}}\sum_{i\neq j, (i,j)\in S_U}\frac{1}{T}\sum_{s=1}^T(1-\bff_s'\bw)e_{is}\frac{1}{T}\sum_{k=1}^T(1-\bff_k'\bw)e_{jk}		 \frac{1}{T}\sum_{t=1}^T( u_{it}u_{jt}-Eu_{it}u_{jt}).
\end{eqnarray*}
By (\ref{eq.D2}), $Ea_{22}^2=o(1)$.

On the other hand, $a_{21}=a_{211}+a_{212}$, where
\begin{eqnarray*}
a_{211}&=&\frac{T}{\sqrt{N}}\sum_{i\neq j, (i,j)\in S_U}(\htheta'\Sig_u^{-1})_i(\htheta'\Sig_u^{-1})_j \frac{1}{T}\sum_{t=1}^T(\widehat u_{it}-u_{it})(\widehat u_{jt}-u_{jt}),\cr
a_{212}&=&\frac{2T}{\sqrt{N}}\sum_{i\neq j, (i,j)\in S_U}(\htheta'\Sig_u^{-1})_i(\htheta'\Sig_u^{-1})_j \frac{1}{T}\sum_{t=1}^Tu_{it}(\widehat u_{jt}-u_{jt}).
\end{eqnarray*}
By the Cauchy-Schwarz inequality,  $\max_{ij}|\frac{1}{T}\sum_{t=1}^T(\widehat u_{it}-u_{it})(\widehat u_{jt}-u_{jt})|=O_P(\frac{\log N}{T})$. Hence
\begin{eqnarray*}
|a_{211}|&\leq& O_P(\frac{\log N}{\sqrt{N}})\sum_{i\neq j, (i,j)\in S_U}|(\htheta'\Sig_u^{-1})_i||(\htheta'\Sig_u^{-1})_j |\cr
&\leq& O_P(\frac{\log N}{\sqrt{N}})\left(\sum_{i\neq j, (i,j)\in S_U}(\htheta'\Sig_u^{-1})_i^2\right)^{1/2}\left(\sum_{i\neq j, (i,j)\in S_U}(\htheta'\Sig_u^{-1})_j^2\right)^{1/2}\cr
&=&O_P(\frac{\log N}{\sqrt{N}})\sum_{i=1}^N(\htheta'\Sig_u^{-1})_i^2\sum_{j: (\Sig_u)_{ij}\neq 0}1\leq O_P(\frac{\log N}{\sqrt{N}})\|\htheta'\Sig_u^{-1}\|^2m_N\cr
&=&O_P(\frac{m_N\sqrt{N}(\log N)^2}{T})=o_P(1).
\end{eqnarray*}

Similar to the proof of term $a_{112}$ in Lemma  \ref{la.3}, $\max_{ij}|\frac{1}{T}\sum_{t=1}^Tu_{it}(\widehat u_{jt}-u_{jt})|=O_P(\frac{\log N}{T})$.
$$
|a_{212}|\leq O_P(\frac{\log N}{\sqrt{N}})\sum_{i\neq j, (i,j)\in S_U}|(\htheta'\Sig_u^{-1})_i||(\htheta'\Sig_u^{-1})_j |=O_P(\frac{m_N\sqrt{N}(\log N)^2}{T})=o_P(1).
$$
In summary, $a_2=a_{22}+a_{211}+a_{212}=o_P(1)$.
\end{proof}

\subsection{Proof of (\ref{eq.D1}) and (\ref{eq.D2})}

For any index set $A$, we let $|A|_0$ denote its number of elements.

\begin{lem}\label{la.5} Recall that $\be_t=\Sig_u^{-1}\bu_t$.   $e_{it}$ and $u_{jt}$ are independent if $i\neq j$.
\end{lem}
\begin{proof}
Because $\bu_t$ is Gaussian, it suffices to show that $\cov(e_{it}, u_{jt})=0$ when $i\neq j.$ Consider the vector $(\bu_t', \be_t')'=\bA(\bu_t',\bu_t')' $, where
$$\bA=\begin{pmatrix}
                                        \bI_N&0\\
                 0&\Sig_u^{-1}                                           \end{pmatrix}. $$
                 Then $\cov(\bu_t', \be_t')=\bA\cov(\bu_t',\bu_t')\bA$, which is
$$\begin{pmatrix}
                                        \bI_N&0\\
                 0&\Sig_u^{-1}   \end{pmatrix}
                 \begin{pmatrix}
           \Sig_u& \Sig_u\\
                  \Sig_u&\Sig_u  \end{pmatrix}
                 \begin{pmatrix}
                                        \bI_N&0\\
                 0&\Sig_u^{-1}   \end{pmatrix}=      \begin{pmatrix}
                                     \Sig_u&\bI_N\\
                 \bI_N&\Sig_u^{-1}   \end{pmatrix}.
$$
This completes the proof.
\end{proof}

\textbf{Proof of (\ref{eq.D1}) }

Let $X=\frac{1}{\sqrt{NT}}\sum_{i=1}^N\sum_{t=1}^T(u_{it}^2-Eu_{it}^2)(\frac{1}{\sqrt{T}}\sum_{s=1}^Te_{is}(1-\bff_s'\bw))^2$. The goal is to show $EX^2=o(T)$.   We show respectively $\frac{1}{T}(EX)^2=o(1)$ and $\frac{1}{T}\var(X)=o(1)$.
The proof  of (\ref{eq.D1})  is the same regardless of the type of  sparsity  in Assumption  \ref{a4.2}.
 For notational simplicity, let
$$\xi_{it}=u_{it}^2-Eu_{it}^2,\quad \zeta_{is}=e_{is}(1-\bff_s'\bw).
$$
Then $X=\frac{1}{\sqrt{NT}}\sum_{i=1}^N\sum_{t=1}^T\xi_{it}(\frac{1}{\sqrt{T}}\sum_{s=1}^T \zeta_{is})^2$.   Because of the serial independence, $\xi_{it}$ is independent of $\zeta_{js}$ if $t\neq s$, for any $i,j\leq N$, which implies $\cov(\xi_{it}, \zeta_{is}\zeta_{ik})=0$ as long as either $s\neq t$ or $k\neq t.$

%Note that $\bu_t$ is Gaussian and independent across $t$ and also independent of $\bff_t$, hence

 \textbf{Expectation}

 For the expectation,
\begin{eqnarray*}
EX&=&\frac{1}{\sqrt{NT}}\sum_{i=1}^N\sum_{t=1}^T\cov(\xi_{it}, (\frac{1}{\sqrt{T}}\sum_{s=1}^T \zeta_{is})^2)=\frac{1}{T\sqrt{NT}}\sum_{i=1}^N\sum_{t=1}^T\sum_{s=1}^T\sum_{k=1}^T\cov(\xi_{it},  \zeta_{is}\zeta_{ik})\cr
&=&\frac{1}{T\sqrt{NT}}\sum_{i=1}^N\sum_{t=1}^T(\cov(\xi_{it}, \zeta_{it}^2)+2\sum_{k\neq t}\cov(\xi_{it}, \zeta_{it}\zeta_{ik}))\cr
&=&\frac{1}{T\sqrt{NT}}\sum_{i=1}^N\sum_{t=1}^T\cov(\xi_{it}, \zeta_{it}^2)=O(\sqrt{\frac{N}{T}}),
\end{eqnarray*}
where the second last equality follows since $E\xi_{it}=E\zeta_{it}=0$ and  when $k\neq t$ $\cov(\xi_{it}, \zeta_{it}\zeta_{ik})=E\xi_{it}\zeta_{it}\zeta_{ik}=E\xi_{it}\zeta_{it}E\zeta_{ik}=0.$ It then follows that $\frac{1}{T}(EX)^2=O(\frac{N}{T^2})=o(1)$, given $N=o(T^2)$.

  \textbf{Variance}

Consider the variance.  We have,
\begin{eqnarray*}
\var(X)&=&\frac{1}{N}\sum_{i=1}^N\var(\frac{1}{\sqrt{T}}\sum_{t=1}^T\xi_{it}(\frac{1}{\sqrt{T}}\sum_{s=1}^T \zeta_{is})^2)\cr
&&+\frac{1}{NT^3}\sum_{i\neq j}\sum_{t,s,k,l,v,p\leq T}\cov(\xi_{it} \zeta_{is}\zeta_{ik},  \xi_{jl}  \zeta_{jv}\zeta_{jp})=B_1+B_2.
\end{eqnarray*}
$B_1$ can be bounded by the Cauchy-Schwarz inequality. Note that $E\xi_{it}=E\zeta_{js}=0$,
\begin{eqnarray*}
B_1&\leq&\frac{1}{N}\sum_{i=1}^NE(\frac{1}{\sqrt{T}}\sum_{t=1}^T\xi_{it}(\frac{1}{\sqrt{T}}\sum_{s=1}^T \zeta_{is})^2)^2\leq \frac{1}{N}\sum_{i=1}^N[E(\frac{1}{\sqrt{T}}\sum_{t=1}^T\xi_{it})^4]^{1/2}[E(\frac{1}{\sqrt{T}}\sum_{s=1}^T \zeta_{is})^8]^{1/2}.
\end{eqnarray*}
Hence  $B_1=O(1)$.

We now show $\frac{1}{T}B_2=o(1)$. Once this is done, it implies $\frac{1}{T}\var(X)=o(1)$. The proof of (\ref{eq.D1}) is then completed because $\frac{1}{T}EX^2=\frac{1}{T}(EX)^2+\frac{1}{T}\var(X)=o(1)$.

For two variables $X,Y$, writing $X\perp Y$ if they are independent.
Note that  $E\xi_{it}=E\zeta_{is}=0$, and when $t\neq s$, $\xi_{it}\perp\zeta_{js}$, $\xi_{it}\perp \xi_{js}$,  $\zeta_{it}\perp\zeta_{js}$   for any $i,j\leq N$. Therefore, it is straightforward to verify that if  the set $\{t,s,k,l,v,p\}$ contains more than three distinct elements,  then $\cov(\xi_{it} \zeta_{is}\zeta_{ik},  \xi_{jl}  \zeta_{jv}\zeta_{jp})=0.$ Hence if we denote $\Xi$ as the set of  $(t,s,k,l,v,p)$ such that $\{t,s,k,l,v,p\}$ contains no more than three distinct elements, then its cardinality satisfies: $|\Xi|_0\leq CT^3$ for some $C>1$, and
$$\sum_{t,s,k,l,v,p\leq T} \cov(\xi_{it} \zeta_{is}\zeta_{ik},  \xi_{jl}  \zeta_{jv}\zeta_{jp})=\sum_{(t,s,k,l,v,p)\in \Xi}\cov(\xi_{it} \zeta_{is}\zeta_{ik},  \xi_{jl}  \zeta_{jv}\zeta_{jp}).
$$
Hence  $$
B_2=\frac{1}{NT^3}\sum_{i\neq j}\sum_{(t,s,k,l,v,p)\in \Xi}\cov(\xi_{it} \zeta_{is}\zeta_{ik},  \xi_{jl}  \zeta_{jv}\zeta_{jp}).
$$
Let us partition $\Xi$ into $\Xi_1\cup\Xi_2$  where each element $(t,s,k,l,v,p)$ in $\Xi_1$ contains  exactly three distinct indices, while each element in $\Xi_2$ contains less than three distinct indices. We know that  $\frac{1}{NT^3}\sum_{i\neq j}\sum_{(t,s,k,l,v,p)\in \Xi_2}\cov(\xi_{it} \zeta_{is}\zeta_{ik},  \xi_{jl}  \zeta_{jv}\zeta_{jp})=O(\frac{1}{NT^3}N^2T^2)=O(\frac{N}{T}),
$ which implies  $$\frac{1}{T}B_2=\frac{1}{NT^4}\sum_{i\neq j}\sum_{(t,s,k,l,v,p)\in \Xi_1}\cov(\xi_{it} \zeta_{is}\zeta_{ik},  \xi_{jl}  \zeta_{jv}\zeta_{jp})+O_p(\frac{N}{T^2}).$$
The first term on the right hand side can be written as $\sum_{h=1}^5B_{2h}$. Each of these five terms is defined and analyzed separately as below.
$$
B_{21}=\frac{1}{NT^4}\sum_{i\neq j}\sum_{t=1}^T\sum_{s\neq t}\sum_{l\neq s, t}E\xi_{it}\xi_{jt}E\zeta_{is}^2E\zeta_{jl}^2\leq O(\frac{1}{NT})\sum_{i\neq j}|E\xi_{it}\xi_{jt}|.
$$
Note that if $(\Sig_{u})_{ij}=0$, $u_{it}$ and $u_{jt}$ are independent, and hence $E\xi_{it}\xi_{jt}=0$. This implies $\sum_{i\neq j}|E\xi_{it}\xi_{jt}|\leq O(1)\sum_{i\neq j, (i,j)\in S_U}1=O(N)$. Hence $B_{21}=o(1)$.
$$
B_{22}=\frac{1}{NT^4}\sum_{i\neq j}\sum_{t=1}^T\sum_{s\neq t}\sum_{l\neq s, t}E\xi_{it}\zeta_{it}E\zeta_{is}\xi_{js}E\zeta_{jl}^2.
$$
By Lemma \ref{la.5}, $u_{js}$ and $e_{is}$ are independent for $i\neq j.$ Also, $u_{js}$ and $\bff_s$ are independent, which implies $\xi_{js}$ and $\zeta_{is}$ are independent. So $E\xi_{js}\zeta_{is}=0$. It follows that $B_{22}=0.$
\begin{eqnarray*}
B_{23}&=&\frac{1}{NT^4}\sum_{i\neq j}\sum_{t=1}^T\sum_{s\neq t}\sum_{l\neq s, t}E\xi_{it}\zeta_{it}E\zeta_{is}\zeta_{js}E\xi_{jl}\zeta_{jl}=O(\frac{1}{NT})\sum_{i\neq j}|E\zeta_{is}\zeta_{js}|\cr
&&=O(\frac{1}{NT})\sum_{i\neq j}|Ee_{is}e_{js}E(1-\bff_s'\bw)^2|=O(\frac{1}{NT})\sum_{i\neq j}|Ee_{is}e_{js}|.
\end{eqnarray*}
By the definition   $\be_s=\Sig_u^{-1}\bu_s$, $\cov(\be_s)=\Sig_u^{-1}$. Hence $Ee_{is}e_{js}=(\Sig_u^{-1})_{ij}$, which implies $B_{23}\leq O(\frac{N}{NT})\|\Sig_u^{-1}\|_1=o(1)$.
$$
B_{24}=\frac{1}{NT^4}\sum_{i\neq j}\sum_{t=1}^T\sum_{s\neq t}\sum_{l\neq s, t}E\xi_{it}\xi_{jt}E\zeta_{is}\zeta_{js}E\zeta_{il}\zeta_{jl}=O(\frac{1}{T}),
$$
which is analyzed in the same way as $B_{21}$.

Finally,
$
B_{25}=\frac{1}{NT^4}\sum_{i\neq j}\sum_{t=1}^T\sum_{s\neq t}\sum_{l\neq s, t}E\xi_{it}\zeta_{jt}E\zeta_{is}\xi_{js}E\zeta_{il}\zeta_{jl}=0$, because $E\zeta_{is}\xi_{js}=0$ when  $i\neq j$, following from Lemma  \ref{la.5}. Therefore, $\frac{1}{T}B_2=o(1)+O(\frac{N}{T^2})=o(1)$.

\textbf{Proof of (\ref{eq.D2})}

For notational simplicity, let $\xi_{ijt}=u_{it}u_{jt}-Eu_{it}u_{jt}$. Because of the serial independence and the Gaussianity, $\cov(\xi_{ijt}, \zeta_{ls}\zeta_{nk})=0$ when either $s\neq t$ or $k\neq t$,  for any $i,j,l,n\leq N$.
In addition, define a set
$$
H=\{(i,j)\in S_U: i\neq j\}.
$$Then by the sparsity assumption, $\sum_{(i,j)\in H}1=D_N=O(N)$. Now let
\begin{eqnarray*}
Z&=&\frac{1}{\sqrt{NT}}\sum_{(i,j)\in H}\sum_{t=1}^T(u_{it}u_{jt}-Eu_{it}u_{jt})[\frac{1}{\sqrt{T}}\sum_{s=1}^Te_{is}(1-\bff_s'\bw)][\frac{1}{\sqrt{T}}\sum_{k=1}^Te_{jk}(1-\bff_k'\bw)]\cr
&=&\frac{1}{\sqrt{NT}}\sum_{(i,j)\in H}\sum_{t=1}^T \xi_{ijt}[\frac{1}{\sqrt{T}}\sum_{s=1}^T \zeta_{is}][\frac{1}{\sqrt{T}}\sum_{k=1}^T\zeta_{jk}]=\frac{1}{T\sqrt{NT}}\sum_{(i,j)\in H}\sum_{t=1}^T\sum_{s=1}^T\sum_{k=1}^T  \xi_{ijt}\zeta_{is}\zeta_{jk}.
\end{eqnarray*}
The goal is to show $\frac{1}{T}EZ^2=o(1)$.  We  respectively show $\frac{1}{T}(EZ)^2=o(1)=\frac{1}{T}\var(Z)$.

 \textbf{Expectation}

 The proof for the expectation is the same regardless of the type of sparsity in Assumption  \ref{a4.2}, and is very similar to that of  (\ref{eq.D1}). In fact,
\begin{eqnarray*}
EZ&=&\frac{1}{T\sqrt{NT}}\sum_{(i,j)\in H}\sum_{t=1}^T\sum_{s=1}^T\sum_{k=1}^T  \cov(\xi_{ijt}, \zeta_{is}\zeta_{jk})=\frac{1}{T\sqrt{NT}}\sum_{(i,j)\in H}\sum_{t=1}^T \cov(\xi_{ijt}, \zeta_{it}^2).
\end{eqnarray*}
 Because $\sum_{(i,j)\in H}1=O(N)$, $EZ=O(\sqrt{\frac{N}{T}}).$ Thus $\frac{1}{T}(EZ)^2=o(1)$.

  \textbf{Variance}

 For the variance,  we have
\begin{eqnarray*}
\var(Z)&=&\frac{1}{T^3N} \sum_{(i,j)\in H}\var(\sum_{t=1}^T\sum_{s=1}^T\sum_{k=1}^T  \xi_{ijt}\zeta_{is}\zeta_{jk})\cr
&&+\frac{1}{T^3N}\sum_{(i,j)\in H,}\sum_{(m,n)\in H, (m,n)\neq (i,j),}\sum_{t,s,k,l,v,p\leq T}\cov(\xi_{ijt}\zeta_{is}\zeta_{jk}, \xi_{mnl}\zeta_{mv}\zeta_{np})\cr
&=&A_1+A_2.
\end{eqnarray*}
By the Cauchy-Schwarz inequality and the serial independence of $\xi_{ijt}$,
\begin{eqnarray*}
A_1&\leq& \frac{1}{N} \sum_{(i,j)\in H}E[\frac{1}{\sqrt{T}}\sum_{t=1}^T\xi_{ijt}\frac{1}{\sqrt{T}}\sum_{s=1}^T  \zeta_{is}\frac{1}{\sqrt{T}}\sum_{k=1}^T\zeta_{jk}]^2\cr
&\leq& \frac{1}{N} \sum_{(i,j)\in H}[E(\frac{1}{\sqrt{T}}\sum_{t=1}^T\xi_{ijt})^4]^{1/2}[E(\frac{1}{\sqrt{T}}\sum_{s=1}^T  \zeta_{is})^8]^{1/4}[E(\frac{1}{\sqrt{T}}\sum_{k=1}^T\zeta_{jk})^8]^{1/4}.
\end{eqnarray*}
So $A_1=O(1)$.
%Because $\zeta_{it}$ is also serially independent and $E\zeta_{it}=0$, $\sum_{s,k,v,p\leq T}E\zeta_{is}\zeta_{ik}\zeta_{jv}\zeta_{jp}=O(T^2)$ uniformly in $i,j\leq N$. This yields  $A_1=O(1)$.

Note that  $E\xi_{ijt}=E\zeta_{is}=0$, and when $t\neq s$, $\xi_{ijt}\perp\zeta_{ms}$, $\xi_{ijt}\perp \xi_{mns}$,  $\zeta_{it}\perp\zeta_{js}$ (independent) for any $i,j,m, n\leq N$. Therefore, it is straightforward to verify that if  the set $\{t,s,k,l,v,p\}$ contains more than three distinct elements,  then $\cov(\xi_{ijt}\zeta_{is}\zeta_{jk}, \xi_{mnl}\zeta_{mv}\zeta_{np})=0.$ Hence for the same set $\Xi$  defined as  before, it satisfies:  $|\Xi|_0\leq CT^3$ for some $C>1$, and
$$\sum_{t,s,k,l,v,p\leq T}\cov(\xi_{ijt}\zeta_{is}\zeta_{jk}, \xi_{mnl}\zeta_{mv}\zeta_{np})=\sum_{(t,s,k,l,v,p)\in \Xi}\cov(\xi_{ijt}\zeta_{is}\zeta_{jk}, \xi_{mnl}\zeta_{mv}\zeta_{np}).
$$

We  proceed by studying the two cases of Assumption  \ref{a4.2} separately, and show that in both cases $\frac{1}{T}A_2=o(1)$. Once this is done, because we have just shown $A_1=O(1)$, then $\frac{1}{T}\var(Z)=o(1)$. The proof is then completed because $\frac{1}{T}EZ^2=\frac{1}{T}(EZ)^2+\frac{1}{T}\var(Z)=o(1)$.

 \textbf{When  $D_N=O(\sqrt{N})$}

 Because $|\Xi|_0\leq CT^3$ and $|H|_0=D_N=O(\sqrt{N})$,   and $|\cov(\xi_{ijt}\zeta_{is}\zeta_{jk}, \xi_{mnl}\zeta_{mv}\zeta_{np})|$ is bounded uniformly in $i,j,m,n\leq N$, we have
  $$
\frac{1}{T}A_2=\frac{1}{T^4N}\sum_{(i,j)\in H,}\sum_{(m,n)\in H, (m,n)\neq (i,j),}\sum_{t,s,k,l,v,p\in\Xi}\cov(\xi_{ijt}\zeta_{is}\zeta_{jk}, \xi_{mnl}\zeta_{mv}\zeta_{np})=O(\frac{1}{T}).
 $$

  \textbf{When    $D_n=O(N)$, and $m_N=O(1)$}

Similar to the proof of the first statement, for the same set $\Xi_1$ that contains exactly three distinct indices in each of its element, (recall $|H|_0=O(N)$)
$$\frac{1}{T}A_2=\frac{1}{NT^4}\sum_{(i,j)\in H,}\sum_{(m,n)\in H, (m,n)\neq (i,j),}\sum_{t,s,k,l,v,p\in\Xi_1}\cov(\xi_{ijt}\zeta_{is}\zeta_{jk}, \xi_{mnl}\zeta_{mv}\zeta_{np})+O(\frac{N}{T^2}).$$
The first term on the right hand side can be written as $\sum_{h=1}^5A_{2h}$. Each of these five  terms is defined and analyzed separately as below.  Before that, let us introduce a useful lemma.

  The following lemma is needed  when $\Sig_u$ has bounded number of nonzero entries in each row ($m_N=O(1)$).   Let $|S|_0$ denote the number of elements in a set $S$ if $S$ is  countable. For any $i\leq N$, let $$A(i)=\{j\leq N: \cov(u_{it}, u_{jt})\neq0\}=\{j\leq N: (i,j)\in S_U\}.$$

   \begin{lem} \label{la.6}Suppose $m_N=O(1)$. For any $i, j\leq N$, let $B(i,j)$ be a set of $k\in\{1,...,N\}$ such that:\\
(i) $k\notin A(i)\cup A(j)$ \\
(ii) there is $p\in A(k)$ such that $\cov(u_{it}u_{jt}, u_{kt}u_{pt})\neq0$.

 Then $
 \max_{i,j\leq N}|B(i,j)|_0=O(1).
 $
\end{lem}
\begin{proof} First we note that if $B(i,j)=\emptyset$, then $|B(i,j)|_0=0.$ If it is not empty, for any $k\in B(i,j)$,  by definition, $k\notin A(i)\cup A(j)$, which implies $\cov(u_{it}, u_{kt})=\cov(u_{jt}, u_{kt})=0$. By the Gaussianity,   $u_{kt}$ is independent of $(u_{it}, u_{jt})$. Hence if  $p\in A(k)$ is such that $\cov(u_{it}u_{jt}, u_{kt}u_{pt})\neq0$, then $u_{pt}$ should be correlated with either $u_{it}$ or $u_{jt}$. We thus must have $p\in A(i)\cup A(j)$. In other words, there is $p\in A(i)\cup A(j)$ such that $\cov(u_{kt}, u_{pt})\neq0$, which implies $k\in A(p)$. Hence,
$$
k\in \bigcup_{p\in A(i)\cup A(j)}A(p)\equiv M(i,j),
$$
and thus $B(i,j)\subset M(i,j)$.
Because $m_N=O(1)$, $\max_{i\leq N}|A(i)|_0=O(1)$, which implies $\max_{i,j}|M(i,j)|_0=O(1)$, yielding the result.
\end{proof}

Now we define and bound each of $A_{2h}$. For any $(i,j)\in H=\{(i,j): (\Sig_u)_{ij}\neq0\}$, we must have $j\in A(i)$.
So
\begin{eqnarray*}
A_{21}&=&\frac{1}{NT^4}\sum_{(i,j)\in H,}\sum_{(m,n)\in H, (m,n)\neq (i,j),}\sum_{t=1}^T\sum_{s\neq t}\sum_{l\neq t, s}E\xi_{ijt}\xi_{mnt}E\zeta_{is}\zeta_{js}E\zeta_{ml}\zeta_{nl}\cr
&\leq&O(\frac{1}{NT})\sum_{(i,j)\in H,}\sum_{(m,n)\in H, (m,n)\neq (i,j)}|E\xi_{ijt}\xi_{mnt}|\cr
&\leq&O(\frac{1}{NT})\sum_{(i,j)\in H}(\sum_{m\in A(i)\cup A(j)}\sum_{n\in A(m)}+\sum_{m\notin A(i)\cup A(j)}\sum_{n\in A(m)})|\cov(u_{it}u_{jt}, u_{mt}u_{nt})|.
\end{eqnarray*}
The first term is $O(\frac{1}{T})$ because $|H|_0=O(N)$ and $|A(i)|_0$ is bounded uniformly by $m_N=O(1)$.  So the number of summands in $\sum_{m\in A(i)\cup A(j)}\sum_{n\in A(m)}$ is bounded.  For the second term, if $m\notin A(i)\cup A(j)$,  $n\in A(m)$ and $\cov(u_{it}u_{jt}, u_{mt}u_{nt})\neq0$, then $m\in B(i,j)$. Hence the second term is bounded by $O(\frac{1}{NT})\sum_{(i,j)\in H}\sum_{m\in B(i,j)}\sum_{n\in A(m)}|\cov(u_{it}u_{jt}, u_{mt}u_{nt})|$, which is also $O(\frac{1}{T})$ by Lemma \ref{la.6}. Hence $A_{21}=o(1)$.

Similarly, applying Lemma \ref{la.6},
$$
A_{22}=\frac{1}{NT^4}\sum_{(i,j)\in H,}\sum_{(m,n)\in H, (m,n)\neq (i,j),}\sum_{t=1}^T\sum_{s\neq t}\sum_{l\neq t, s}E\xi_{ijt}\xi_{mnt}E\zeta_{is}\zeta_{ms}E\zeta_{jl}\zeta_{nl}=o(1),
$$
which is proved in the same lines of those of $A_{21}$.

Also note three simple facts: (1) $\max_{j\leq N}|A(j)|_0=O(1)$, (2) $(m,n)\in H$ implies $n\in A(m)$, and (3) $\xi_{mms}=\xi_{nms}$.  The term $A_{23}$ is defined as
\begin{eqnarray*}
A_{23}&=&\frac{1}{NT^4}\sum_{(i,j)\in H,}\sum_{(m,n)\in H, (m,n)\neq (i,j),}\sum_{t=1}^T\sum_{s\neq t}\sum_{l\neq t, s}E\xi_{ijt}\zeta_{it}E\zeta_{js}\xi_{mns}E\zeta_{ml}\zeta_{nl}\cr
&\leq&O(\frac{1}{NT})\sum_{j=1}^N\sum_{i\in A(j)}1\sum_{(m,n)\in H, (m,n)\neq (i,j)}|E\zeta_{js}\xi_{mns}|\cr
&\leq&O(\frac{2}{NT})\sum_{j=1}^N\sum_{n\in A(j)}|E\zeta_{js}\xi_{jns}|+O(\frac{1}{NT})\sum_{j=1}^N\sum_{ m\neq j, n\neq j}|E\zeta_{js}\xi_{mns}|=a+b.
\end{eqnarray*}
Term $a=O(\frac{1}{T})$. For $b$, note that  Lemma \ref{la.5} implies that when $m,n\neq j$, $u_{ms}u_{ns}$ and $e_{js}$ are independent because of the Gaussianity. Also because $\bu_s$ and $\bff_s$ are independent, hence $\zeta_{js}$ and $\xi_{mms}$ are independent, which implies that $b=0$. Hence $A_{23}=o(1)$.

The same argument as of $A_{23}$ also implies
$$
A_{24}=\frac{1}{NT^4}\sum_{(i,j)\in H,}\sum_{(m,n)\in H, (m,n)\neq (i,j),}\sum_{t=1}^T\sum_{s\neq t}\sum_{l\neq t, s}E\xi_{ijt}\zeta_{mt}E\zeta_{is}\xi_{mns}E\zeta_{il}\zeta_{nl}=o(1)
$$
Finally, because $\sum_{(i,j)\in H}1\leq\sum_{i=1}^N\sum_{j\in A(i)}1\leq m_N\sum_{i=1}^N1$, and $m_N=O(1)$, we have
\begin{eqnarray*}
A_{25}&=&\frac{1}{NT^4}\sum_{(i,j)\in H,}\sum_{(m,n)\in H, (m,n)\neq (i,j),}\sum_{t=1}^T\sum_{s\neq t}\sum_{l\neq t, s}E\xi_{ijt}\zeta_{it}E\zeta_{is}\zeta_{ms}E\xi_{mnl}\zeta_{nl}\cr
&\leq&O(\frac{1}{NT})\sum_{(i,j)\in H,}\sum_{(m,n)\in H, (m,n)\neq (i,j)}|E\xi_{ijt}\zeta_{it}E\zeta_{is}\zeta_{ms}E\xi_{mnl}\zeta_{nl}|\cr
&\leq &O(\frac{1}{NT})\sum_{i=1}^N\sum_{m=1}^N|E\zeta_{is}\zeta_{ms}|\leq O(\frac{1}{NT})\sum_{i=1}^N\sum_{m=1}^N|(\Sig_u^{-1})_{im}|E(1-\bff_s'\bw)^2\cr
&\leq &O(\frac{N}{NT})\|\Sig_u^{-1}\|_1=o(1).
\end{eqnarray*}
In summary, $\frac{1}{T}A_2=o(1)+O(\frac{N}{T^2})=o(1)$. This completes the proof.

\subsection{Further technical lemmas for Section 4}

We cite a   lemma that will be needed throughout the proofs.

\begin{lem}\label{la.1} Under Assumption \ref{a4.1}, there is $C>0$, \\
(i) $P(\max_{i,j\leq N}|\frac{1}{T}\sum_{t=1}^Tu_{it}u_{jt}-Eu_{it}u_{jt}|>C\sqrt{\frac{\log N}{T}})\rightarrow0$.\\
(ii) $P(\max_{i\leq K, j\leq N}|\frac{1}{T}\sum_{t=1}^Tf_{it}u_{jt}|>C\sqrt{\frac{\log N}{T}})\rightarrow0$.\\
(iii) $P(\max_{j\leq N}|\frac{1}{T}\sum_{t=1}^Tu_{jt}|>C\sqrt{\frac{\log N}{T}})\rightarrow0$.
\end{lem}
\begin{proof} The proof follows  from Lemmas A.3 and B.1 in \cite{FLM11}.
 \end{proof}

\begin{lem}\label{la.1add} When the distribution of $(\bu_t, \bff_t)$ is independent of $\btheta$,   there is $C>0$, \\
(i)  $\sup_{\sbtheta\in\Theta}P(\max_{j\leq N}|\widehat\theta_j-\theta_j|>C\sqrt{\frac{\log N}{T}}|\btheta)\rightarrow0$\\
 (ii) $\sup_{\sbtheta\in\Theta}P(\max_{i,j\leq N}|\hsig_{ij}-\sigma_{ij}|>C\sqrt{\frac{\log N}{T}}|\btheta)\rightarrow0$,\\
 (iii) $\sup_{\sbtheta\in\Theta}P(\max_{i\leq N}|\hsig_{i}-\sigma_{i}|>C\sqrt{\frac{\log N}{T}}|\btheta)\rightarrow0$.

\end{lem}

\begin{proof}

Note that $\widehat\theta_j-\theta_j=\frac{1}{a_{f, T} T}\sum_{t=1}^Tu_{jt}(1-\bff_t'\bw)$. Here  $a_{f, T}=1-\bar\bff'\bw\rightarrow^p 1-E\bff_t'(E\bff_t\bff_t')^{-1}E\bff_t>0$, hence $a_{f, T}$ is bounded away from zero with probability approaching one.  Thus  by Lemma \ref{la.1}, there is $C>0$ independent of $\btheta$, such that
$$
\sup_{\sbtheta\in\Theta}P(\max_{j\leq N}|\widehat\theta_j-\theta_j|>C\sqrt{\frac{\log N}{T}}|\btheta)=P(\max_j|\frac{1}{a_{f, T} T}\sum_{t=1}^Tu_{jt}(1-\bff_t'\bw)|>C\sqrt{\frac{\log N}{T}})\rightarrow0
$$
(ii) There is $C$ independent of $\btheta$, such that the event $$
A=\{\max_{i,j}|\frac{1}{T}\sum_{t=1}^Tu_{it}u_{jt}-\sigma_{ij}|<C\sqrt{\frac{\log N}{T}},\quad \frac{1}{T}\sum_{t=1}^T\|\bff_t\|^2<C\}
$$
has probability approaching one. Also, there is $C_2$ also independent of $\btheta$ such that the event $B=\{\max_i\frac{1}{T}\sum_tu_{it}^2<C_2\}$ occurs with probability approaching one. Then
on the event $A\cap B$, by the triangular and Cauchy-Schwarz inequalities,
  $$|\widehat\sigma_{ij}-\sigma_{ij}|\leq C\sqrt{\frac{\log N}{T}}+2\max_{i}\sqrt{\frac{1}{T}\sum_t(\widehat u_{it}-u_{it})^2C_2}+\max_i\frac{1}{T}\sum_t(u_{it}-\widehat u_{it})^2.$$
It can be shown that
$$
\max_{i\leq N}\frac{1}{T}\sum_{t=1}^T(\widehat u_{it}-u_{it})^2\leq\max_i(\|\widehat\bb_i-\bb_i\|^2+(\widehat\theta_i-\theta_i)^2)(\frac{1}{T}\sum_{t=1}^T\|\bff_t\|^2+1).
$$
Note that $\widehat\bb_i-\bb_i$ and $\widehat\theta_i-\theta_i$ only depend on $(\bff_t, \bu_t)$ (independent of $\btheta$). By Lemma 3.1 of \cite{FLM11}, there is $C_3>0$ such that
$
\sup_{\bb,\sbtheta}P(\max_{i\leq N}\|\widehat\bb_i-\bb_i\|^2+(\widehat\theta_i-\theta_i)^2>C_3\frac{\log N}{T})=o(1).
$
Combining the last two displayed inequalities yields, for  $C_4=(C+1)C_3$,
$$
\sup_{\sbtheta}P(\max_{i\leq N}\frac{1}{T}\sum_{t=1}^T(\widehat u_{it}-u_{it})^2>C_4\frac{\log N}{T}|\btheta)=o(1),
$$
which yields the desired result.

 (iii): Recall $\widehat\sigma_j^2=\widehat\sigma_{jj}/a_{f, T}$, and $\sigma_j^2=\sigma_{jj}/(1-E\bff_t'(E\bff_t\bff_t')^{-1}E\bff_t)$. Moreover, $a_{f, T}$ is independent of $\btheta.$ The result  follows immediately from part (ii). \end{proof}

\begin{lem}\label{la.2} For any $\epsilon>0$, $\sup_{\sbtheta}P(\|\hSig_u^{-1}-\Sig_u^{-1}\|>\epsilon|\btheta)=o(1)$.
\end{lem}
\begin{proof}
By Lemma {\ref{la.1add}} (ii), $\sup_{\sbtheta\in\Theta}P(\max_{i,j\leq N}|\hsig_{ij}-\sigma_{ij}|>C\sqrt{\frac{\log N}{T}}|\btheta)\rightarrow1.$ By   \cite{FLM11}, on the event $\max_{i,j\leq N}|\hsig_{ij}-\sigma_{ij}|\leq C\sqrt{\frac{\log N}{T}}$, there is constant $C'$ that is independent of $\btheta$,
$
\|\widehat \Sig_u^{-1}-\Sig_u^{-1}\|\leq C'm_N(\frac{\log N}{T})^{1/2}.
$
Hence the result follows due to the sparse condition $m_N(\frac{\log N}{T})^{1/2}=o(1)$.
\end{proof}

\newpage

\normalsize

\bibliographystyle{ims}
\bibliography{liaoBib}

\begin{comment}

\end{comment}

\end{document}